\documentclass{article}

\pdfoutput=1

\usepackage[utf8]{inputenc}
\usepackage[T1]{fontenc}
\usepackage[english]{babel}

\usepackage[colorlinks]{hyperref}
\usepackage[sc]{titlesec}

\usepackage{bookmark}
\usepackage{booktabs}
\usepackage{fullpage}
\usepackage{times}

\usepackage{amsmath}
\usepackage{amssymb}
\usepackage{amsthm}
\usepackage{amsfonts}
\usepackage{mathtools}

\usepackage{graphicx}
\usepackage{tikz}
\usepackage{caption}
\usepackage{subcaption}
\usepackage{natbib}
\usepackage{url}

\usepackage{multicol}
\usepackage{array}

\usepackage{csquotes}
\usepackage{comment}
\usepackage{cancel}

\usetikzlibrary{shapes, arrows, positioning}

\newtheorem{theorem}{Theorem}
\newtheorem{lemma}{Lemma}

\DeclareMathOperator{\diag}{diag}
\DeclareMathOperator{\bdiag}{bdiag}
\DeclareMathOperator{\Tr}{Tr}
\DeclareMathOperator{\Var}{Var}
\DeclareMathOperator{\Cov}{Cov}
\DeclareMathOperator{\vecm}{vec}

\DeclareMathOperator{\logdet}{logdet}

\DeclareMathOperator{\TV}{TV}

\DeclareMathOperator{\err}{err}
\DeclareMathOperator*{\argmin}{argmin}

\DeclareMathOperator{\KLsymbol}{KL}

\newcommand{\Jmat}{J}
\newcommand{\innov}{\varepsilon}
\newcommand{\noise}{\eta}
\newcommand{\bbE}{\mathbb{E}}
\newcommand{\bbP}{\mathbb{P}}
\newcommand{\bbR}{\mathbb{R}}
\newcommand{\bbN}{\mathbb{N}}
\newcommand{\one}{\mathbf{1}}
\newcommand{\Sigmad}{\Sigma_{\mathrm{d}}}
\newcommand{\KL}[2]{\KLsymbol\left\{#1  ~\middle\Vert~  #2 \right\}}
\newcommand{\KLmax}{\KLsymbol_{\mathrm{avg}}}
\newcommand{\basis}{\mathrm{e}}
\newcommand{\thetamax}{\vartheta}
\newcommand{\tmix}{t_{\mathrm{mix}}}
\newcommand{\phantomeq}{\phantom{{}={}}}
\newcommand{\phantomleq}{\phantom{{}\leq{}}}

\newcommand{\phantomplus}{\phantom{{}+{}}}

\title{\sc Minimax Estimation of Partially-Observed Vector AutoRegressions}

\author{
 \textbf{Guillaume Dalle}\\
 CERMICS, École des Ponts\\
 Marne-la-Vallée, France\\
 guillaume.dalle@enpc.fr\\
 \and
 \textbf{Yohann De Castro}\\
 Institut Camille Jordan, École Centrale Lyon\\
 Écully, France\\
 yohann.de-castro@ec-lyon.fr\\
}

\date{April 2022}

\begin{document}

\maketitle

\begin{abstract}
High-dimensional time series are a core ingredient of the statistical modeling toolkit, for which numerous estimation methods are known.
But when observations are scarce or corrupted, the learning task becomes much harder.
The question is: how much harder?

In this paper, we study the properties of a partially-observed Vector AutoRegressive process, which is a state-space model endowed with a stochastic observation mechanism.
Our goal is to estimate its sparse transition matrix, but we only have access to a small and noisy subsample of the state components.
Interestingly, the sampling process itself is random and can exhibit temporal correlations, a feature shared by many realistic data acquisition scenarios.

We start by describing an estimator based on the Yule-Walker equation and the Dantzig selector, and we give an upper bound on its non-asymptotic error.
Then, we provide a matching minimax lower bound, thus proving near-optimality of our estimator.
The convergence rate we obtain sheds light on the role of several key parameters such as the sampling ratio, the amount of noise and the number of non-zero coefficients in the transition matrix.
These theoretical findings are commented and illustrated by numerical experiments on simulated data.
\end{abstract}

\section{Introduction}

Time series provide a natural representation for periodic measurements of a stochastic process.
In particular, those defined by linear Gaussian recursions may be the most widely used and the easiest to study.
Well-known examples include the AutoRegressive (AR) process and its multivariate counterpart, the Vector AutoRegressive (VAR) process.

Industrial applications of these models encounter two main challenges.
First, they often involve signals in high dimension, which means sparsity assumptions play an important role.
Second, the variables of interest are rarely measured exactly or entirely.
Indeed, physical constraints such as the cost of sensors can make it impossible to capture every component of the system's state at all times.
It is therefore natural to ask: \emph{how much harder does high-dimensional learning become when one only observes a fraction of the relevant values?}

\subsection{Context of the Study}

To answer this question, we study a state-space model where the state~$X_t \in \bbR^D$ follows a VAR process of order~$1$ over a period of length~$T$.
Since the dimension~$D$ of~$X_t$ is high, we assume that its transition matrix~$\theta \in \bbR^{D \times D}$ is~$s$-sparse (there are no more than~$s$ non-zero coefficients in each row).
However, we do not observe the state itself: our observations~$Y_t$ only involve the subset of components~$X_{t, d}$ for which~$\pi_{t, d} = 1$, where~$\pi_t$ is a vector of Bernoulli variables.
To make matters worse, this subset is corrupted with noise, which leads to the following generative procedure:
\begin{equation} \label{eq:model}
  X_t = \theta X_{t-1} + \mathcal{N}(0, \sigma^2 I) \qquad \pi_{t, d} \sim \mathcal{B}(p)  \qquad Y_t = \diag(\pi_t) X_t + \mathcal{N}(0, \omega^2 I).
\end{equation}
When we write~$\pi_{t, d} \sim \mathcal{B}(p)$, we mean that the marginals of the sampling variables are identical, which requires that every state component be sampled with equal probability~$p$.
However, we reject the standard independence assumption in favor of temporal dependencies between the Bernoulli variables~$\pi_{t, d}$ (see Section~\ref{sec:trains} for a practical justification).

\medskip

To shed light on the properties of our model, we start by constructing a sparse estimator for~$\theta$, whose non-asymptotic error we upper bound.
We complement this finding with a lower bound on the minimax error that does not depend on the choice of estimator.
Upper and lower bound match in most regards, which proves their optimality.
A rough summary of our analysis is that the best possible estimator~$\widehat{\theta}$ satisfies
\begin{equation} \label{eq:convergence_rate}
  \lVert \widehat{\theta} - \theta \rVert_{\infty} \lesssim \left(1 + \frac{\omega^2}{\sigma^2} \right) \frac{s}{p\sqrt{T}}
\end{equation}
with high probability.
We observe that the error does not depend on the state dimension~$D$, but only on the sparsity~$s$ of the transition matrix.
As expected, it decreases linearly as~$p$ grows, since more information becomes available.
Lastly, it is a function of~$\omega^2 / \sigma^2$, which means that precise recovery of~$\theta$ is only possible when the noise is not too much larger than the signal.

Novel features of our work include the first proof of a minimax lower bound in this setting (to the best of our knowledge), the investigation of temporal correlations within the sampling process, the combination of discrete and continuous concentration inequalities to obtain error estimates, as well as detailed numerical experiments on simulated data.

\subsection{Example of Application} \label{sec:trains}

Our study was inspired by concrete questions related to delay propagation on railway networks, which came up during a collaboration with a leading railway company.
When external factors (weather, passenger behavior, mechanical failures) trigger a primary delay, resource conflicts between trains can amplify the initial incident and send ripple effects through the whole network.
Understanding and predicting this propagation phenomenon is a crucial task for traffic management and robust scheduling.

To model it, we construct a network graph~$\mathcal{G} = (\mathcal{V}, \mathcal{E})$ linking the railway stations, and we assume the existence of a hidden congestion variable~$X_{t, d}$ that lives on the edges~$d \in \mathcal{E}$.
This congestion evolves according to a VAR process, whose transition matrix~$\theta$ represents pairwise interactions between edges.
The sparsity structure of~$\theta$ expresses the local nature of delay propagation, which is why it is closely related to the adjacency structure of~$\mathcal{G}$.
Indeed, between times~$t$ and~$t+1$, edges are expected to transmit congestion to their close neighbors, and not to regions of the network that are very far away.

Unfortunately for us,~$X_t$ is never observed directly.
The only information we have is collected by the trains whenever they cross an edge of the network.
The crossing time of a train is influenced by the congestion, but also by other individual factors: in this sense, our observations~$Y_t$ are a noisy version of the underlying process~$X_t$.
Furthermore, the observations are limited in size: the dimension of~$X_t$ is the number of edges~$D = |\mathcal{E}|$, while the dimension of~$Y_t$ is linked to the number of trains on the timetable and the length of their respective journeys.
We can thus define a random variable~$\pi_{t, d}$ equal to~$1$ if a train crosses edge~$d$ between~$t$ and~$t+1$, and~$0$ otherwise.
A more realistic model would account for the possibility of multiple trains crossing an edge in the same time step, especially if the discretization interval is large.
However, our binary assumption greatly simplifies exposition without betraying the qualitative behavior of the system.
Crucially, this sampling mechanism exhibits temporal correlations: periods of dense traffic are likely to be followed by dense traffic, which means that the sequence of sampling variables~$\pi_{t, d}$ is not independently distributed.

We recognize the framework of Equation~\eqref{eq:model}, and can therefore apply the theoretical result of Equation~\eqref{eq:convergence_rate}.
This error quantification provides useful insight on the estimation of~$\theta$, which is essential to help railway operatives dimension their data sets or evaluate prediction uncertainty.

\subsection{Related Works} \label{sec:related_works}

The theory of VAR processes has been known for a long time: the book of \citet{lutkepohlNewIntroductionMultiple2005} provides a detailed account.
If we have full and noiseless observations of the process~$X_t$, we can use conditional Least Squares to estimate~$\theta$ by minimizing the quadratic error~$\sum_t \lVert X_t - \theta X_{t-1} \rVert_2^2$.
This is equivalent to solving the Yule-Walker equation~$\Gamma_h = \theta \Gamma_{h-1}$, where we replace the autocovariance matrix~$\Gamma_h = \Cov[X_{t+h}, X_t]$ with its empirical counterpart~$\widehat{\Gamma}_h$.
In the case of Gaussian innovations, both approaches coincide with the Maximum Likelihood Estimator~(MLE).

Neither of these methods was initially designed for missing or noisy data.
Luckily, statistical estimation with imprecise measurements has been thoroughly studied \citep{buonaccorsiMeasurementErrorModels2010}.
The same goes for incomplete data sets ; an extensive survey was recently published by \citet{littleStatisticalAnalysisMissing2019}.
According to their terminology, our work deals with data that is missing completely at random (MCAR), which means that the projection~$\pi_t$ is independent from the underlying process~$X_t$.
We also assume to know the distribution of the missingness indicators~$1-\pi_{t, d}$, which is not necessarily true for other applications (e.g. clinical trials).

A principled approach to deal with missing data would require extending the MLE to partially-observed time series, also known as state-space models \citep{cappeInferenceHiddenMarkov2006}.
Most of the time, exact or approximate inference is achievable using some version of the Kalman filter \citep{kalmanNewApproachLinear1960} or particle methods \citep{doucetSequentialMonteCarlo2000}, whereas parameter estimation typically involves the Expectation-Maximization (EM) algorithm \citep{shumwayApproachTimeSeries1982}.
Unfortunately, the EM algorithm is hard to analyze explicitly in terms of statistical error, which is why other methods are sometimes preferred in theoretical studies.
In particular, plug-in methods that use covariance estimates within the Yule-Walker equation have been quite popular in the machine learning community.

\medskip

In this line of work, the core challenge is the high dimension~$D$ of the VAR process~$X_t$.
To address it, many authors use sparsity-inducing penalties as a way to reduce data requirements and computational workload.
In the last ten years, the LASSO \citep{tibshiraniRegressionShrinkageSelection1996} has been increasingly applied to random designs exhibiting correlations or missing data.
This trend started with the seminal work of \citet{lohHighdimensionalRegressionNoisy2012}, and numerous other papers followed \citep[see for example][]{basuRegularizedEstimationSparse2015, kockOracleInequalitiesHigh2015, melnykEstimatingStructuredVector2016, jalaliMissingDataSparse2018}.

As an alternative to the LASSO, the Dantzig selector \citep{candesDantzigSelectorStatistical2007} enforces sparsity in the objective and data fidelity in the constraints.
While the LASSO requires solving a Quadratic Program (QP), for instance with proximal methods, the Dantzig selector gives rise to a Linear Program (LP) which can be parallelized across dimensions.
\citet{hanDirectEstimationHigh2015} studied its application to VAR estimation, obtaining finite-sample error bounds with very natural hypotheses.
A little later, \citet{raoEstimationAutoregressiveProcesses2017} extended these results to the more general scenario in which a hidden VAR process is randomly sampled or projected, and then corrupted with noise.
This last work is quite similar to ours, but we think that the proof they present to control the non-asymptotic error is incomplete at best\footnote{
  Indeed, the combination of discrete and Gaussian concentration inequalities as performed on page 2 (middle of right column) of the supplementary material for \citet{raoEstimationAutoregressiveProcesses2017} glosses over the fact that~$L_F$ is itself a random variable.
  As we will discover during our own proof, this introduces an additional difficulty and forces us to use a more complex Gaussian concentration result (Lemma~\ref{lem:conditional_hanson_wright}).
  See \url{https://web.stanford.edu/~milind/papers/system_id_icassp_proof.pdf} for the supplementary material in question.
}.

Another salient feature of our paper is the search for a minimax lower bound, which allows us to prove the optimality of our convergence rates.
To the best of our knowledge, this was only attempted once for partially-observed VAR processes.
\citet{raoFundamentalEstimationLimits2017} presented a lower bound on the minimax error in a setting very similar to ours, but their result is less generic in several regards.
Indeed, we account for the possibility of temporal correlations within sampling, as well as observation noise.
Moreover, unlike the one proposed by \citet{raoFundamentalEstimationLimits2017}, our proof focuses on geometric properties and doesn't make use of the admissible set of transition matrices until the very end.
This makes it easy to handle many different types of structured transitions without additional work: sparse, Toeplitz, banded, etc.

Finally, the error bounds we obtain are backed up by detailed numerical experiments on simulated data, which allow us to visualize the influence of every parameter of interest.

\subsection{Outline of the Paper}

In Section~\ref{sec:model_and_estimator}, we define the generative procedure behind the partially-observed VAR process, and we present a sparse estimator of the transition matrix.
We then state both of our theoretical results in Section~\ref{sec:lower_upper_bound}: an upper-bound on the error of our specific estimator, complemented by a minimax lower bound on the error of any estimation algorithm.
Section~\ref{sec:experiments} contains numerical experiments demonstrating the impact of various parameters, which lead to the conclusion in Section~\ref{sec:conclusion}.

Appendix~\ref{sec:convergence_proof} is dedicated to proving the convergence rate of the sparse estimator, while Appendix~\ref{sec:lower_bound_proof} contains the derivation of the minimax lower bound.
A number of useful results from linear algebra and probability are presented in Appendix~\ref{sec:lemmas} to make the paper as self-contained as possible.
Most of them are well-known, some were obtained or adapted specifically for our proof.
Appendix~\ref{sec:glossary} contains a summary of the main notations and symbols.

\section{The Partially-Observed VAR Process and its Sparse Estimator} \label{sec:model_and_estimator}

Before stating our theoretical results, we introduce our statistical model and the estimator we use.

\subsection{Model Definition} \label{sec:model_definition}

The model we study was described approximately in the introduction.
We now fill the gaps of the generative procedure it relies on.

\paragraph{The underlying state}~$X = (X_t)_{t \in [T]}$ follows a stationary VAR process of order~$1$.
This process has dimension~$D$ and the following recursive definition:
\begin{equation} \label{eq:x_model}
  X_{t} = \theta X_{t-1} + \innov_t \qquad \text{with} \qquad \innov_t \sim \mathcal{N}(0, \Sigma).
\end{equation}
Here~$\theta \in \bbR^{D \times D}$ is the transition matrix and~$\Sigma \in \bbR^{D \times D}$ is the covariance matrix of the innovations (in the introduction, we assumed~$\Sigma = \sigma^2 I$).

To ensure stationarity of the VAR process, we must constrain the spectral radius of~$\theta$ to satisfy~$\rho(\theta) < 1$.
Throughout the paper, we actually make the following (slightly stronger) assumption on the spectral norm of~$\theta$: there exists~$\thetamax \in (0, 1)$ such that for all the values of~$\theta$ we consider,~$\lVert \theta \rVert_2 \leq \thetamax < 1$.
Furthermore, we only study row-sparse transition matrices, having at most~$s$ nonzero coefficients in each row.
In other words, we restrict our choice of parameters to
\begin{equation} \label{eq:bigtheta_definition}
  \theta \in \Theta_s \quad \text{where} \quad \Theta_s = \{\theta \in \bbR^{D \times D}: \lVert \theta \rVert_2 \leq \thetamax < 1 \quad \text{and} \quad \forall i, \lVert \theta_{i, \cdot} \rVert_0 \leq s\}.
\end{equation}
We denote by~$\sigma_{\min}^2 = \lambda_{\min}(\Sigma)$ and~$\sigma_{\max}^2 = \lambda_{\max}(\Sigma)$ the minimum and maximum eigenvalues of the covariance matrix~$\Sigma$.

\paragraph{The observation mechanism} we chose implies that we do not have direct access to the latent process~$X_t$.
To construct the observations~$Y_t$, we sample a subset of state components according to the binary vectors~$\pi_t$.
Then, independent Gaussian noise with variance~$\omega^2$ is added to these selected components, and we observe the result.
If we denote by~$\Pi_t = \diag(\pi_t)$ the diagonal projection matrix, we have
\begin{equation} \label{eq:y_model}
  Y_t = \Pi_t X_t + \noise_t \qquad \text{with} \qquad \noise_t \sim \mathcal{N}(0, \omega^2 I).
\end{equation}
An essential hypothesis we make is the mutual independence between our three sources of randomness: the innovations~$\innov_t$, the projections~$\pi_t$ and the observation noise~$\noise_t$.

\medskip

A major feature of the present work is the non-deterministic selection of observed state components, that is, the fact that~$\pi_t$ is a random sequence of Bernoulli vectors following a known distribution.
In order to sum up the amount of information available using one parameter~$p \in (0, 1)$, we want this distribution to satisfy the following condition: each component~$X_{t,d}$ of the latent state must be sampled with the same marginal probability~$p = \bbP(\pi_{t, d} = 1)$.

On the other hand, we also want to introduce temporal dependencies between the projections.
The simplest way to achieve that is through a Markovian hypothesis: independently along each dimension~$d$, time indices~$t$ are selected for observation according to a binary-valued Markov chain with transition matrix~$\mathcal{T} = \begin{psmallmatrix} 1 - a & a \\ b & 1 - b \end{psmallmatrix}$.
Its coefficients are chosen to make the chain stationary with invariant measure~$(\tfrac{b}{a+b}, \tfrac{a}{a+b}) = (1-p, p)$.
Note that when~$a = 1-b = p$, this reduces to independent sampling of each component with probability~$p$.
We also assume there exists a universal constant~$\chi$ such that~$0 < \chi \leq a, b \leq 1-\chi < 1$: this means that the chain does not transition too fast nor too slowly.

\medskip

Our data set is built from~$N$ independent realizations of this process.
For the sake of simplicity however, we will prove all convergence theorems in the case~$N = 1$: extending those results to the general case simply amounts to replacing~$T$ with~$NT$ in the resulting error bounds.

\subsection{Sparse Estimator for the Transition Matrix}  \label{sec:estimator_def}

We now introduce the estimation method chosen for this problem.

\paragraph{The transition estimator} presented here is a straightforward generalization of the one used by \citet{raoEstimationAutoregressiveProcesses2017}.
The lag-$h$ covariance matrix of the VAR process~$X_t$ is given by the Yule-Walker recursion (see Lemma~\ref{lem:x_covariance}):
\begin{equation} \label{eq:yule_walker}
  \Gamma_h(\theta) =  \Cov_\theta[X_{t+h}, X_{t}] = \theta \Gamma_{h-1}(\theta) = \theta^h \Gamma_0(\theta)
\end{equation}
We can use it to define a simple two-step procedure:
\begin{enumerate}
  \item For a given~$h_0$, build estimators~$\widehat{\Gamma}_{h_0}$ and~$\widehat{\Gamma}_{h_0 + 1}$ of the covariances~$\Gamma_{h_0}$ and~$\Gamma_{h_0+1}$.
  \item Use them to approximate the transition matrix by inverting Equation~\eqref{eq:yule_walker}.
\end{enumerate}
A simple inversion technique uses the Moore-Penrose pseudoinverse (just in case~$\widehat{\Gamma}_{h_0}$ is singular):
\begin{equation} \label{eq:theta_estimator_dense}
  \widehat{\theta}^{\text{dense}} = \widehat{\Gamma}_{h_0+1} \widehat{\Gamma}_{h_0}^\dagger.
\end{equation}
The problem with this procedure is that is does not guarantee sparsity of~$\widehat{\theta}$.
To obtain a sparse result, we follow \citet{hanDirectEstimationHigh2015} and cast Equation~\eqref{eq:yule_walker} as a soft constraint enforcing proximity between~$\widehat{\Gamma}_{h_0+1}$ and~$\widehat{\theta} \widehat{\Gamma}_{h_0}$.
This amounts to solving the following constrained optimization problem:
\begin{equation} \label{eq:theta_estimator}
  \widehat{\theta} \in \argmin_{M \in \bbR^{D \times D}} \lVert \vecm(M) \rVert_1 \quad \text{subject to} \quad \lVert M \widehat{\Gamma}_{h_0} - \widehat{\Gamma}_{h_0+1} \rVert_{\max} \leq \lambda.
\end{equation}
Here~$\lVert \vecm(\cdot) \rVert_1$ denotes the sum of the absolute values of all the coefficients of a matrix, while~$\lVert \cdot \rVert_{\max}$ is the maximum of these absolute values.
Given that both of these norms are piecewise linear, the problem of Equation~\eqref{eq:theta_estimator} can be reformulated as an LP. It can even be decomposed along each dimension, which allows for an efficient and parallel solution procedure.
The only thing left to do is decide how to estimate the covariance matrices~$\Gamma_h$.

\paragraph{The covariance estimator} we use is a variant of the empirical covariance.
Since~$Y_t = \Pi_t X_t + \noise_t$ where~$\noise_t$ is zero-mean, a natural proxy for~$X_t$ is obtained by inverting the sampling operator:~$\widehat{X}_t = \Pi_t^\dagger Y_t$.
It would therefore seem logical to build an estimator of~$\Gamma_h$ by plugging this proxy into the empirical covariance between~$X_{t+h}$ and~$X_t$.
However, in order for this idea to work, we must make two small adjustments.

To account for the random sampling, the plug-in empirical covariance must be scaled elementwise by a matrix~$S(h) = \bbE[\pi_{t+h} \pi_t']$.
Intuitively, since~$\widehat{X}_{t+h} \widehat{X}_t'$ has a fraction~$p^2$ of nonzero coefficients, we need to divide it by something close to~$p^2$ to get an unbiased covariance estimator.
Furthermore, to account for the observation noise, we must incorporate an additive correction~$- \omega^2 I$.
This correction becomes unnecessary for~$h \geq 1$ since the observation noise~$\noise_t$ is independent across time.

In conclusion, we obtain the following covariance estimator:
\begin{equation} \label{eq:gamma_estimator}
  \widehat{\Gamma}_h =  \frac{1}{S(h)} \odot \frac{1}{T-h} \sum_{t=1}^{T-h} \left(\Pi_{t+h}^\dagger Y_{t+h}\right) \left(\Pi_t^\dagger Y_{t}\right)' - \one_{\{h = 0\}} \omega^2 I.
\end{equation}
The coefficients of the scaling matrix~$S(h)$ are computed in Lemma~\ref{lem:proj_moments}.

\section{Lower and Upper Bound on the Estimation Error} \label{sec:lower_upper_bound}

We now have the necessary background to formulate our theoretical results.
In all the following statements (and their proofs), the letter~$c$ denotes a universal positive constant, which may change from one line to the next but never depends on any varying problem parameters.
More specifically, statements involving it should always be understood as \enquote{there exists~$c > 0$ such that}...

\subsection{Main Theorems}

We start by bounding the non-asymptotic error of the estimator we just introduced.

\begin{theorem}[Error upper bound] \label{thm:convergence_rate_theta}
  Consider the partially-observed VAR model defined in Section~\ref{sec:model_definition}.
  We use the estimator~$\widehat{\theta}$ of Section~\ref{sec:estimator_def} with~$h_0 = 0$, and we suppose that~$T$ is \enquote{large enough}, as specified by Equations~\eqref{eq:T_constraint_convergence1} and~\eqref{eq:T_constraint_convergence2}.
  Let us define
  \begin{equation} \label{eq:gammau_sparse}
    \gamma_u(\theta) =
    \frac{\lVert \theta \rVert_{\infty} + 1}{(1-\lVert \theta \rVert_2)^2}
    \frac{\sigma_{\max}^2 + \omega^2}{\lVert \Gamma_0(\theta)^{-1} \rVert_1^{-1}} \qquad \text{and} \qquad q_u = \min\{p, 1-b\} \leq p.
  \end{equation}
  Then there is a value of~$\lambda$ such that the following upper bound holds with probability at least~$1 - \delta$:
  \begin{equation} \label{eq:theta_sparse_error}
    \lVert \widehat{\theta} - \theta \rVert_{\infty}
    \leq c \frac{\gamma_u(\theta) s}{\sqrt{T p q_u}} \sqrt{\log(D/\delta)}.
  \end{equation}
\end{theorem}

\begin{proof}
  The argument combines discrete and continuous concentration inequalities, to account for both the Bernoulli sampling and the Gaussian noise.
  More precisely, we exploit a recent Chernoff bound that applies to non-reversible Markov chains, and we plug it into a conditional version of the Hanson-Wright inequality that we derived specifically for our purposes.
  See Appendix~\ref{sec:convergence_proof} for more details.
\end{proof}

We now move on to a minimax lower bound which is estimator-independent, and quantifies the intrinsic difficulty of our statistical problem.
The term minimax means that we study the probability of making an error of magnitude~$\zeta$, when we pick the best possible estimator~$\widehat{\theta}$ and nature replies by choosing the worst possible parameter~$\theta$:
\begin{equation} \label{eq:minimax_def}
  \mathfrak{P}(\zeta) = \inf_{\widehat{\theta}} \sup_{\theta \in \Theta_s} \bbP_{\theta} \left[ \lVert \widehat{\theta} - \theta \rVert_{\infty} \geq \zeta \right].
\end{equation}
More precisely, we want to find a threshold~$\zeta$ such that the probability of exceeding it is non-negligible, for instance~$\mathfrak{P}(\zeta) \geq \frac{1}{2}$.
The evolution of this threshold will tell us how the error behaves with respect to the various problem parameters.

\begin{theorem}[Error lower bound] \label{thm:lower_bound_sparse}
  Consider the partially-observed VAR model defined in Section~\ref{sec:model_definition}.
  We suppose that~$T$ is \enquote{large enough}, as specified by Equations~\eqref{eq:Tlarge_minimax_1} and~\eqref{eq:Tlarge_minimax_2}.
  Let us define
  \begin{equation} \label{eq:gammal_q}
    \gamma_\ell = (1-\thetamax)^{3/2} \frac{\sigma_{\min}^2 + \omega^2}{\sigma_{\max}^2}
    \qquad \text{and} \qquad
    q_\ell = \max\{1-b, 2p-(1-b)\} \geq p.
  \end{equation}
  Then the following minimax lower bound holds:
  \begin{equation}
    \inf_{\widehat{\theta}} \sup_{\theta \in \Theta_s} \bbP_{\theta} \left[ \lVert \widehat{\theta} - \theta \rVert_{\infty} \geq c \frac{\gamma_\ell s}{\sqrt{T p q_\ell}} \right] \geq \frac{1}{2}.
  \end{equation}
\end{theorem}

\begin{proof}
  The argument is based on an information-theoretical result known as Fano's inequality.
  To apply it, we need to upper bound the Kullback-Leibler (KL) divergence between the distributions~$\bbP_{\theta_0}(\Pi, Y)$ and~$\bbP_{\theta_1}(\Pi, Y)$, where~$\theta_0$ and~$\theta_1$ are sufficiently far apart.
  See Appendix~\ref{sec:lower_bound_proof} for more details.
\end{proof}

\subsection{Influence of the Problem Parameters}

Let us now compare the error bounds of Theorems~\ref{thm:convergence_rate_theta} and~\ref{thm:lower_bound_sparse}.
Our first remark is that~$s$ and~$T$ play exactly the same roles in both bounds (up to a logarithmic factor), which shows that the dependency of the error in~$s/\sqrt{T}$ is optimal.

\paragraph{The sampling parameters} appear as~$1/\sqrt{p q_u}$ in the upper bound, whereas the lower bound scales as~$1/\sqrt{p q_\ell}$ instead.
This means that we have not proven the optimality of either bound with respect to~$p$ or~$b$.
However, it is reassuring to note that there is no conflict between them since~$q_\ell \geq p \geq q_u$.
Furthermore, when~$a=1-b=p$ (that is, when Markov sampling boils down to independent sampling), both bounds simplify into the~$1/p$ dependency we would expect (since~$q_u = q_\ell = p$).
So in the case of independent sampling,~$1/p$ is indeed the optimal rate.

\paragraph{The~$\ell_2$ norm of the transition matrix} plays opposite roles on each side.
In the lower bound,~$1-\thetamax = 1 - \max_{\theta \in \Theta_s} \lVert \theta \rVert_2$ appears in the numerator, whereas in the upper bound,~$1 - \lVert \theta \rVert_2$ appears in the denominator.
It is likely that these dependencies are suboptimal, but at least they are compatible with one another: as~$\lVert \theta \rVert_2 \to 1$, that is, as the VAR process becomes unstable, the lower bound tends to~$0$ and the upper bound to~$+\infty$.
This is a reflection of the fact that our proofs make heavy use of the distance between~$\theta$ and the unit sphere, which means they become meaningless when~$\theta$ gets too large.

\paragraph{The variances~$\Sigma$ and~$\omega^2$} are involved in~$\gamma_\ell$ for the lower bound, and in~$\gamma_u(\theta)$ for the upper bound.
In both cases, the ratio~$\gamma$ tells us whether the underlying process is large enough to be detected among the noise.
Roughly speaking, the magnitude of~$X_t$ is related to the spectrum of~$\Sigma$, while the magnitude of~$Y_t$ is related to the spectrum of~$\Sigma + \omega^2 I$.
If the latter is significantly larger than the former, recovering~$X_t$ (and thus~$\theta$) is a hopeless endeavor.

To simplify the comparison, let us assume in this discussion that~$\Sigma = \sigma^2 I$, and that~$\theta$ commutes with its transpose.
Then we have~$\lVert \Gamma_0^{-1}(\theta) \rVert_1^{-1} = \big \lVert \big(\sigma^2(I - \theta \theta')^{-1} \big)^{-1} \big\rVert_1^{-1} = \sigma^2 \lVert I - \theta \theta' \rVert_1^{-1}$, and we can give a simpler expression of~$\gamma_\ell$ and~$\gamma_u(\theta)$:
\begin{align*}
  \gamma_u(\theta)
  = \frac{(\lVert \theta' \rVert_1 + 1) \lVert I - \theta \theta' \rVert_1}{(1-\thetamax)^2}
  \frac{\sigma^2 + \omega^2}{\sigma^2} &  &
  \gamma_\ell
  = (1-\thetamax)^{3/2} \frac{\sigma^2 + \omega^2}{\sigma^2}.
\end{align*}
We recognize the same dependency in both bounds, namely~$\gamma \propto 1 + \frac{\sigma^2}{\omega^2}$.
Lemma~\ref{lem:signal_to_noise} gives a heuristic argument linking this functional form to the asymptotic behavior of the MLE.

\subsection{Extension to VAR Processes of Higher Order}

Although our results only apply to state-space models based on an underlying VAR process of order~$1$, we could try to extend them to the more general case of VAR($K$) processes. Just for this Section, suppose~$X_t$ is no longer given by Equation~\eqref{eq:x_model}, but instead satisfies:
\begin{align*}
  X_t = \theta_1 X_{t-1} + \theta_2 X_{t-2} + ... + \theta_K X_{t-K} + \innov_t.
\end{align*}
Then we can represent this as a VAR(1) process using augmented variables \citep{lutkepohlNewIntroductionMultiple2005}.
Indeed, observe that defining $\widetilde{X}_t = \begin{pmatrix}
    X_t & X_{t-1} & \cdots & X_{t-K+1}
  \end{pmatrix}'$ and~$\widetilde{\innov}_t = \begin{pmatrix}
    \innov_t & 0 & \cdots & 0
  \end{pmatrix}'$ yields
\begin{align*}
  \widetilde{X}_t = \widetilde{\theta} \widetilde{X}_{t-1} + \widetilde\innov_t \qquad \text{with} \qquad
  \widetilde{\theta} = \begin{bmatrix}
    \theta_1 & \theta_2 & \cdots & \theta_{K-1} & \theta_K \\
    I_D      & 0        & \cdots & 0            & 0        \\
    0        & I_D      &        & 0            & 0        \\
    \vdots   &          & \ddots & \vdots       & \vdots   \\
    0        & 0        & \cdots & I_D          & 0
  \end{bmatrix}.
\end{align*}
Unfortunately, by this reasoning, the Markov sampling mechanism that generates~$\Pi_t$ gives rise to a new distribution for~$\widetilde{\Pi}_t$ which is no longer part of the same family.
Indeed, the augmented sampling process~$\widetilde{\Pi}_t$ is still Markovian but with a memory of size~$K$ instead of~$1$.
Therefore, the adaptation would require new arguments and we leave it for future work.

\section{Numerical Illustrations} \label{sec:experiments}

We now illustrate our results on simulated data.
All experiments were performed on a Dell Precision 5530 mobile workstation with Intel Core i7-8850H CPU (2.60GHz~$\times$ 12) and 31 GiB of RAM, running under Ubuntu 20.04.
Our code was written in Julia \citep{bezansonJuliaFreshApproach2017},
linear optimization problems were modeled using JuMP \citep{dunningJuMPModelingLanguage2017} and solved with the COIN-OR Clp solver \citep{forrestCoinorClpRelease2022}.
The reproducible Pluto notebook used to generate all the plots will be made available on GitHub as soon as the review procedure is complete and anonymity is no longer required.

\subsection{Data Generation}

Simulating a partially-observed VAR process with known transition matrix~$\theta$ allows us to compute the estimation error~$\lVert \widehat{\theta} - \theta \rVert_{\infty}$ and study the influence of parameters such as~$T$,~$D$,~$s$,~$p$,~$\omega$, etc.
Real values for~$\theta$ were drawn using independent standard Gaussian distributions for each coefficient, and then normalized to satisfy~$\lVert \theta \rVert_2 = \thetamax = \frac12$.
To simplify comparison with the theoretical bounds, we used a diagonal innovation covariance~$\Sigma = \sigma^2 I$ and set the sampling parameters to~$a = 1-b = p$, which amounts to independent sampling (except for the experiment that focuses specifically on the influence of~$b$).
When not mentioned explicitly, all other parameters are equal to their default values given below (we assume~$\omega$ is known):
\begin{equation*}
  T = 10000 \qquad D = 5 \qquad \sigma = 1.0 \qquad \omega = 0.1 \qquad p = 1.0.
\end{equation*}
Most of the simulations are run in a dense estimation scenario.
For those that require the sparse procedure, selecting a good regularization parameter~$\lambda$ is paramount: indeed, Theorem~\ref{thm:convergence_rate_theta} is only valid for a specific value of~$\lambda$ (which is not known in practice, but we can hope to approximate this near-optimal choice).

A standard way to tune~$\lambda$ would be cross-validation.
However, evaluating a choice of~$\lambda$ (and the resulting estimate~$\widehat{\theta}$) requires inferring the hidden state sequence~$X_t$ from the observations~$Y_t$.
If the projection matrices~$\Pi_t$ were deterministic, the inference could be performed with Kalman filtering \citep{kalmanNewApproachLinear1960}, but since they are stochastic, the distribution of~$(X, Y)$ is no longer jointly Gaussian and the justification behind the Kalman filter breaks down.
Finding an appropriate inference method in our setting will be the topic of future studies.

In the meantime, to tune~$\lambda$, we suppose that the sparsity level of the real transition matrix~$\theta$ is known.
We then use this target sparsity~$\widehat{s}$ to guide a dichotomy search on~$\lambda$, until we find a transition matrix estimate~$\widehat{\theta}$ whose row sparsity level is sufficiently close to~$\widehat{s}$.

\subsection{Results}

\begin{figure}[htbp]
  \begin{subfigure}{0.47\textwidth}
    \includegraphics[width=\textwidth]{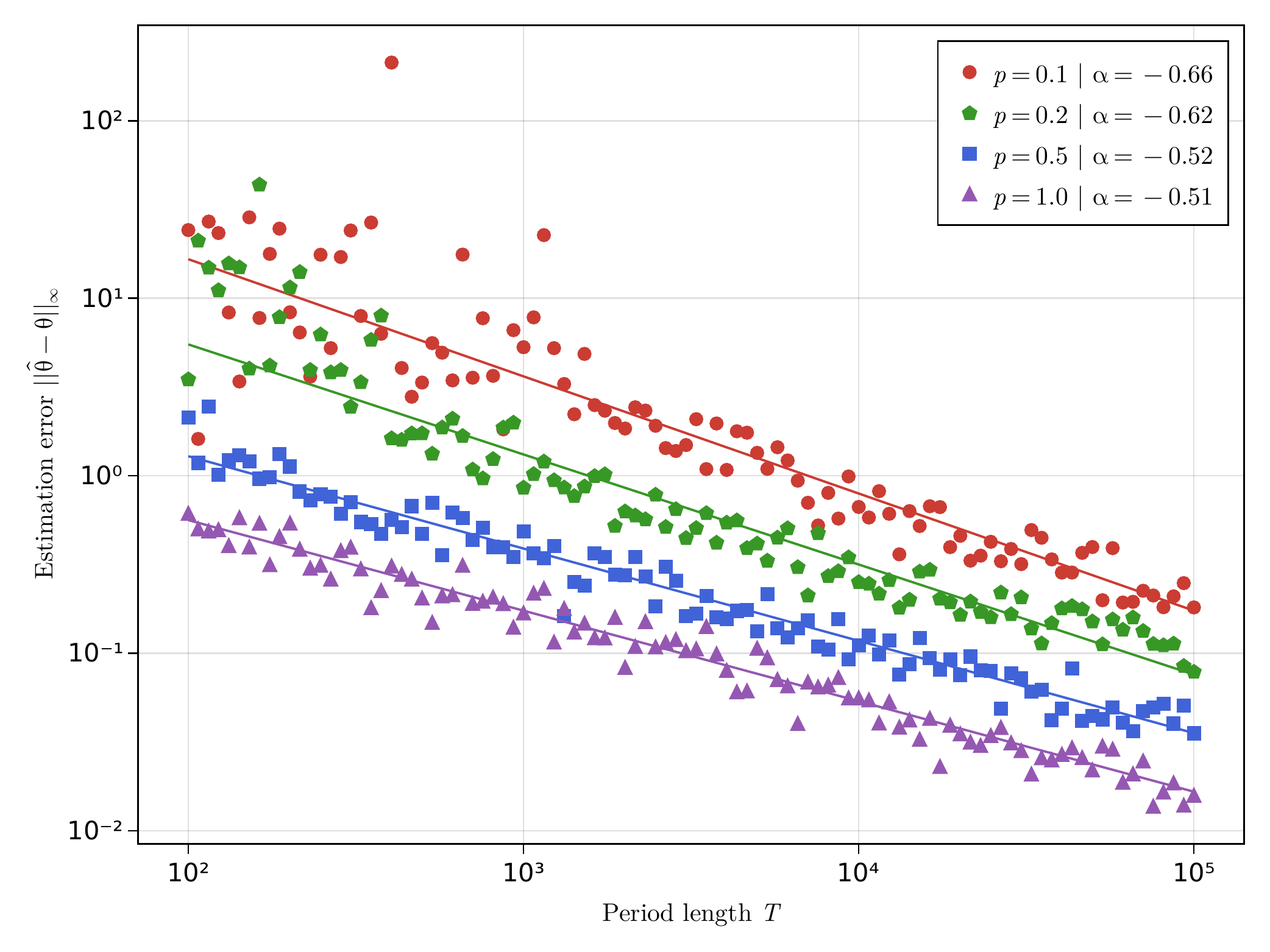}
    \subcaption{Influence of~$T$}
    \label{fig:influence_T}
  \end{subfigure}
  \hfill
  \begin{subfigure}{0.47\textwidth}
    \includegraphics[width=\textwidth]{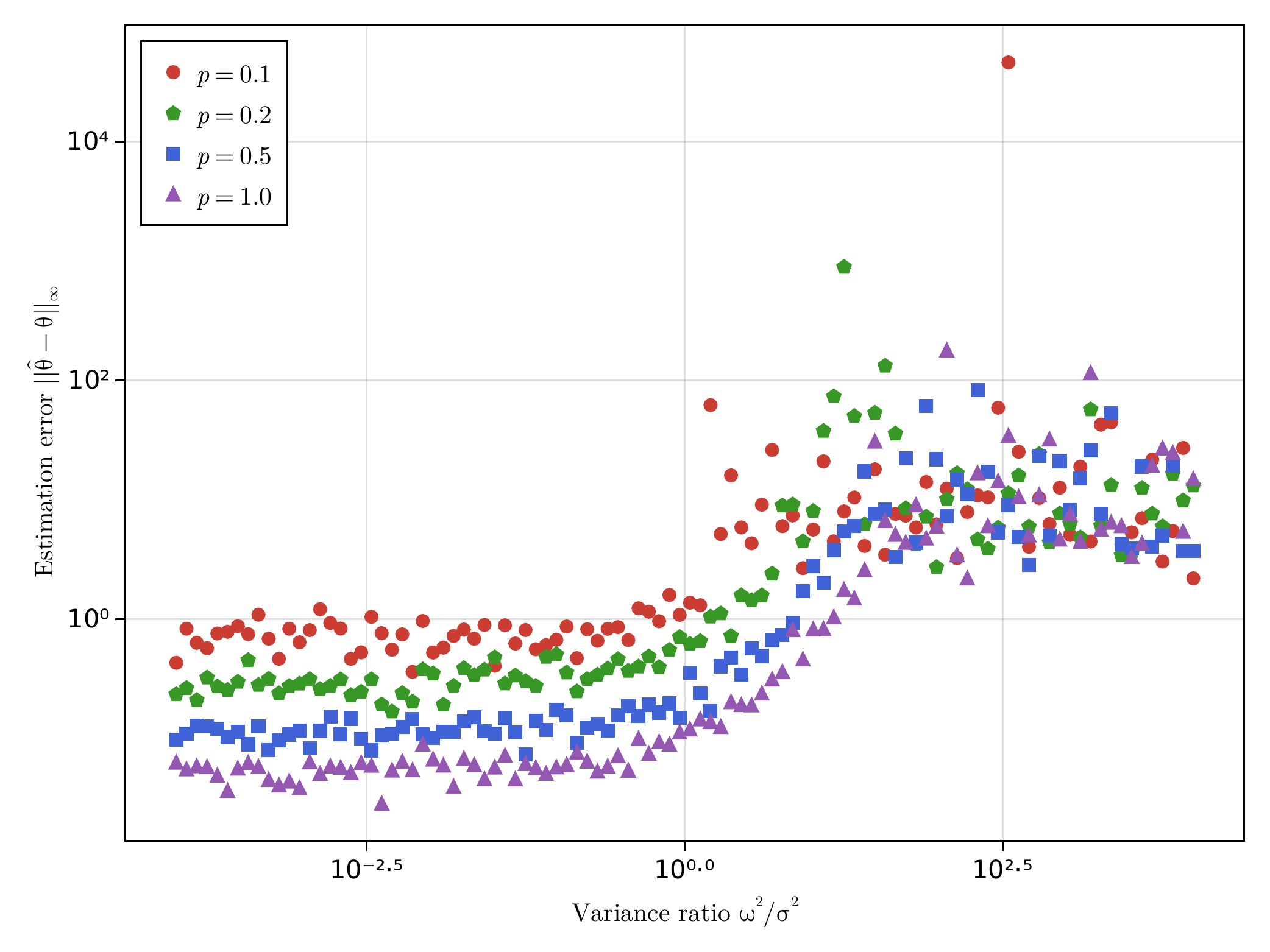}
    \subcaption{Influence of $\omega$}
    \label{fig:influence_omega}
  \end{subfigure}

  \begin{subfigure}{0.47\textwidth}
    \includegraphics[width=\textwidth]{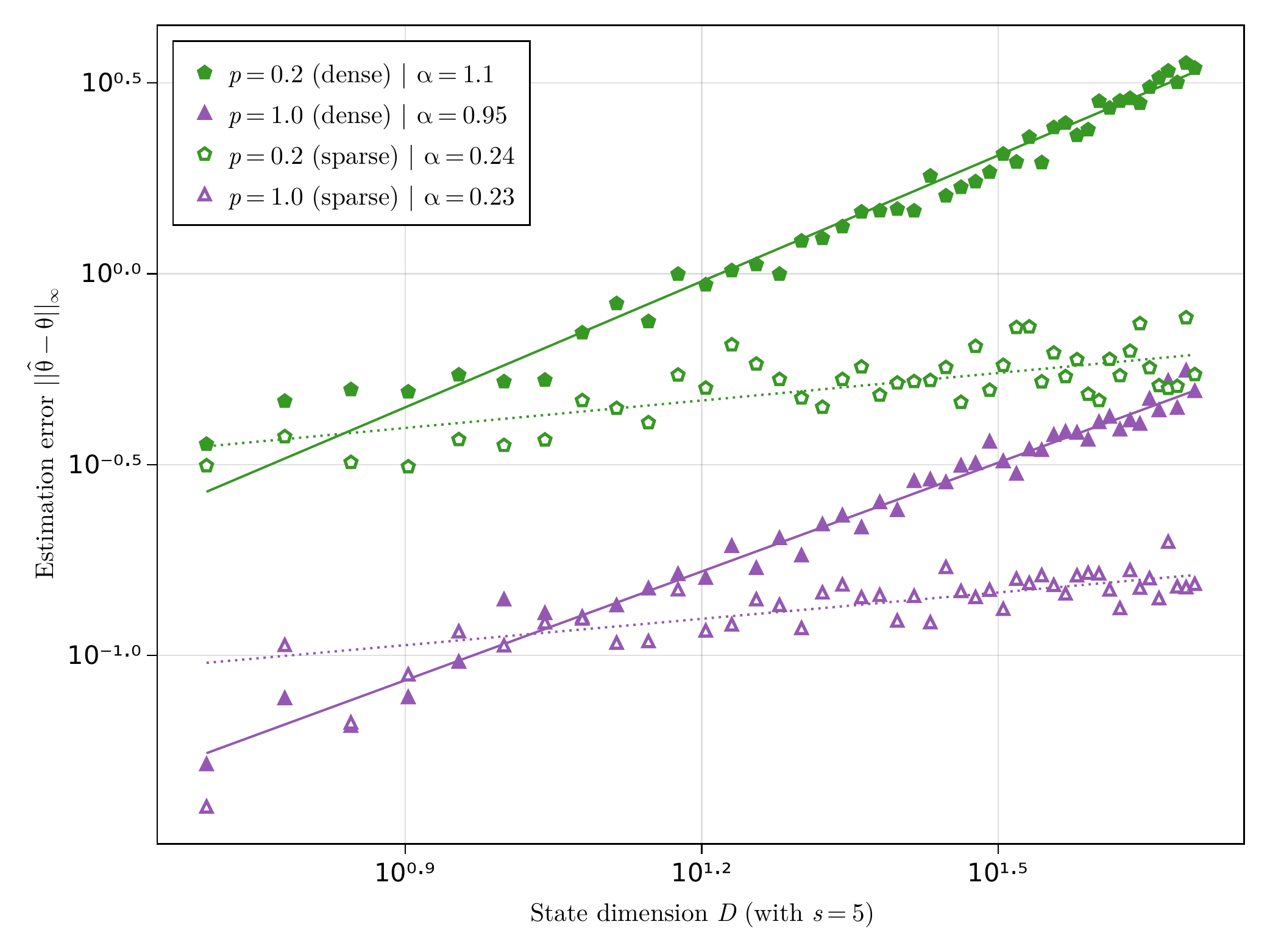}
    \subcaption{Influence of~$D$ with fixed~$s$}
    \label{fig:influence_D_fixed_s}
  \end{subfigure}
  \hfill
  \begin{subfigure}{0.47\textwidth}
    \includegraphics[width=\textwidth]{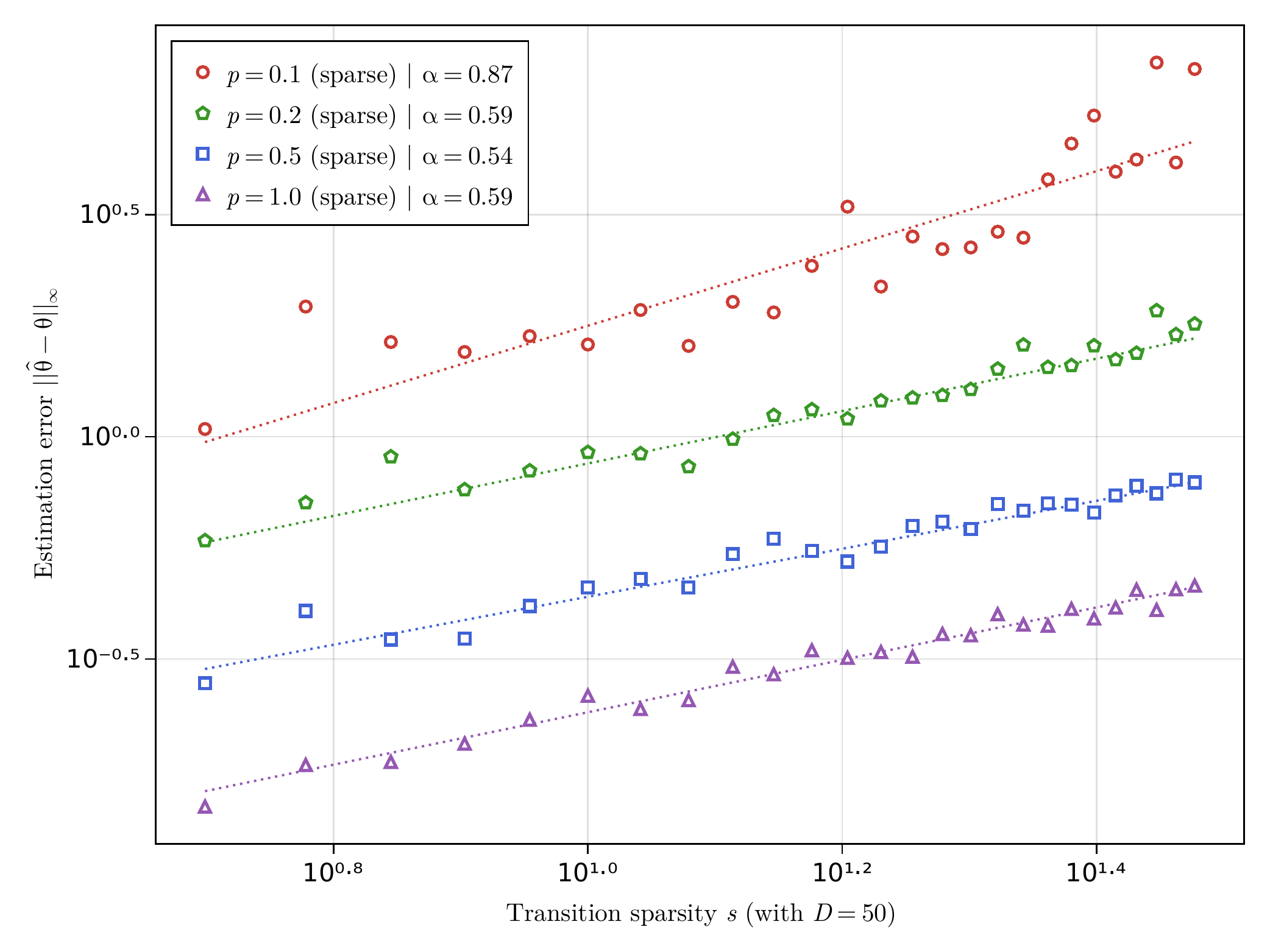}
    \subcaption{Influence of~$s$ with fixed~$D$}
    \label{fig:influence_s_fixed_D}
  \end{subfigure}

  \begin{subfigure}{0.47\textwidth}
    \includegraphics[width=\textwidth]{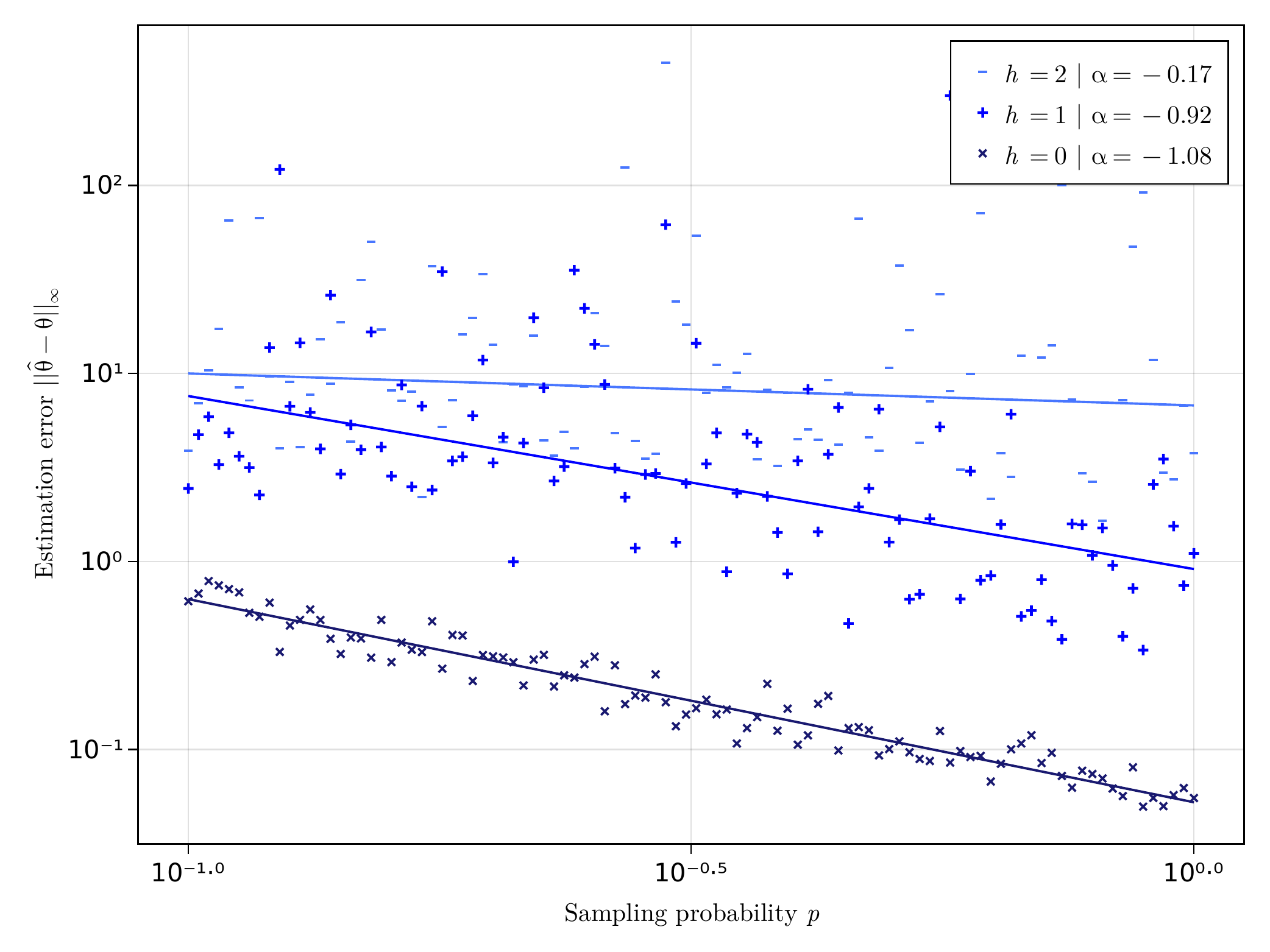}
    \subcaption{Influence of~$p$ and~$h_0$}
    \label{fig:influence_h0}
  \end{subfigure}
  \hfill
  \begin{subfigure}{0.47\textwidth}
    \includegraphics[width=\textwidth]{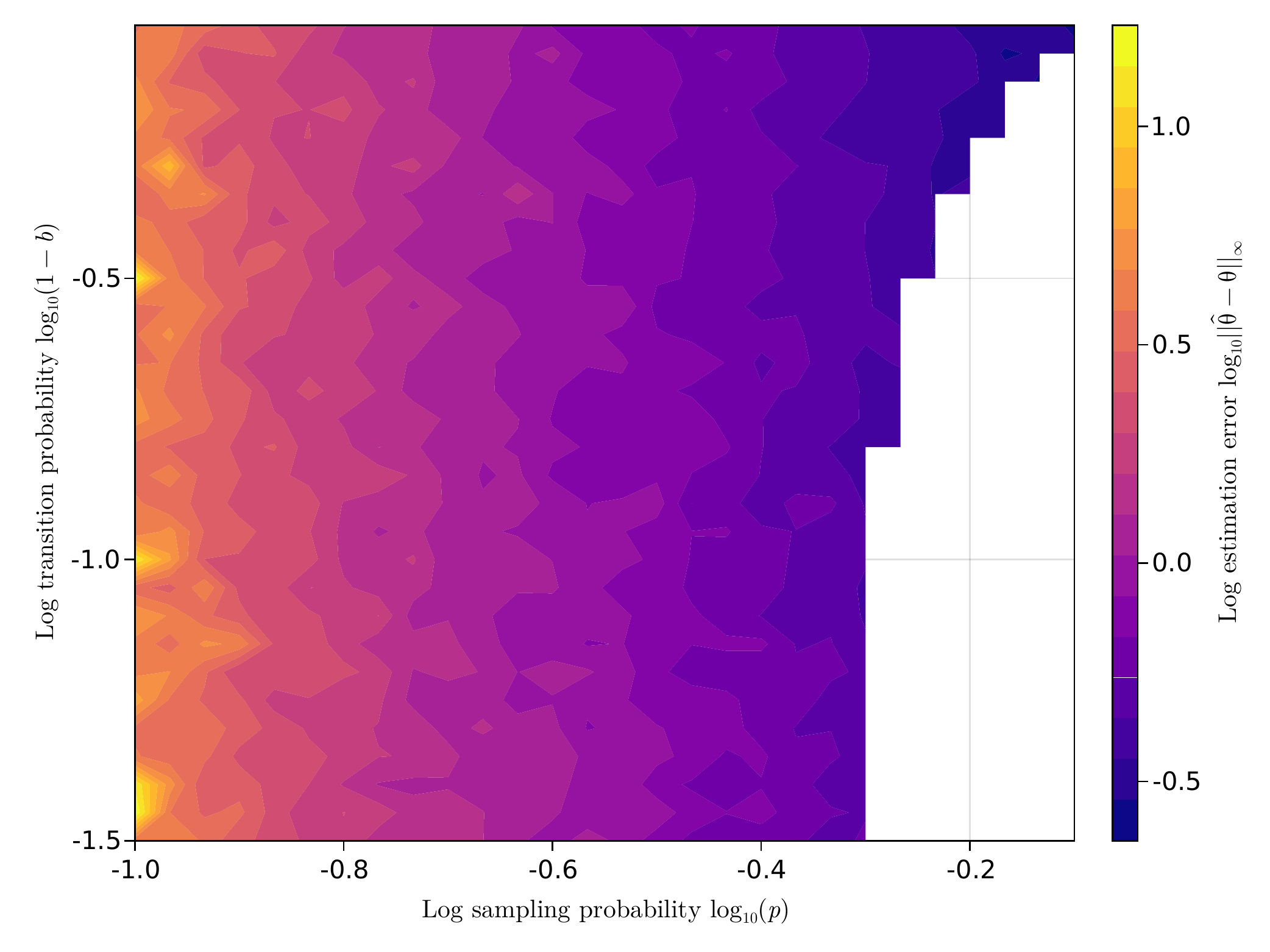}
    \subcaption{Influence of~$(p, b)$}
    \label{fig:influence_pb}
  \end{subfigure}

  \caption{Impact of model parameters on the estimation error}
  \label{fig:simulations}
\end{figure}

The main results are presented on Figure~\ref{fig:simulations}.
With the exception of~\ref{fig:influence_pb}, all plots have the estimation error~$\lVert \widehat{\theta} - \theta \rVert_{\infty}$ on their~$y$-axis, and some parameter of interest on their~$x$-axis.
The axes are displayed with logarithmic scaling, in order to highlight the exponent of the dependencies.
Each point corresponds to one run of the algorithm, aimed at estimating a single random value of~$\theta$.
When a straight line is added to a scatter plot, it is the result of a Theil-Sen regression \citep{senEstimatesRegressionCoefficient1968} applied to the points of the same color: its slope is denoted by~$\alpha$ in the legend.

\medskip

Figure~\ref{fig:influence_T} confirms that the error decreases as~$1/\sqrt{T}$.
This is only true because the sampling probability~$p$ remains constant.
If instead we had a limited observation budget but an increasing temporal precision, we would have~$p \propto 1/T$, in which case the error would increase as~$\sqrt{T}$ instead of decreasing.

Figure~\ref{fig:influence_omega} exhibits three clearly identifiable regimes with respect to the noise variance.
In the first one, corresponding to~$\omega / \sigma \ll 1$, the error remains small and constant.
Then, the error increases when~$\omega / \sigma \simeq 1$.
In the third phase, corresponding to~$\omega / \sigma \gg 1$, the error remains high and volatile.
This is consistent with the theoretical dependency in~$1 + \omega^2 / \sigma^2$.

Figure~\ref{fig:influence_D_fixed_s} compares the respective benefits of sparse and dense estimation by increasing the ambient dimension~$D$ while keeping the true sparsity level~$s$ constant.
The error for~$\widehat{\theta}^{\text{dense}}$ scales linearly with~$D$, while its sparse counterpart~$\widehat{\theta}$ achieves a much slower error growth.
As a side note, the fact that the error grows with~$D$ is not surprising. Indeed, we measure it with the~$\ell_{\infty}$ operator norm, which scales with the dimension of the matrix.

Figure~\ref{fig:influence_s_fixed_D} takes the opposite perspective by increasing the number of nonzero coefficients in a space of fixed dimension.
In this case, the theory predicts that the error should scale linearly with~$s$, but the slope we observe is below~$1$.
Our interpretation is that the function~$\gamma_u(\theta)$ also depends on the sparsity level in complicated ways through~$\theta$, especially since the real values are renormalized to satisfy~$\lVert \theta \rVert_2 = \frac12$.

Figure~\ref{fig:influence_h0} shows that the error evolves as~$1/p$, which is consistent with our upper bound.
It is also informative w.r.t. the choice of~$h_0$.
Choosing~$h_0 = 0$ means we need to know~$\omega$ to perform estimation.
If this parameter is unknown, we can choose~$h_0 \geq 1$, which leads to a much higher variance of the estimator (this is not visible in our results since we wrote the proof in the case where~$h_0 = 0$).
An alternate solution would be to keep~$h_0 = 0$ and plug in a guess such as~$\omega = 0$, effectively trading lower variance for a higher bias.

Figure~\ref{fig:influence_pb} takes a closer look at the role of the Markov sampling parameter~$b$.
The white region corresponds to values of~$b$ for which there is no~$a$ such that~$p = a / (a+b)$.
On this logarithmic heatmap, we see regularly-spaced and nearly vertical contour lines, which is consistent with a convergence rate of~$1/p$ that does not depend on~$b$.
We conjecture that~$1/p$ is the true order of magnitude for the optimal error, and that the dependencies~$1/\sqrt{pq_u}$ and~$1/\sqrt{p q_\ell}$ from our Theorems could be refined and brought together with a more careful theoretical analysis.

\section{Conclusion and Perspectives} \label{sec:conclusion}

In this paper, we studied a partially-observed VAR process, whose latent state components are randomly projected and corrupted with noise before being observed.
The temporal correlations within the sampling process are a novel feature, and combining both sources of randomness (discrete and continuous) required the use of tailored probabilistic methods.
We provided upper and lower bounds for the optimal estimation error on the transition matrix, and found that these bounds roughly match.
Our analysis, supported by empirical results, sheds light on the intrinsic difficulty of such statistical problems, which arise naturally when analyzing several types of network processes.

However, our study leaves many questions open for future work.
On the theoretical side, bridging the gap between our bounds will probably require more sophisticated tools to capture the precise behavior of Markov sampling.
Going from the uniform case, where the sampling probability equals~$p$ everywhere, to more realistic heterogeneous settings, is also a worthy avenue to explore.
On the practical side, this linear Gaussian model may not perform well when applied to real prediction problems.
Finding ways to enhance it will be necessary if we want to gather insights on complex high-dimensional dynamics, especially for graph-structured data.

\section*{Acknowledgements}

The authors would like to thank their colleague Axel Parmentier for his collaboration and careful proofreading. Clément Mantoux, Éloïse Berthier, Maxime Godin and Pierre Marion (by alphabetical order of first names) provided more support and advice than can be described in such a constrained space. We also thank Emeline Luirard for her help with a critical Lemma, and Mathieu Besançon for his last-minute look at the draft.

We are very grateful to the SNCF, especially its departments DGEX Solutions (SNCF Réseau) and Transilien (SNCF Voyageurs) for providing us with the inspiration behind this work.

\bibliography{POVAR}

\appendix

\section{Proof of the Estimator's Convergence Rate} \label{sec:convergence_proof}

Here we present the detailed proof of Theorem~\ref{thm:convergence_rate_theta}.

\subsection{Overview}

The main steps of the argument are the following:
\begin{enumerate}
  \item Prove the Yule-Walker Equation~\eqref{eq:yule_walker} and deduce an expression for the covariance matrix of~$X$ (Lemmas~\ref{lem:x_covariance} and~\ref{lem:norm_gamma}).
  \item Justify the formula of Equation~\eqref{eq:gamma_estimator} for~$\widehat{\Gamma}_h$ by showing that it defines an unbiased estimator of~$\Gamma_h$ (Lemmas~\ref{lem:gamma_estimator_unbiased} and~\ref{lem:proj_moments}).
  \item Fixing two indices~$d_1$ and~$d_2$, rewrite~$(\widehat{\Gamma}_h - \Gamma_h)_{d_1, d_2}$ using quadratic forms~$g_a' \Psi_a' L \Psi_b g_b$ of standard Gaussian vectors (Lemma~\ref{lem:proba_split}).
  \item Control the deviation of the matrix~$L$ using discrete concentration inequalities (Lemmas~\ref{lem:L_norm_reformulation},~\ref{lem:L_spectral_bound},~\ref{lem:proj_concentration}, ~\ref{lem:L_frob_bound} and~\ref{lem:L_trace}).
  \item Apply a conditional version of the Hanson-Wright inequality (Lemma~\ref{lem:conditional_hanson_wright}) to the quadratic forms~$g_a' \Psi_a' L \Psi_b g_b$ (Lemma~\ref{lem:hanson_wright_application}).
  \item Obtain a high-probability control on~$\lVert \widehat{\Gamma}_h - \Gamma_h \rVert_{\max}$ with a union bound (Lemma~\ref{lem:convergence_rate_gamma}).
  \item Deduce the error of~$\widehat{\theta}$ from the error of~$\widehat{\Gamma}_{h_0}$ and~$\widehat{\Gamma}_{h_0+1}$ by drawing inspiration from \citet{hanDirectEstimationHigh2015} (Lemmas~\ref{lem:feasibility_theta} and ~\ref{lem:error_max_norm}).
\end{enumerate}

\subsection{Covariance Matrices}

The Yule-Walker equation is a direct consequence of the VAR recursion, as can be seen from this Lemma.

\begin{lemma}[VAR covariance matrices] \label{lem:x_covariance}
  The autocovariance matrices of the stationary VAR process defined by Equation~\eqref{eq:x_model} have the following expressions:
  \begin{align*}
    \Gamma_0(\theta) & = \Cov_\theta[X_t] = \sum_{k=0}^{\infty}{\theta^k \Sigma \theta'^k} \\
    \Gamma_h(\theta) & =  \Cov_\theta[X_{t+h}, X_{t}] = \theta^h \Gamma_0(\theta).
  \end{align*}
\end{lemma}

\begin{proof}
  We start by noting that according to Equation~\eqref{eq:x_model}, the stacked vector~$X = (X_t)_{t \in [T]}$ follows a~$TD$-dimensional centered multivariate Gaussian distribution.
  The covariance matrix of~$X_t$ can be deduced from the recursion:
  \begin{equation*}
    \Gamma_0(\theta) = \Cov_\theta[X_t] = \Cov_\theta[\theta X_{t-1} + \innov_t] = \theta \Cov_\theta[X_{t-1}] \theta' + \Sigma = \theta \Gamma_0(\theta) \theta' + \Sigma.
  \end{equation*}
  There is a unique stationary solution:
  \begin{equation*}
    \Gamma_0(\theta) = \sum_{k=0}^{\infty}{\theta^k \Sigma \theta'^k}.
  \end{equation*}
  The covariance matrix between~$X_{t + h}$ and~$X_t$ is obtained similarly:
  \begin{align*}
    \Gamma_h(\theta) & = \Cov_\theta[X_{t + h}, X_t] = \bbE[X_{t + h} X_t']  = \bbE[(\theta X_{t + h-1} + \innov_{t + h}) X_t'] \\
                     & = \theta \Cov_\theta[X_{t + h-1}, X_t] = \theta^h \Cov_\theta[X_t, X_t] = \theta^h \Gamma_0(\theta).
  \end{align*}
  And~$\Cov_\theta[X_t, X_{t + h}] = \Cov_\theta[X_{t + h}, X_t]'$.
  In other words, we just proved that
  \begin{equation*}
    \Cov_\theta [X] =    \begin{bmatrix}
      \Gamma_0(\theta)              & \Gamma_0(\theta) \theta'^1 & \Gamma_0(\theta) \theta'^2 & \cdots & \Gamma_0(\theta) \theta'^{T-1} \\
      \theta^1 \Gamma_0(\theta)     & \Gamma_0(\theta)           & \Gamma_0(\theta) \theta'^1 &        &                                \\
      \theta^2 \Gamma_0(\theta)     & \theta^1 \Gamma_0(\theta)  & \Gamma_0(\theta)           &        &                                \\
      \vdots                        &                            &                            & \ddots &                                \\
      \theta^{T-1} \Gamma_0(\theta) &                            &                            &        & \Gamma_0(\theta)
    \end{bmatrix}
  \end{equation*}
\end{proof}

The following result will come in handy later.

\begin{lemma}[Norm of~$\Gamma_0(\theta)$] \label{lem:norm_gamma}
  The covariance matrix~$\Gamma_0(\theta)$ satisfies
  \begin{equation*}
    \lVert \Gamma_0(\theta) \rVert_2 \leq \frac{\sigma_{\max}^2}{1 - \thetamax^2}
  \end{equation*}
\end{lemma}

\begin{proof}
  By Lemma~\ref{lem:x_covariance},
  \begin{equation*}
    \lVert \Gamma_0(\theta) \rVert_2 \leq \sum_{k=0}^{\infty} \lVert \theta^k \Sigma \theta'^k \rVert_2 \leq \sum_{k=0}^{\infty} \lVert \theta \rVert_2^k \lVert \Sigma \rVert_2 \lVert \theta \rVert_2^k = \frac{\lVert \Sigma \rVert_2}{1 - \lVert \theta \rVert_2^2} \leq \frac{\sigma_{\max}^2}{1 - \thetamax^2}.
  \end{equation*}
\end{proof}

\subsection{Construction of the Covariance Estimator}

Now we justify the construction of our covariance estimator.
Let~$h_0 = 0$: for most of the proof, we fix a lag value~$h \in \{h_0, h_0+1\} = \{0, 1\}$.

\begin{lemma}[Bias of the covariance estimator] \label{lem:gamma_estimator_unbiased}
  The estimator~$\widehat{\Gamma}_h$ given by Equation~\eqref{eq:gamma_estimator} for the covariance matrix~$\Gamma_h$ is unbiased.
\end{lemma}

\begin{proof}
  First, let us remember that since~$\Pi_t = \diag(\pi_t)$ is diagonal and binary, we also have~$\Pi_t^{\dagger} = \Pi_t' = \Pi_t$.
  By Equation~\eqref{eq:y_model},
  \begin{equation} \label{eq:estimator_decomp}
    \begin{aligned}
      (\Pi_{t+h}^\dagger Y_{t+h})(\Pi_t^\dagger Y_t)'
       & = \Pi_{t+h}^\dagger (\Pi_{t+h} X_{t+h} + \noise_{t+h})(X_t' \Pi_t' + \noise_t') \Pi_t^\dagger{}'                              \\
       & = \diag(\pi_{t+h}) \left( X_{t+h} X_t' + X_{t+h} \noise_t' + \noise_{t+h} X_t' + \noise_{t+h} \noise_t' \right) \diag(\pi_t).
    \end{aligned}
  \end{equation}
  Taking the conditional expectation and removing the cross-product terms (by independence of~$X$ and~$\Pi$), we get:
  \begin{equation*}
    \bbE[(\Pi_{t+h}^\dagger Y_{t+h})(\Pi_t^\dagger Y_t)'  |  \Pi] = \diag(\pi_{t+h}) \left( \bbE[X_{t+h} X_t'] + \bbE[\noise_{t+h} \noise_t'] \right) \diag(\pi_t).
  \end{equation*}
  Since~$\bbE \left[ X_{t+h} X_t' \right] = \Gamma_h$ and~$\bbE[\noise_{t+h} \noise_t] = \one_{\{h=0\}} \omega^2 I$, we are left with:
  \begin{equation*}
    \bbE[(\Pi_{t+h}^\dagger Y_{t+h})(\Pi_t^\dagger Y_t)'  |  \Pi]
    = (\pi_{t+h} \pi_t') \odot \Gamma_h + \one_{\{h=0\}} \omega^2 \diag(\pi_t).
  \end{equation*}
  where~$\odot$ is the elementwise Hadamard product.
  We now take the expectation w.r.t.~$\Pi$:
  \begin{equation*}
    \bbE[(\Pi_{t+h}^\dagger Y_{t+h})(\Pi_t^\dagger Y_t)']
    = \bbE[\pi_{t+h} \pi_t'] \odot \Gamma_h + \one_{\{h=0\}} \omega^2 \bbE[\diag(\pi_t)].
  \end{equation*}
  Dividing elementwise by the scaling matrix~$S(h) = \bbE[\pi_{t+h} \pi_t']$, we get
  \begin{align*}
    \bbE\left[ \frac{1}{\bbE[\pi_{t+h} \pi_t']} \odot (\Pi_{t+h}^\dagger Y_{t+h})(\Pi_t^\dagger Y_t)'\right]
     & = \Gamma_h + \one_{\{h=0\}} \omega^2 \bbE[\diag(\pi_t)] \odot \frac{1}{\bbE[\pi_t \pi_t']} \\
     & = \Gamma_h + \one_{\{h=0\}} \omega^2 \diag \left(\frac{\bbE [\pi_t]}{\bbE[\pi_t^2]}\right) \\
     & = \Gamma_h + \one_{\{h=0\}} \omega^2 I
  \end{align*}
  which shows that our estimator
  \begin{equation*}
    \widehat{\Gamma}_h = \frac{1}{T-h} \sum_{t=1}^{T-h} \frac{1}{\bbE[\pi_{t+h} \pi_t']} \odot  (\Pi_{t+h}^\dagger Y_{t+h})(\Pi_t^\dagger Y_t)' - \one_{\{h=0\}} \omega^2 I
  \end{equation*}
  is unbiased.
\end{proof}

Note that since the process~$(\Pi_t)$ is stationary, the coefficients of~$S(h)$ do not depend on~$t$.
They are computed in the next Lemma.

\begin{lemma} \label{lem:proj_moments}
  The second-order moments of~$\pi$ are given by
  \begin{equation*}
    S(h)_{d_1, d_2} = \bbE \left[\pi_{t+h,d_1} \pi_{t,d_2}\right] = \begin{cases}
      p^2                    & \text{if~$d_1 \neq d_2$}             \\
      p                      & \text{if~$d_1 = d_2$ and~$h = 0$}    \\
      p^2 + p(1-p) (1-a-b)^h & \text{if~$d_1 = d_2$ and~$h \geq 1$}
    \end{cases}
  \end{equation*}
  In particular, every coefficient of the scaling matrix~$S(h)$ is lower-bounded by
  \begin{equation*}
    \min_{d_1, d_2, h} S(h)_{d_1, d_2} = \min \{p^2, p(1-b)\} = p q_u \quad \text{where} \quad q_u = \min \{p, 1-b\}.
  \end{equation*}
\end{lemma}

\begin{proof}
  Let~$i = (t+h, d_1)$ and~$j = (t, d_2)$ be two indices in~$[T] \times [D]$.
  We have~$\bbE [\pi_i] = \bbE [\pi_i^2] = p$.
  If~$d_1 \neq d_2$, then the variables~$\pi_i$ and~$\pi_j$ belong to independent Markov chains, and thus~$\bbE[\pi_i \pi_j] = p^2$.
  Otherwise, we have~$i = (t+h, d)$ and~$j = (t, d)$, which means these two variables are part of the same Markov chain. Stationarity yields
  \begin{equation*}
    \bbE[\pi_i \pi_j] = \bbP(\pi_{t,d} = 1) \times \bbP(\pi_{t+h,d} = 1  |  \pi_{t,d} = 1) =  p (\mathcal{T}^h)_{11}.
  \end{equation*}
  When diagonalizing the transition matrix~$\mathcal{T}$, we see that the bottom-right coefficient of~$\mathcal{T}^h$ is
  \begin{equation*}
    (\mathcal{T}^h)_{11} = \frac{a + b(1 - a - b)^h}{a + b} = p + (1-p) (1 - a - b)^h.
  \end{equation*}
  Plugging this in, we get
  \begin{equation*}
    \bbE[\pi_i \pi_j] = p^2 + p(1-p)(1-a-b)^h.
  \end{equation*}
  Among all the possible values of~$S(h)_{d_1, d_2}$, the smallest one is~$p^2$ if~$1-a-b \geq 0$, and~$p^2 + p(1-p)(1-a-b)$ otherwise.
  But since
  \begin{align*}
    p + (1-p)(1-a-b) & = \frac{a}{a+b} + \frac{b}{a+b}(1-a-b)                       \\
                     & = \frac{a + b - ab - b^2}{a+b} = \frac{a(1-b) + b(1-b)}{a+b} \\
                     & = 1-b,
  \end{align*}
  we conclude
  \begin{equation*}
    \min_{d_1, d_2, h} S(h)_{d_1, d_2} = \min \{p^2, p^2 + p(1-p)(1-a-b)\} = \min \{p^2, p(1-b)\}.
  \end{equation*}
\end{proof}

\subsection{Gaussian Concentration, Episode 1} \label{sec:gaussian_concentration_ep1}

From now on, we will study the concentration of~$\widehat{\Gamma}_h$, coefficient by coefficient.
Let us fix two indices~$d_1$ and~$d_2$: our goal is to control the deviation of~$(\widehat{\Gamma}_h)_{d_1, d_2}$ around its mean.

\begin{lemma}[Deviation of~$(\widehat{\Gamma}_h)_{d_1, d_2}$] \label{lem:proba_split}
  The deviation probability for~$(\widehat{\Gamma}_h)_{d_1, d_2}$ can be decomposed as follows:
  \begin{align*}
    \bbP(|(\widehat{\Gamma}_h - \Gamma_h)_{d_1, d_2}| \geq u)
     & \leq \phantomplus \bbP\left(|g_{\innov}' \Psi_{\innov}' L \Psi_{\innov} g_{\innov} - \bbE\left[g_{\innov}' \Psi_{\innov}' L \Psi_{\innov} g_{\innov}\right]| \geq u/4 \right) \\
     & \phantomleq + \bbP\left(|g_\noise' \Psi_{\noise}' L \Psi_{\innov} g_{\innov} - \bbE\left[g_\noise' \Psi_{\noise}' L \Psi_{\innov} g_{\innov}\right]| \geq u/4 \right)         \\
     & \phantomleq + \bbP\left(|g_{\innov}' \Psi_{\innov}' L \Psi_{\noise} g_\noise - \bbE\left[g_{\innov}' \Psi_{\innov}' L \Psi_{\noise} g_\noise\right]| \geq u/4 \right)         \\
     & \phantomleq + \bbP\left(|g_\noise' \Psi_{\noise}' L \Psi_{\noise} g_\noise - \bbE\left[g_\noise' \Psi_{\noise}' L \Psi_{\noise} g_\noise\right]| \geq u/4 \right)
  \end{align*}
  where the random matrix~$L$ is defined in Equation~\eqref{eq:L_def},~$\Psi_{\innov}$ and~$\Psi_{\noise}$ are defined in Equation~\eqref{eq:psi_definition}, and~$g_{\innov}$ and~$g_{\noise}$ are standard Gaussian vectors.
\end{lemma}

\begin{proof}
  We denote by~$\basis_d$ the basis vector filled with zeros except for a~$1$ in position~$d$.
  By Equation~\eqref{eq:gamma_estimator},
  \begin{align*}
    (\widehat{\Gamma}_h + \one_{\{h=0\}} \omega^2 I)_{d_1, d_2}
     & = \frac{1}{T-h} \sum_{t=1}^{T-h} \left( \frac{1}{S(h)} \odot (\Pi_{t+h}^\dagger Y_{t+h})(\Pi_t^\dagger Y_t)' \right)_{d_1, d_2}                        \\
     & = \frac{1}{T-h} \sum_{t=1}^{T-h} \frac{1}{S(h)_{d_1, d_2}} \basis_{d_1}' (\Pi_{t+h}^\dagger Y_{t+h}) (\Pi_t^\dagger Y_t)' \basis_{d_2}                 \\
     & = \frac{1}{T-h} \sum_{t=1}^{T-h} \Tr \left[\frac{\basis_{d_2} \basis_{d_1}'}{S(h)_{d_1, d_2}} (\Pi_{t+h}^\dagger Y_{t+h}) (\Pi_t^\dagger Y_t)' \right]
  \end{align*}
  Equation~\eqref{eq:estimator_decomp} allows us to rewrite~$(\Pi_{t+h}^\dagger Y_{t+h}) (\Pi_t^\dagger Y_t)'$:
  \begin{align*}
    (\widehat{\Gamma}_h + \one_{\{h=0\}} \omega^2 I)_{d_1, d_2}
     & = \frac{1}{T-h} \sum_{t=1}^{T-h} X_t' \diag(\pi_t) \frac{\basis_{d_2} \basis_{d_1}'}{S(h)_{d_1, d_2}} \diag(\pi_{t+h})  X_{t+h}                     \\
     & \phantomeq + \frac{1}{T-h} \sum_{t=1}^{T-h} \noise_t'  \diag(\pi_t) \frac{\basis_{d_2} \basis_{d_1}'}{S(h)_{d_1, d_2}} \diag(\pi_{t+h}) X_{t+h}     \\
     & \phantomeq + \frac{1}{T-h} \sum_{t=1}^{T-h} X_t' \diag(\pi_t) \frac{\basis_{d_2} \basis_{d_1}'}{S(h)_{d_1, d_2}} \diag(\pi_{t+h}) \noise_{t+h}      \\
     & \phantomeq + \frac{1}{T-h} \sum_{t=1}^{T-h} \noise_t' \diag(\pi_t) \frac{\basis_{d_2} \basis_{d_1}'}{S(h)_{d_1, d_2}} \diag(\pi_{t+h}) \noise_{t+h}
  \end{align*}
  Let us denote by~$P_t$ the projection of~$\bbR^{TD}$ keeping only the components associated with time~$t$, i.e. such that ~$X_t = P_t X$ and~$\noise_t = P_t \noise$.
  We recognize the following matrix~$L$ in all four lines of the expression above:
  \begin{equation} \label{eq:L_def}
    \begin{aligned}
      L & = \frac{1}{T-h} \sum_{t=1}^{T-h} P_t' \diag(\pi_t) \frac{\basis_{d_2} \basis_{d_1}'}{S(h)_{d_1, d_2}} \diag(\pi_{t+h}) P_{t+h} \\
        & = \frac{1}{T-h} \sum_{t=1}^{T-h} P_t' \frac{\pi_{t+h, d_1} \pi_{t, d_2} \basis_{d_2} \basis_{d_1}'}{S(h)_{d_1, d_2}} P_{t+h}
    \end{aligned}
  \end{equation}
  This leads to:
  \begin{equation*}
    (\widehat{\Gamma}_h + \one_{\{h=0\}} \omega^2 I)_{d_1, d_2} = X' L X + \noise' L X + X' L \noise + \noise' L \noise
  \end{equation*}
  Since~$X$ and~$\noise$ both follow centered multivariate Gaussian distributions, we can express them as linear combinations of standard Gaussian vectors~$g_\innov$ and~$g_\noise$ of dimension~$TD$ (indexed by the source of randomness):
  \begin{equation}
    X = \Psi_{\innov} g_{\innov} \qquad \text{and} \qquad \noise = \Psi_{\noise} g_\noise
  \end{equation}
  where~$\Psi_{\innov}$ and~$\Psi_{\noise}$ are the square roots of the respective covariance matrices
  \begin{equation} \label{eq:psi_definition}
    \Psi_{\innov} = \Cov[X]^{1/2} \qquad \text{and} \qquad
    \Psi_{\noise} = \Cov[\noise]^{1/2} = \omega I.
  \end{equation}
  We substitute~$X$ and~$\noise$ to get:
  \begin{equation*}
    (\widehat{\Gamma}_h + \one_{\{h=0\}} \omega^2 I)_{d_2, d_1} =
    g_{\innov}' \Psi_{\innov}' L \Psi_{\innov} g_{\innov}
    + g_\noise' \Psi_{\noise}' L \Psi_{\innov} g_{\innov}
    + g_{\innov}' \Psi_{\innov}' L \Psi_{\noise} g_\noise
    + g_\noise' \Psi_{\noise}' L \Psi_{\noise} g_\noise,
  \end{equation*}
  which implies
  \begin{align*}
    ( \widehat{\Gamma}_h - \Gamma_h )_{d_1, d_2}
     & =  g_{\innov}' \Psi_{\innov}' L \Psi_{\innov} g_{\innov} - \bbE[g_{\innov}' \Psi_{\innov}' L \Psi_{\innov} g_{\innov}]       \\
     & \phantomeq + g_\noise' \Psi_{\noise}' L \Psi_{\innov} g_{\innov} - \bbE[g_\noise' \Psi_{\noise}' L \Psi_{\innov} g_{\innov}] \\
     & \phantomeq + g_{\innov}' \Psi_{\innov}' L \Psi_{\noise} g_\noise - \bbE[g_{\innov}' \Psi_{\innov}' L \Psi_{\noise} g_\noise] \\
     & \phantomeq + g_\noise' \Psi_{\noise}' L \Psi_{\noise} g_\noise - \bbE[g_\noise' \Psi_{\noise}' L \Psi_{\noise} g_\noise].
  \end{align*}
  The union bound gives us the expected result.
\end{proof}

Now, our goal is to apply a Gaussian concentration inequality to these deviation probabilities.
However, since~$L$ is generated by the discrete sampling process~$\pi$, it is random, and so are the products~$\Psi_a' L \Psi_b$ (where~$a,b \in \{\innov, \noise\}$).
We thus need a conditional version of the Hanson-Wright inequality (Lemma~\ref{lem:conditional_hanson_wright}), in which the following random variables will come into play:
\begin{itemize}
  \item The spectral norm~$\lVert \Psi_a' L \Psi_b \rVert_2$
  \item The Frobenius norm~$\lVert \Psi_a' L \Psi_b \rVert_F^2$
  \item The shifted trace~$\Tr(\Psi_a' L \Psi_b - \bbE[\Psi_a' L \Psi_b])$
\end{itemize}

\subsection{Interlude: Discrete Concentration}

We exploit discrete concentration results to bound the deviations of the three quantities we just mentioned, starting with the norms.

\begin{lemma}[Norm reformulation for~$L$] \label{lem:L_norm_reformulation}
  The spectral and Frobenius norms of~$L$ are given by
  \begin{equation*}
    \lVert L \rVert_2 = \frac{\max_{t \in [T-h]} \pi_{t+h, d_1} \pi_{t, d_2}}{(T-h) S(h)_{d_1, d_2}}   \quad \text{and} \quad
    \lVert L \rVert_F^2 = \frac{1}{(T-h)^2 S(h)_{d_1, d_2}} \sum_{t=1}^{T-h} \pi_{t+h, d_1} \pi_{t, d_2}.
  \end{equation*}
\end{lemma}

\begin{proof}
  We first notice that~$L$ has a block-superdiagonal structure of rank~$h$:
  \begin{equation} \label{eq:L_formula}
    L = \frac{1}{T-h} \sum_{t=1}^{T-h} P_t' L_{[t, t+h]} P_{t+h} \qquad \text{with} \qquad L_{[t, t+h]} = \frac{\pi_{t+h, d_1} \pi_{t, d_2}}{S(h)_{d_1, d_2}} \basis_{d_2} \basis_{d_1}'.
  \end{equation}
  The spectral and Frobenius norms of such a matrix can easily be deduced from those of its blocks.
  Since~$\lVert \basis_{d_2} \basis_{d_1}' \rVert_2 = \lVert \basis_{d_2} \basis_{d_1}' \rVert_F = 1$ and the~$\pi_t$ are binary-valued, this leads to the following formulas:
  \begin{align*}
    \lVert L \rVert_2   & = \frac{1}{T-h} \max_{t \in [T-h]} \lVert L_{[t, t+h]} \rVert_2 = \frac{1}{(T-h) S(h)_{d_1, d_2}} \max_{t \in [T-h]} \pi_{t+h, d_1} \pi_{t, d_2}        \\
    \lVert L \rVert_F^2 & = \frac{1}{(T-h)^2} \sum_{t=1}^{T-h} \lVert L_{[t, t+h]} \rVert_F^2 = \frac{1}{(T-h)^2 S(h)_{d_1, d_2}^2} \sum_{t=1}^{T-h} \pi_{t+h, d_1} \pi_{t, d_2}.
  \end{align*}
\end{proof}

We can bound the spectral norm for free.

\begin{lemma}[Spectral norm bound for~$L$] \label{lem:L_spectral_bound}
  With probability~$1$, the spectral norm~$\lVert L \rVert_2$ satisfies
  \begin{equation*}
    \lVert L \rVert_2 \leq \frac{c}{T p q_u}
  \end{equation*}
\end{lemma}

\begin{proof}
  Note that~$S(h)_{d_1, d_2} \geq p q_u$, and since~$h \in \{0, 1\}$, we can state that~$T-h \geq cT$.
  By Lemma~\ref{lem:L_norm_reformulation}, we deduce
  \begin{equation*}
    \lVert L \rVert_2 = \frac{\max_{t \in [T-h]} \pi_{t+h, d_1} \pi_{t, d_2}}{(T-h) S(h)_{d_1, d_2}} \leq \frac{1}{(T-h) S(h)_{d_1, d_2}} \leq \frac{1}{(T-h) p q_u} \leq \frac{1}{c T p q_u}.
  \end{equation*}
\end{proof}

The Frobenius norm requires a little more work because of the sum it contains.

\begin{lemma}[Concentration of the sampling Bernoullis] \label{lem:proj_concentration}
  For all~$u \in [0, 1]$,
  \begin{equation*}
    \bbP \left( \left\lvert \frac{1}{T-h} \sum_{t=1}^{T-h} \pi_{t+h,d_1} \pi_{t,d_2} - S(h)_{d_1, d_2} \right\rvert \geq u S(h)_{d_1, d_2} \right) \leq c_1 \exp(-c_2 u^2 T S(h)_{d_1, d_2}).
  \end{equation*}
\end{lemma}

\begin{proof}
  We distinguish three cases:
  \begin{itemize}
    \item When~$d_1 = d_2$ and~$h=0$, we have~$\pi_{t+h,d_1} = \pi_{t,d_2}$, which is a~$2$-state Markov chain with transition matrix~$\mathcal{T} \otimes I$, depicted on Figure~\ref{fig:state_space_1}.
    \item When~$d_1 \neq d_2$, the couple~$(\pi_{t,d_2}, \pi_{t+h,d_1})$ is a~$4$-state Markov chain with transition matrix~$\mathcal{T} \otimes \mathcal{T}$ since the chains~$\pi_{t+h, d_1}$ and~$\pi_{t, d_2}$ evolve along independent dimensions. It is shown on Figure~\ref{fig:state_space_2}.
    \item When ~$d_1 = d_2$ and~$h \geq 1$, we must study the~$(h+1)$-tuple~$(\pi_{t, d_1}, \pi_{t+1, d_1}, ..., \pi_{t+h,d_1})$.
          It is a~$2^{h+1}$-state Markov chain with transition matrix~$\mathcal{S}(h)$, whose non-reversible transition diagram can be seen on Figure~\ref{fig:state_space_3}.
  \end{itemize}

  \begin{figure}[htbp]
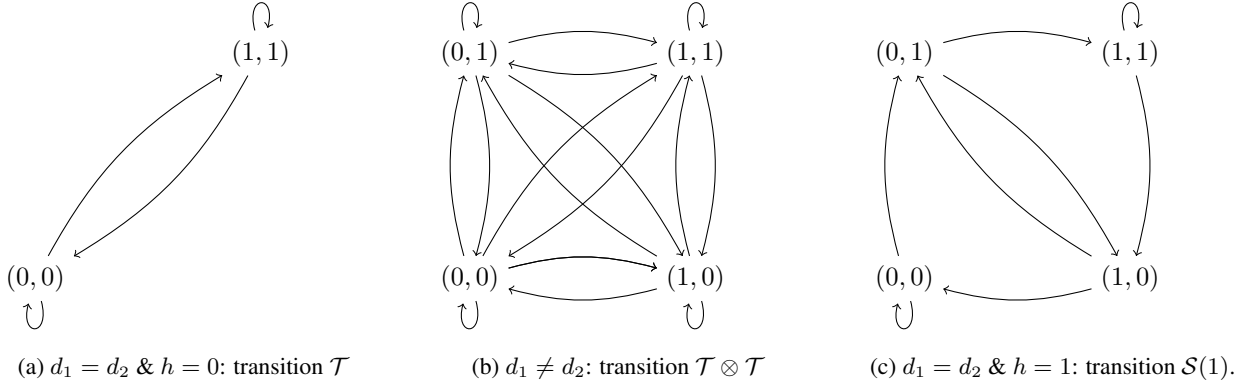

    \begin{subfigure}{0.3\linewidth}
      \tikz{
        \node (00) at (0, 0) {$(0, 0)$};
        \node (11) at (3, 3) {$(1, 1)$};
        \draw[->] (00) to [bend left = 15] (11);
        \draw[->] (11) to [bend left = 15] (00);
        \path (00) edge[loop below] (00);
        \path (11) edge[loop above] (11);
      }
      \subcaption{$d_1 = d_2$~\&~$h=0$: transition~$\mathcal{T}$}
      \label{fig:state_space_1}
    \end{subfigure}
    \hfill
    \begin{subfigure}{0.3\linewidth}
      \tikz{
        \node (00) at (0, 0) {$(0, 0)$};
        \node (01) at (0, 3) {$(0, 1)$};
        \node (10) at (3, 0) {$(1, 0)$};
        \node (11) at (3, 3) {$(1, 1)$};
        \draw[->] (00) to [bend left = 15] (01);
        \draw[->] (00) to [bend left = 15] (10);
        \draw[->] (00) to [bend left = 15] (10);
        \draw[->] (00) to [bend left = 15] (11);
        \draw[->] (01) to [bend left = 15] (00);
        \draw[->] (01) to [bend left = 15] (10);
        \draw[->] (01) to [bend left = 15] (11);
        \draw[->] (10) to [bend left = 15] (00);
        \draw[->] (10) to [bend left = 15] (01);
        \draw[->] (10) to [bend left = 15] (11);
        \draw[->] (11) to [bend left = 15] (00);
        \draw[->] (11) to [bend left = 15] (01);
        \draw[->] (11) to [bend left = 15] (10);
        \path (00) edge[loop below] (00);
        \path (01) edge[loop above] (01);
        \path (10) edge[loop below] (10);
        \path (11) edge[loop above] (11);
      }
      \subcaption{$d_1 \neq d_2$: transition~$\mathcal{T} \otimes \mathcal{T}$}
      \label{fig:state_space_2}
    \end{subfigure}
    \hfill
    \begin{subfigure}{0.3\linewidth}
      \tikz{
        \node (00) at (0, 0) {$(0, 0)$};
        \node (01) at (0, 3) {$(0, 1)$};
        \node (10) at (3, 0) {$(1, 0)$};
        \node (11) at (3, 3) {$(1, 1)$};
        \draw[->] (00) to [bend left = 15] (01);
        \draw[->] (01) to [bend left = 15] (11);
        \draw[->] (11) to [bend left = 15] (10);
        \draw[->] (10) to [bend left = 15] (00);
        \draw[->] (01) to [bend left = 15] (10);
        \draw[->] (10) to [bend left = 15] (01);
        \path (00) edge[loop below] (00);
        \path (11) edge[loop above] (11);
      }
      \subcaption{$d_1 = d_2$~\&~$h=1$: transition~$\mathcal{S}(1)$.}
      \label{fig:state_space_3}
    \end{subfigure}

    \caption{State space and transitions for the Markov chains used in the discrete concentration result} 
    \label{fig:state_space}
  \end{figure}

  In all of these cases, our variable of interest~$\pi_{t+h, d_1} \pi_{t, d_2}$ is a function of the underlying Markov chain.
  The relevant functions are:
  \begin{equation*}
    f_1: x \mapsto x \qquad f_2: (x,y) \mapsto yx \qquad f_3: (x_0, ..., x_h) \mapsto x_h x_0.
  \end{equation*}
  We note that since~$\chi \leq a, b \leq 1-\chi$, all the coefficients of~$\mathcal{T}$ are greater than~$\chi$.
  Furthermore, all the coefficients of~$\mathcal{T} \otimes \mathcal{T}$ are greater than~$\chi^2$.
  Finally, all the coefficients of~$\mathcal{S}(h)^{h+1}$ are greater than~$\chi^{h+1}$, because all pairs of states are connected after~$h+1$ steps.
  Let us illustrate this with~$h=1$:
  \begin{equation*}
    \mathcal{S}(1) =  \begin{pmatrix}
      1-a & a & 0 & 0   \\
      0   & 0 & b & 1-b \\
      1-a & a & 0 & 0   \\
      0   & 0 & b & 1-b
    \end{pmatrix}
    \qquad
    \mathcal{S}(1)^2 =  \begin{pmatrix}
      (1-a)^2 & a(1-a) & ab     & a(1-b)  \\
      (1-a)b  & ab     & (1-b)b & (1-b)^2 \\
      (1-a)^2 & a(1-a) & ab     & a(1-b)  \\
      (1-a)b  & ab     & (1-b)b & (1-b)^2
    \end{pmatrix}.
  \end{equation*}
  Subsequently, all the transition matrices~$\mathcal{R}$ we are interested in, namely~$\mathcal{R} \in \{\mathcal{T}, \mathcal{T} \otimes \mathcal{T}, \mathcal{S}(h)^{h+1}\}$, satisfy the Doeblin condition with~$r = h+1$ and~$\delta = \chi^{h+1}$:
  \begin{equation*}
    \mathcal{R}^{h+1} \geq \chi^{h+1} \begin{pmatrix}
      1      & \cdots & 1      \\
      \vdots & \ddots & \vdots \\
      1      & \cdots & 1
    \end{pmatrix}.
  \end{equation*}
  Since we only consider~$h \in \{0, 1\}$ and since~$\chi$ is fixed for our purposes, these chains fulfill the assumptions of Lemma~\ref{lem:chernoff_doeblin}.
  We thus conclude:
  \begin{equation*}
    \bbP\left( \left\lvert \frac{1}{T-h} \sum_{t=1}^{T-h} \pi_{t+h,d_1} \pi_{t,d_2} - S(h)_{d_1, d_2} \right\rvert \geq u S(h)_{d_1, d_2} \right) \leq c_1 \exp\left(-c_2 u^2 (T-h) S(h)_{d_1, d_2} \right).
  \end{equation*}
  We finally replace~$T-h$ with~$cT$ in the exponential, leading to the result we announced.
\end{proof}

Based on this concentration property, we can now bound the norms of the random matrix~$L$ with high probability.

\begin{lemma}[Frobenious norm bound for~$L$] \label{lem:L_frob_bound}
  For any~$\delta$ such that Equation~\eqref{eq:T_constraint_convergence1} holds, with probability at least~$1-\delta$, the Frobenius norm~$\lVert L \rVert_F^2$ satisfies
  \begin{equation*}
    \lVert L \rVert_F^2 \leq \frac{c}{T p q_u}.
  \end{equation*}
\end{lemma}

\begin{proof}
  By Lemma~\ref{lem:proj_concentration}: for all~$u \in [0, 1]$,
  \begin{equation*}
    \bbP \left( \frac{1}{T-h} \sum_{t=1}^{T-h} \pi_{t+h, d_1} \pi_{t, d_2} \geq (1+u) S(h)_{d_1, d_2}  \right)
    \leq c_1 \exp(-c_2 u^2 T S(h)_{d_1, d_2}).
  \end{equation*}
  We remember the expression of Lemma~\ref{lem:L_norm_reformulation} for~$\lVert L \rVert_F^2$ and notice that:
  \begin{align*}
     & \bbP\left( \lVert L \rVert_F^2 \geq \frac{1+u}{(T-h)S(h)_{d_1, d_2}}\right)                                                                                                                          \\
     & \quad = \bbP \left( \frac{1}{(T-h)S(h)_{d_1, d_2}} \left(\frac{1}{T-h} \sum_{t=1}^{T-h} \frac{\pi_{t+h, d_1} \pi_{t, d_2}}{S(h)_{d_1, d_2}} \right) \geq \frac{1}{(T-h)S(h)_{d_1, d_2}}(1+u) \right) \\
     & \quad \leq c_1 \exp(-c_2 u^2 T S(h)_{d_1, d_2})
  \end{align*}
  We finally recall that~$S(h)_{d_1, d_2} \geq p q_u$ and~$T-h \geq cT$, so that
  \begin{align*}
    \bbP \left( \lVert L \rVert_F^2 \geq \frac{1+u}{c T p q_u}  \right)
     & \leq \bbP \left( \lVert L \rVert_F^2 \geq \frac{1+u}{(T-h) p q_u}  \right)           \\
     & \leq \bbP \left( \lVert L \rVert_F^2 \geq \frac{1+u}{(T-h) S(h)_{d_1, d_2}}  \right) \\
     & \leq c_1 \exp \left(-c_2 u^2 T S(h)_{d_1, d_2} \right)                               \\
     & \leq c_1 \exp \left(-c_2 u^2 T p q_u \right).
  \end{align*}
  All we need to make sure that~$\bbP \left( \lVert L \rVert_F^2 \geq \frac{1+u}{c T p q_u}  \right) \leq \delta$ is to choose~$u$ such that
  \begin{equation*}
    c_1 \exp \left(-c_2 u^2 T p q_u \right) \leq \delta \quad \iff \quad u \geq \sqrt{\frac{\log(c_1/\delta)}{c_2 T p q_u}}
  \end{equation*}
  Note that we can replace~$\log(c_1 / \delta)$ by a constant times~$\log(1/\delta)$ to simplify expressions: this is possible as long as~$\delta$ is chosen \enquote{small enough} (i.e. smaller than some universal constant).
  We will assume this fairly often in the rest of the proof.

  For Lemma~\ref{lem:proj_concentration} to apply, we must ensure that our choice of~$u$ is smaller than~$1$.
  With the previous discussion in mind, $u \leq 1$ is implied by
  \begin{equation} \label{eq:T_constraint_convergence1}
    \sqrt{\frac{\log(1/\delta)}{T p q_u}} \leq c.
  \end{equation}
  If this holds, then we have
  \begin{equation*}
    \bbP \left( \lVert L \rVert_F^2 \geq \frac{2}{c T p q_u} \right) \leq \bbP \left( \lVert L \rVert_F^2 \geq \frac{1+u}{c T p q_u}  \right) \leq \delta.
  \end{equation*}
  This yields the result we wanted.
\end{proof}

We now move on to studying the shifted trace of~$\Psi_a' L \Psi_b$, which is the last ingredient we need for our application of Lemma~\ref{lem:conditional_hanson_wright}.

\begin{lemma}[Trace bound for the~$L$ matrices] \label{lem:L_trace}
  For all~$u \in [0, 1]$,
  \begin{align*}
    \bbP(|\Tr(\Psi_{\innov}' L \Psi_{\innov} - \bbE[\Psi_{\innov}' L \Psi_{\innov}])| \geq u)
     & \leq c_1 \exp\left(-\frac{c_2 u^2 T p q_u}{\lVert \Gamma_h \rVert_2^2} \right) \\
    \bbP(|\Tr(\Psi_{\noise}' L \Psi_{\noise} - \bbE[\Psi_{\noise}' L \Psi_{\noise}])| \geq u)
     & \leq c_1 \exp\left(-\frac{c_2 u^2 T p q_u}{\omega^4} \right).
  \end{align*}
\end{lemma}

\begin{proof}
  We can compute an explicit formula thanks to Equation~\eqref{eq:L_formula}: if~$a \in \{\innov, \noise\}$ then
  \begin{align*}
    \Tr(\Psi_{a}' L \Psi_{a})
     & = \Tr \left(\frac{1}{T-h} \sum_{t=1}^{T-h} \Psi_{a}' P_t' \frac{\pi_{t+h, d_1} \pi_{t, d_2}}{S(h)_{d_1, d_2}} \basis_{d_2} \basis_{d_1}' P_{t+h} \Psi_{a} \right)   \\
     & = \frac{1}{T-h} \sum_{t=1}^{T-h}  \frac{\pi_{t+h, d_1} \pi_{t, d_2}}{S(h)_{d_1, d_2}} \Tr \left(\Psi_{a}' P_t' \basis_{d_2} \basis_{d_1}' P_{t+h} \Psi_{a}  \right) \\
     & = \frac{1}{T-h} \sum_{t=1}^{T-h}  \frac{\pi_{t+h, d_1} \pi_{t, d_2}}{S(h)_{d_1, d_2}}  \left(\basis_{d_1}' P_{t+h} \Psi_{a} \Psi_{a}' P_t' \basis_{d_2} \right)     \\
     & = \frac{1}{T-h} \sum_{t=1}^{T-h}  \frac{\pi_{t+h, d_1} \pi_{t, d_2}}{S(h)_{d_1, d_2}} \left( (\Psi_{a} \Psi_{a}')_{[t+h, t]} \right)_{d_1, d_2}
  \end{align*}
  where~$\left( (\Psi_{a} \Psi_{a}')_{[t+h, t]} \right)_{d_1, d_2}$ denotes the~$(d_1, d_2)$ coefficient of the~$(t+h, t)$ block of~$\Psi_{a} \Psi_{a}'$.
  Now is the time to look back on Equation~\eqref{eq:psi_definition}, which tells us that both~$\Psi_\innov \Psi_\innov'$ and~$\Psi_\noise \Psi_\noise'$ are constant along their superdiagonal of rank~$h$.
  We thus find that
  \begin{align*}
    \Tr(\Psi_{\innov}' L \Psi_{\innov} - \bbE[\Psi_{\innov}' L \Psi_{\innov}])
     & = (\Gamma_h)_{d_1, d_2} \left( \frac{1}{T-h} \sum_{t=1}^{T-h} \frac{\pi_{t+h,d_1} \pi_{t, d_2}}{S(h)_{d_1, d_2}} - 1 \right)                  \\
    \Tr(\Psi_{\noise}' L \Psi_{\noise} - \bbE[\Psi_{\noise}' L \Psi_{\noise}])
     & = (\one_{\{h=0\}} \omega^2 I)_{d_1, d_2} \left( \frac{1}{T-h} \sum_{t=1}^{T-h} \frac{\pi_{t+h,d_1} \pi_{t, d_2}}{S(h)_{d_1, d_2}} - 1 \right)
  \end{align*}
  Like before, we can apply Lemma~\ref{lem:proj_concentration}: for all~$u \in [0, 1]$,
  \begin{equation*}
    \bbP \left( \left\lvert \frac{1}{T-h} \sum_{t=1}^{T-h} \frac{\pi_{t+h, d_1} \pi_{t, d_2}}{S(h)_{d_1, d_2}} - 1 \right\rvert \geq u \right) \leq c_1 \exp(-c_2 u^2TS(h)_{d_1, d_2}) \leq c_1 \exp(-c_2 u^2 T p q_u).
  \end{equation*}
  Since~$|(\Gamma_h)_{d_1, d_2}| \leq \lVert \Gamma_h \rVert_2$ and~$(\one_{\{h=0\}} \omega^2 I)_{d_1, d_2} \leq \omega^2$, we can deduce
  \begin{align*}
    \bbP\left(|\Tr(\Psi_{\innov}' L \Psi_{\innov} - \bbE[\Psi_{\innov}' L \Psi_{\innov}])| \geq u \lVert \Gamma_h \rVert_2 \right)
     & \leq c_1 \exp(-c_2 u^2 T p q_u) \\
    \bbP\left(|\Tr(\Psi_{\noise}' L \Psi_{\noise} - \bbE[\Psi_{\noise}' L \Psi_{\noise}])| \geq u \omega^2 \right)
     & \leq c_1 \exp(-c_2 u^2 T p q_u)
  \end{align*}
  which, after rescaling, yields the result we announced.
\end{proof}

\subsection{Gaussian Concentration, Episode 2} \label{sec:convergence_gamma_proof}

We are now ready to apply our conditional concentration result.

\begin{lemma}[Applying Hanson-Wright] \label{lem:hanson_wright_application}
  Let~$\delta > 0$ and~$u \in [0, 1]$.
  Assume that Equations~\eqref{eq:T_constraint_convergence1} and~\eqref{eq:minimum_quadratic} hold.
  Then the deviation probability for~$(\widehat{\Gamma}_h)_{d_1, d_2}$ satisfies
  \begin{equation*}
    \bbP(|(\widehat{\Gamma}_h - \Gamma_h)_{d_1, d_2}| \geq u)
    \leq 4 \delta + c_1 \exp \left(- \frac{c_2 u^2 T p q_u}{\max \left\{(\lVert \Psi_{\innov} \rVert_2^2 + \omega^2)^2, \lVert \Gamma_h \rVert_2^2, \omega^4 \right\}} \right).
  \end{equation*}
\end{lemma}

\begin{proof}
  The conclusion we had reached before our discrete interlude is given by Lemma~\ref{lem:proba_split}, and we can rewrite it as
  \begin{equation*}
    \bbP(|(\widehat{\Gamma}_h - \Gamma_h)_{d_1, d_2}| \geq u) \leq p_{\innov \innov} + p_{\noise \innov} + p_{\innov \noise} + p_{\noise \noise},
  \end{equation*}
  where each~$p_{ab}$ represents a deviation probability for a specific quadratic form~$g_a' \Psi_a' L \Psi_b g_b$.
  Let us choose~$\delta$ such that Equation~\eqref{eq:T_constraint_convergence1} holds.
  By Lemmas~\ref{lem:L_spectral_bound} and~\ref{lem:L_frob_bound}, with probability at least~$1-\delta$, the following eight inequalities occur at the same time (we use Lemma~\ref{lem:frobenius_spectral_product} to split the products):
  \begin{align*}
    \lVert \Psi_{\innov}' L \Psi_{\innov} \rVert_F^2 & \leq \frac{c \lVert \Psi_{\innov} \rVert_2^4}{T p q_u}                                 &
    \lVert \Psi_{\innov}' L \Psi_{\innov} \rVert_2   & \leq \frac{c \lVert \Psi_{\innov} \rVert_2^2}{T p q_u}                                   \\
    \lVert \Psi_{\noise}' L \Psi_{\innov} \rVert_F^2 & \leq \frac{c \lVert \Psi_{\noise} \rVert_2^2 \lVert \Psi_{\innov} \rVert_2^2}{T p q_u} &
    \lVert \Psi_{\noise}' L \Psi_{\innov} \rVert_2   & \leq \frac{c \lVert \Psi_{\noise} \rVert_2 \lVert \Psi_{\innov} \rVert_2}{T p q_u}       \\
    \lVert \Psi_{\innov}' L \Psi_{\noise} \rVert_F^2 & \leq \frac{c \lVert \Psi_{\innov} \rVert_2^2 \lVert \Psi_{\noise} \rVert_2^2}{T p q_u} &
    \lVert \Psi_{\innov}' L \Psi_{\noise} \rVert_2   & \leq \frac{c \lVert \Psi_{\innov} \rVert_2 \lVert \Psi_{\noise} \rVert_2}{T p q_u}       \\
    \lVert \Psi_{\noise}' L \Psi_{\noise} \rVert_F^2 & \leq \frac{c \lVert \Psi_{\noise} \rVert_2^4}{T p q_u}                                 &
    \lVert \Psi_{\noise}' L \Psi_{\noise} \rVert_2   & \leq \frac{c \lVert \Psi_{\noise} \rVert_2^2}{T p q_u}.
  \end{align*}
  The spectral norm of~$\Psi_{\noise}$ is easily seen to equal~$\lVert \Psi_{\noise} \rVert_2 = \lVert \omega^2 I \rVert_2^{1/2} = \omega$, which allows us to lighten these expressions.
  From there, Lemma~\ref{lem:conditional_hanson_wright} (applied with~$X=g_a$,~$Y=g_b$ and~$A = \Psi_a' L \Psi_b$) provides the concentration bounds we need\footnote{
    The additional trace terms that appear when applying Lemma~\ref{lem:conditional_hanson_wright} (as opposed to the non-conditional version of Lemma~\ref{lem:hanson_wright}) are absent from the papers by \citet{raoEstimationAutoregressiveProcesses2017, raoFundamentalEstimationLimits2017}, which is why we think their upper bound proofs are incomplete.
  }:
  \begin{align*}
    p_{\innov \innov} & \leq \delta + 2 \exp \left( -c T p q_u \min \left\{\frac{(u/4)^2}{\lVert \Psi_{\innov} \rVert_2^4}, \frac{(u/4)}{\lVert \Psi_{\innov} \rVert_2^2} \right\} \right) + \bbP\left(|\Tr(\Psi_{\noise}' L \Psi_{\innov}) - \bbE[\Psi_{\innov}' L \Psi_{\innov}])| \geq u/8\right) \\
    p_{\noise \innov} & \leq \delta + 2 \exp \left( -c T p q_u \min \left\{\frac{(u/4)^2}{\omega^2 \lVert \Psi_{\innov} \rVert_2^2}, \frac{(u/4)}{\omega \lVert \Psi_{\innov} \rVert_2} \right\} \right)                                                                                             \\
    p_{\innov \noise} & \leq \delta + 2 \exp \left( -c T p q_u \min \left\{\frac{(u/4)^2}{\lVert \Psi_{\innov} \rVert_2^2 \omega^2}, \frac{(u/4)}{\lVert \Psi_{\innov} \rVert_2 \omega} \right\} \right)                                                                                             \\
    p_{\noise \noise} & \leq \delta + 2 \exp \left( -c T p q_u \min \left\{\frac{(u/4)^2}{\omega^4}, \frac{(u/4)}{\omega^2} \right\} \right) + \bbP\left(|\Tr(\Psi_{\noise}' L \Psi_{\noise}) - \bbE[\Psi_{\noise}' L \Psi_{\noise}])| \geq u/8\right).
  \end{align*}
  The denominators inside the minima can be unified: for the left column,
  \begin{equation*}
    \max \left\{ \lVert \Psi_{\innov} \rVert_2^4, \lVert \Psi_{\innov} \rVert_2^2 \omega^2, \omega^4 \right\} \leq \left(\lVert \Psi_{\innov} \rVert_2^2 + \omega^2\right)^2,                                                                   \end{equation*}
  and for the right column,
  \begin{equation*}
    \max \left\{ \lVert \Psi_{\innov} \rVert_2^2, \lVert \Psi_{\innov} \rVert_2 \omega, \omega^2 \right\} \leq \left(\lVert \Psi_{\innov} \rVert_2 + \omega\right)^2 \leq 2\left(\lVert \Psi_{\innov} \rVert_2^2 + \omega^2\right).
  \end{equation*}
  This means we can upper bound each of the four minima by
  \begin{equation*}
    \min \left\{\left(\frac{u/4}{\lVert \Psi_{\innov} \rVert_2^2 + \omega^2}\right)^2, \frac{u/8}{\lVert \Psi_{\innov} \rVert_2^2 + \omega^2} \right\}.
  \end{equation*}
  From now on, we additionally suppose that
  \begin{equation} \label{eq:minimum_quadratic}
    \frac{u/4}{\lVert \Psi_{\innov} \rVert_2^2 + \omega^2} \leq \frac{1}{2}
  \end{equation}
  This enables us to get rid of these minima by reducing them to the (smaller) quadratic term on the left.
  We end up with
  \begin{align*}
    p_{\innov \innov} & \leq \delta + 2 \exp \left( -\frac{c u^2 T p q_u}{\left(\lVert \Psi_{\innov} \rVert_2^2 + \omega^2\right)^2} \right) + \bbP\left(|\Tr(\Psi_{\innov}' L \Psi_{\innov}) - \bbE[\Psi_{\innov}' L \Psi_{\innov}])| \geq u/8\right)  \\
    p_{\noise \innov} & \leq \delta + 2 \exp \left( -\frac{c u^2 T p q_u}{\left(\lVert \Psi_{\innov} \rVert_2^2 + \omega^2\right)^2} \right)                                                                                                            \\
    p_{\innov \noise} & \leq \delta + 2 \exp \left( -\frac{c u^2 T p q_u}{\left(\lVert \Psi_{\innov} \rVert_2^2 + \omega^2\right)^2} \right)                                                                                                            \\
    p_{\noise \noise} & \leq \delta + 2 \exp \left( -\frac{c u^2 T p q_u}{\left(\lVert \Psi_{\innov} \rVert_2^2 + \omega^2\right)^2} \right) + \bbP\left(|\Tr(\Psi_{\noise}' L \Psi_{\noise}) - \bbE[\Psi_{\noise}' L \Psi_{\noise}])| \geq u/8\right).
  \end{align*}
  As for the trace terms, they are taken care of by Lemma~\ref{lem:L_trace}:
  \begin{align*}
    \bbP\left(|\Tr(\Psi_{\noise}' L \Psi_{\innov}) - \bbE[\Psi_{\innov}' L \Psi_{\innov}])| \geq u/8\right)
     & \leq c_3 \exp\left(-c_4 \frac{(u/8)^2 T p q_u}{\lVert \Gamma_h \rVert_2^2} \right) \\
    \bbP\left(|\Tr(\Psi_{\noise}' L \Psi_{\noise}) - \bbE[\Psi_{\noise}' L \Psi_{\noise}])| \geq u/8\right)
     & \leq c_3 \exp\left(-c_4 \frac{(u/8)^2 T p q_u}{\omega^4} \right)
  \end{align*}
  We plug this in and rearrange to get:
  \begin{align*}
    p_{\innov \innov} + p_{\noise \innov} + p_{\innov \noise} + p_{\noise \noise}
     & \leq 4 \delta + c_1 \exp \left(- \frac{c_2 u^2 T p q_u}{\max \left\{(\lVert \Psi_{\innov} \rVert_2^2 + \omega^2)^2, \lVert \Gamma_h \rVert_2^2, \omega^4 \right\}} \right).
  \end{align*}
\end{proof}

The following result will simplify the denominator inside the exponential.

\begin{lemma}[Spectral norms of~$\Psi_{\innov}$ and~$\Gamma_h$] \label{lem:psi_norms}
  The matrices~$\Psi_{\innov}$ and~$\Gamma_h$ satisfy:
  \begin{equation*}
    \lVert \Psi_{\innov} \rVert_2^2 \leq \frac{\sigma_{\max}^2}{(1 - \thetamax)^2} \qquad \text{and} \qquad \lVert \Gamma_h \rVert_2 \leq \frac{\thetamax^{h} \sigma_{\max}^2}{1-\thetamax} .
  \end{equation*}
  As a consequence,
  \begin{equation*}
    \max \left\{(\lVert \Psi_{\innov} \rVert_2^2 + \omega^2)^2, \lVert \Gamma_h \rVert_2^2, \omega^4 \right\} \leq \frac{(\sigma_{\max}^2 + \omega^2)^2}{(1-\thetamax)^4}
  \end{equation*}
\end{lemma}

\begin{proof}
  By Lemma~\ref{lem:x_covariance}, we can write~$\Psi_{\innov}^2$ as a sum of Kronecker products (one for each block).
  Let~$J_t$ he a matrix full of zeros, except for the subdiagonal of rank~$t$, which is full of ones.
  Then we have:
  \begin{equation*}
    \Psi_{\innov}^2 = \Cov[X] = I \otimes \Gamma_0(\theta) + \sum_{t=1}^{T-1} \left[ \Jmat_t \otimes \theta^t \Gamma_0(\theta) + \Jmat_t' \otimes \Gamma_0(\theta) \theta'^t \right]
  \end{equation*}
  This gives us control over its spectral norm thanks to Lemma~\ref{lem:kron_singular_values}:
  \begin{align*}
    \lVert \Psi_{\innov} \rVert_2^2 = \lVert \Psi_{\innov}^2 \rVert_2
     & \leq  \lVert I \rVert_2 \times \lVert \Gamma_0(\theta) \rVert_2 + \sum_{t=1}^{T-1} \left[ \lVert \Jmat_t \rVert_2 \times \lVert \theta^t \Gamma_0(\theta) \rVert_2 + \lVert \Jmat_t' \rVert_2 \times \lVert \Gamma_0(\theta) \theta'^t \rVert_2 \right] \\
     & \leq \lVert \Gamma_0(\theta) \rVert_2 \left( 1 + 2\sum_{t=1}^{T-1} \lVert \theta \rVert_2^t \right)
  \end{align*}
  We now use Lemma~\ref{lem:norm_gamma}:
  \begin{align*}
    \lVert \Psi_{\innov} \rVert_2^2
     & \leq \frac{\sigma_{\max}^2}{1 - \thetamax^2} \left(1 + 2 \frac{\thetamax}{1 - \thetamax} \right) = \frac{\sigma_{\max}^2}{(1-\thetamax)^2}.
  \end{align*}
  We now turn to~$\Gamma_h$ with Lemmas~\ref{lem:x_covariance} and~\ref{lem:norm_gamma}:
  \begin{equation*}
    \lVert \Gamma_h \rVert_2 = \lVert \theta^h \Gamma_0(\theta) \rVert_2 \leq \frac{\thetamax^h \sigma_{\max}^2}{1-\thetamax^2}.
  \end{equation*}
  In particular, we have
  \begin{align*}
    \max \left\{(\lVert \Psi_{\innov} \rVert_2^2 + \omega^2)^2, \lVert \Gamma_h \rVert_2^2, \omega^4 \right\}
     & \leq \max \left\{\left(\frac{\sigma_{\max}^2}{(1-\thetamax)^2} + \omega^2\right)^2, \left(\frac{\thetamax^h \sigma_{\max}^2}{1-\thetamax^2}\right)^2, \omega^4 \right\} \\
     & \leq \frac{(\sigma_{\max}^2 + \omega^2)^2}{(1-\thetamax)^4}
  \end{align*}
\end{proof}

We can now control the error of the covariance estimator:

\begin{lemma}[Max norm convergence rate of the covariance estimator] \label{lem:convergence_rate_gamma}
  Let~$\delta > 0$ be small enough.
  Assume that Equations~\eqref{eq:T_constraint_convergence1} and~\eqref{eq:T_constraint_convergence2} hold.
  Then the covariance estimator~$\widehat{\Gamma}_h$ from Equation~\eqref{eq:gamma_estimator} satisfies
  \begin{equation*}
    \lVert \widehat{\Gamma}_h - \Gamma_h \rVert_{\max} \leq c \frac{\sigma_{\max}^2 + \omega^2}{(1-\thetamax)^2} \frac{\sqrt{\log(D/\delta)}}{\sqrt{T p q_u}} = \err(\delta)
  \end{equation*}
  with probability greater than~$1 - \delta$.
\end{lemma}

\begin{proof}
  Let us plug Lemma~\ref{lem:psi_norms} into Lemma~\ref{lem:hanson_wright_application}
  \begin{equation*}
    \bbP(|(\widehat{\Gamma}_h - \Gamma_h)_{d_1, d_2}| \geq u)
    \leq 4 \delta + c_1 \exp \left(-\frac{c_2 (1-\thetamax)^4 u^2 T p q_u}{(\sigma_{\max}^2 + \omega^2)^2} \right).
  \end{equation*}
  All that is left to do is choose~$u$ such that
  \begin{equation*}
    \bbP(|(\widehat{\Gamma}_h - \Gamma_h)_{d_1, d_2}| \geq u) \leq 8\delta,
  \end{equation*}
  which will be true if
  \begin{equation*}
    c_1 \exp \left(- \frac{c_2 (1-\thetamax)^4 T p q_u}{(\sigma_{\max}^2 + \omega^2)^2} u^2 \right) \leq 4\delta \quad \iff \quad u \geq \sqrt{\frac{\log(c_1/4\delta) (\sigma_{\max}^2 + \omega^2)^2}{c_2 (1-\thetamax)^4 T p q_u}}.
  \end{equation*}
  As long as~$\delta$ is small enough, we can take
  \begin{equation}
    u = c \frac{\sqrt{\log(1/\delta)} (\sigma_{\max}^2 + \omega^2)}{(1-\thetamax)^2 \sqrt{T p q_u}}.
  \end{equation}
  For Lemma~\ref{lem:hanson_wright_application} to apply, we must verify that~$u \in [0, 1]$ and that Equation~\eqref{eq:minimum_quadratic} is satisfied.
  In other words, we have to ensure that
  \begin{equation*}
    c \frac{\sqrt{\log(1/\delta)} (\sigma_{\max}^2 + \omega^2)}{(1-\thetamax)^2 \sqrt{T p q_u}}
    \leq \min\{1, 2 (\lVert \Psi_{\innov} \rVert_2^2 + \omega^2)\}
  \end{equation*}
  Using Lemma~\ref{lem:psi_norms}, this is implied by the condition
  \begin{equation}  \label{eq:T_constraint_convergence2}
    \frac{\sqrt{\log(1/\delta)} \max \{1, (\sigma_{\max}^2 + \omega^2)^{-1} \}}{(1-\thetamax)^2 \sqrt{T p q_u}} \leq c
  \end{equation}
  Under these hypotheses, we just proved that with probability at least~$1-8\delta$,
  \begin{equation*}
    |(\widehat{\Gamma}_h - \Gamma_h)_{d_1, d_2}|  \leq c \frac{\sigma_{\max}^2 + \omega^2}{(1-\thetamax)^2} \frac{\sqrt{\log(1/\delta)}}{\sqrt{T p q_u}}.
  \end{equation*}
  We finish with a union bound, applying the previous result to all pairs~$(d_1, d_2) \in [D]^2$.
  With probability greater than~$1-8D^2\delta$, we have:
  \begin{equation*}
    \max_{d_1, d_2} |(\widehat{\Gamma}_h - \Gamma_h)_{d_1, d_2}| = \lVert \widehat{\Gamma}_h - \Gamma_h \rVert_{\max} \leq c \frac{\sigma_{\max}^2 + \omega^2}{(1-\thetamax)^2} \frac{\sqrt{\log(1/\delta)}}{\sqrt{T p q_u}}.
  \end{equation*}
  Replacing~$\delta$ with~$8 D^2 \delta$ gives us the result we wanted: with probability greater than~$1-\delta$,
  \begin{equation*}
    \lVert \widehat{\Gamma}_h - \Gamma_h \rVert_{\max} \leq c \frac{\sigma_{\max}^2 + \omega^2}{(1-\thetamax)^2} \frac{\sqrt{\log(D/\delta)}}{\sqrt{T p q_u}}.
  \end{equation*}
\end{proof}

\subsection{Behavior of the Dantzig selector} \label{sec:convergence_theta_proof}

We now walk the final steps from the error on~$\widehat{\Gamma}_h$ to the error on~$\widehat{\theta}$.
In order to recover Theorem~\ref{thm:convergence_rate_theta}, we adapt the convergence proof from \citet[Appendix A.1]{hanDirectEstimationHigh2015}.
However, we use our own notations and our custom concentration results for~$\widehat{\Gamma}_h$.
To make comparison between both papers easier, we provide a dictionary of the main notations in Table~\ref{tab:notations_trans}.

\begin{table}
  \centering
  \begin{tabular}{c|c|c|}
                            & This paper                                                                     & \citet{hanDirectEstimationHigh2015}              \\ \hline
    VAR def                 & ~$X_t = \theta X_{t-1} + \varepsilon_t$                                        & ~$X_t = A_1' X_{t-1} + Z_t$                      \\
    Covariance              & ~$\Gamma_h = \Cov(X_h, X_0)~$                                                  & ~$\Sigma_i = \Cov(X_0, X_i)$                     \\
    Yule-Walker             & ~$\Gamma_h = \theta^h \Gamma_0$                                                & ~$\Sigma_i = \Sigma_0 A_1^i$                     \\
    Covariance estimate     & ~$\widehat{\Gamma}_h$                                                          & ~$S_i$                                           \\
    Covariance error        & ~$\err(\delta)$                                                                & ~$\zeta_i$                                       \\
    Optimization constraint & ~$\lVert M \widehat{\Gamma}_0 - \widehat{\Gamma}_1 \rVert_{\max} \leq \lambda$ & ~$\lVert S_0 M - S_1 \rVert_{\max} \leq \lambda$ \\
    Optimization objective  & ~$\lVert \vecm(M) \rVert_1$                                                    & ~$\lVert \vecm(M) \rVert_1$                      \\
    Threshold in proof      & ~$\nu$                                                                         & ~$\lambda_1$                                     \\
  \end{tabular}
  \caption{Notation correspondence between this paper and \citet{hanDirectEstimationHigh2015}}
  \label{tab:notations_trans}
\end{table}

Our sparse transition estimator is defined as a solution to~\eqref{eq:theta_estimator}.
The end goal is to control the error~$\lVert\widehat{\theta} - \theta\rVert_1$, where~$\theta = \Gamma_1 \Gamma_0^{-1}$ is the true transition matrix.
We start by choosing a specific~$\lambda$ such that~$\theta$ is feasible with high probability.

\begin{lemma}[Feasibility of the real~$\theta$] \label{lem:feasibility_theta}
  If we select the penalization level
  \begin{equation*}
    \lambda = (\lVert \theta \rVert_{\infty} + 1) \err(\delta),
  \end{equation*}
  then with probability at least~$1-\delta$, the real~$\theta$ is a feasible solution to the optimization problem~\eqref{eq:theta_estimator}.
\end{lemma}

\begin{proof}
  \begin{align*}
    \lVert \theta \widehat{\Gamma}_0 - \widehat{\Gamma}_1 \rVert_{\max}
     & = \lVert \Gamma_1 \Gamma_0^{-1} \ \widehat{\Gamma}_0 - \widehat{\Gamma}_1 \rVert_{\max}                                                                      \\
     & = \lVert \Gamma_1 \Gamma_0^{-1} \ \widehat{\Gamma}_0 - \Gamma_1 + \Gamma_1 - \widehat{\Gamma}_1 \rVert_{\max}                                                \\
     & \leq \lVert \Gamma_1 \Gamma_0^{-1} \ \widehat{\Gamma}_0 - \Gamma_1 \Gamma_0^{-1} \Gamma_0 \rVert_{\max} + \lVert \Gamma_1 - \widehat{\Gamma}_1 \rVert_{\max} \\
     & = \lVert \theta (\widehat{\Gamma}_0 - \Gamma_0)  \rVert_{\max} + \lVert \Gamma_1 - \widehat{\Gamma}_1 \rVert_{\max}
  \end{align*}
  By Lemma~\ref{lem:maxnorm_l1linf},
  \begin{equation*}
    \lVert \theta (\widehat{\Gamma}_0 - \Gamma_0)  \rVert_{\max} \leq \lVert \theta \rVert_{\infty} \lVert \widehat{\Gamma}_0 - \Gamma_0  \rVert_{\max}
  \end{equation*}
  By Lemma~\ref{lem:convergence_rate_gamma}, with probability greater than~$1-2\delta$,
  \begin{equation*}
    \lVert \widehat{\Gamma}_0 - \Gamma_0 \rVert_{\max} \leq \err(\delta) \quad \text{and} \quad
    \lVert \widehat{\Gamma}_1 - \Gamma_1 \rVert_{\max} \leq \err(\delta)
  \end{equation*}
  which implies
  \begin{equation*}
    \lVert \theta \widehat{\Gamma}_0 - \widehat{\Gamma}_1 \rVert_{\max} \leq (\lVert \theta \rVert_{\infty} + 1) \err(\delta).
  \end{equation*}
  This is exactly the feasibility criterion for~\eqref{eq:theta_estimator} if~$\lambda = (\lVert \theta \rVert_{\infty} + 1) \err(\delta)$.
\end{proof}

\begin{lemma}[Error on~$\widehat{\theta}$ in max norm] \label{lem:error_max_norm}
  If we select~$\lambda = (\lVert \theta \rVert_{\infty} + 1) \err(\delta)$, then with probability at least~$1-\delta$, the max norm error of~$\widehat{\theta}$ satisfies
  \begin{equation*}
    \lVert \widehat{\theta} - \theta \rVert_{\max}
    \leq 2 \lambda \lVert \Gamma_0^{-1} \rVert_1.
  \end{equation*}
\end{lemma}

\begin{proof}
  \begin{align*}
    \lVert \widehat{\theta} - \theta \rVert_{\max}
     & = \lVert \widehat{\theta} - \Gamma_1 \Gamma_0^{-1} \rVert_{\max}                                                                                                                                                                                                               \\
     & = \lVert (\widehat{\theta} \Gamma_0 - \Gamma_1) \Gamma_0^{-1} \rVert_{\max}                                                                                                                                                                                                    \\
     & = \lVert (\widehat{\theta} \Gamma_0 - \widehat{\theta} \widehat{\Gamma}_0 + \widehat{\theta} \widehat{\Gamma}_0 - \widehat{\Gamma}_1 + \widehat{\Gamma}_1 - \Gamma_1) \Gamma_0^{-1} \rVert_{\max}                                                                              \\
     & \leq \lVert (\widehat{\theta} \Gamma_0 - \widehat{\theta} \widehat{\Gamma}_0) \Gamma_0^{-1} \rVert_{\max} + \lVert (\widehat{\theta} \widehat{\Gamma}_0 - \widehat{\Gamma}_1) \Gamma_0^{-1} \rVert_{\max} + \lVert (\widehat{\Gamma}_1 - \Gamma_1) \Gamma_0^{-1} \rVert_{\max}
  \end{align*}
  By Lemma~\ref{lem:maxnorm_l1linf},
  \begin{align*}
    \lVert \widehat{\theta} - \theta \rVert_{\max}
     & \leq \left(\lVert \widehat{\theta} (\Gamma_0 - \widehat{\Gamma}_0) \rVert_{\max} + \lVert \widehat{\theta} \widehat{\Gamma}_0 - \widehat{\Gamma}_1\rVert_{\max} + \lVert \widehat{\Gamma}_1 - \Gamma_1 \rVert_{\max} \right) \lVert \Gamma_0^{-1} \rVert_1                      \\
     & \leq \left(\lVert \widehat{\theta} \rVert_{\infty} \lVert \Gamma_0 - \widehat{\Gamma}_0 \rVert_{\max} + \lVert \widehat{\theta} \widehat{\Gamma}_0 - \widehat{\Gamma}_1\rVert_{\max} + \lVert \widehat{\Gamma}_1 - \Gamma_1 \rVert_{\max} \right) \lVert \Gamma_0^{-1} \rVert_1
  \end{align*}
  We want to control~$\lVert \widehat{\theta} \rVert_{\infty}$ using~$\lVert \theta \rVert_{\infty}$.
  Let us recall that the operator~$\ell_{\infty}$ norm is equal to the maximum~$\ell_{1}$ norm of the rows of a matrix.
  To control the rows of~$\widehat{\theta}$, we notice that the optimization problem defining~$\widehat{\theta}$, namely
  \begin{equation*}
    \min_{M \in \bbR^{D \times D}}   \lVert \vecm(M) \rVert_1 \quad \text{s.t.} \quad \lVert M \widehat{\Gamma}_0  - \widehat{\Gamma}_1 \rVert_{\max} \leq \lambda
  \end{equation*}
  is equivalent to the row-wise minimization
  \begin{align*}
    \forall i, \quad \min_{M_{i, \cdot} \in \bbR^{1\times D}} \lVert M_{i, \cdot} \rVert_1 \quad \text{s.t.} \quad \lVert M_{i, \cdot} \widehat{\Gamma}_0 - (\widehat{\Gamma}_1)_{i, \cdot} \rVert_{\max} \leq \lambda
  \end{align*}
  From this, we deduce that each row of the optimum~$\widehat{\theta}$ satisfies~$\lVert \widehat{\theta}_{i, \cdot} \rVert_1 \leq \lVert \theta_{i, \cdot} \rVert_1$, which implies~$\lVert \widehat{\theta} \rVert_{\infty} \leq \lVert \theta \rVert_{\infty}$.
  Going back to our error estimate, we get:
  \begin{equation*}
    \lVert \widehat{\theta} - \theta \rVert_{\max}
    \leq \left(\lVert \theta \rVert_{\infty} \lVert \Gamma_0 - \widehat{\Gamma}_0 \rVert_{\max} + \lVert \widehat{\theta} \widehat{\Gamma}_0 - \widehat{\Gamma}_1\rVert_{\max} + \lVert \widehat{\Gamma}_1 - \Gamma_1 \rVert_{\max} \right) \lVert \Gamma_0^{-1} \rVert_1
  \end{equation*}
  Note that the middle term is smaller than~$\lambda$ because the optimum~$\widehat{\theta}$ is a feasible solution.
  Meanwhile, the first and third term are smaller than~$\err(\delta)$ with probability~$1-\delta$:
  \begin{equation*}
    \lVert \widehat{\theta} - \theta \rVert_{\max}
    \leq \left(\lVert \theta \rVert_{\infty} \err(\delta) + \lambda + \err(\delta) \right) \lVert \Gamma_0^{-1} \rVert_1 = 2 \lambda \lVert \Gamma_0^{-1} \rVert_1
  \end{equation*}
\end{proof}

To complete the proof of Theorem~\ref{thm:convergence_rate_theta}, we simply need to go from the max norm to the~$\ell_{\infty}$ operator norm.

\begin{proof}
  Let~$\nu > 0$ be a threshold (to be chosen later). We define
  \begin{equation*}
    s_1 = \max_i \sum_j \min\left\{\frac{|\theta_{i,j}|}{\nu}, 1\right\} \quad \text{and} \quad \mathcal{I}_i = \{j: |\theta_{i,j}| \geq \nu\}
  \end{equation*}
  With high probability, the following holds for any row~$i$:
  \begin{align*}
    \lVert \widehat{\theta}_{i, \cdot} - \theta_{i, \cdot} \rVert_1
     & \leq \lVert \widehat{\theta}_{i, \mathcal{I}_i^c} - \theta_{i, \mathcal{I}_i^c} \rVert_1 + \lVert \widehat{\theta}_{i, \mathcal{I}_i} - \theta_{i, \mathcal{I}_i} \rVert_1                                                                                                      \\
     & \leq \lVert \widehat{\theta}_{i, \mathcal{I}_i^c} \rVert_1 + \lVert \theta_{i, \mathcal{I}_i^c} \rVert_1 + \lVert \widehat{\theta}_{i, \mathcal{I}_i} - \theta_{i, \mathcal{I}_i} \rVert_1                                                                                      \\
     & = (\lVert \widehat{\theta}_{i, \cdot} \rVert_1 - \lVert \widehat{\theta}_{i, \mathcal{I}_i} \rVert_1) + \lVert \theta_{i, \mathcal{I}_i^c} \rVert_1 + \lVert \widehat{\theta}_{i, \mathcal{I}_i} - \theta_{i, \mathcal{I}_i} \rVert_1                                           \\
     & \leq \lVert \theta_{i, \cdot} \rVert_1 - \lVert \widehat{\theta}_{i, \mathcal{I}_i} \rVert_1 + \lVert \theta_{i, \mathcal{I}_i^c} \rVert_1 + \lVert \widehat{\theta}_{i, \mathcal{I}_i} - \theta_{i, \mathcal{I}_i} \rVert_1                                                    \\
     & = (\lVert \theta_{i,\mathcal{I}_i} \rVert_1 + \lVert \theta_{i,\mathcal{I}_i^c} \rVert_1) - \lVert \widehat{\theta}_{i, \mathcal{I}_i} \rVert_1 + \lVert \theta_{i, \mathcal{I}_i^c} \rVert_1 + \lVert \widehat{\theta}_{i, \mathcal{I}_i} - \theta_{i, \mathcal{I}_i} \rVert_1 \\
     & = 2 \lVert \theta_{i, \mathcal{I}_i^c} \rVert_1 + (\lVert \theta_{i, \mathcal{I}_i} \rVert_1 - \lVert \widehat{\theta}_{i, \mathcal{I}_i} \rVert_1) + \lVert \widehat{\theta}_{i, \mathcal{I}_i} - \theta_{i, \mathcal{I}_i} \rVert_1                                           \\
     & \leq 2 \lVert \theta_{i, \mathcal{I}_i^c} \rVert_1 + 2 \lVert \widehat{\theta}_{i, \mathcal{I}_i} - \theta_{i, \mathcal{I}_i} \rVert_1
  \end{align*}
  By definition of~$\mathcal{I}_i$, for all~$j \in \mathcal{I}_i^c$,~$|\theta_{i,j}| \leq \nu$, hence
  \begin{equation*}
    \lVert \theta_{i, \mathcal{I}_i^c} \rVert_1 = \sum_{j \in \mathcal{I}_i^c} |\theta_{i,j}| = \sum_{j \in \mathcal{I}_i^c} \min\{|\theta_{i,j}|, \nu\} \leq \sum_{j} \min\{|\theta_{i,j}|, \nu\} \leq \nu s_1
  \end{equation*}
  Meanwhile, the second term satisfies
  \begin{equation*}
    \lVert \widehat{\theta}_{i, \mathcal{I}_i} - \theta_{i, \mathcal{I}_i} \rVert_1
    \leq |\mathcal{I}_i| \times \lVert \widehat{\theta} - \theta \rVert_{\max}
  \end{equation*}
  And by definition of~$\mathcal{I}_i$, for all~$j \in \mathcal{I}_i$,~$|\theta_{i,j}| \geq \nu$, hence
  \begin{equation*}
    |\mathcal{I}_i| = \sum_{j \in \mathcal{I}_i} 1 = \sum_{j \in \mathcal{I}_i} \min \left\{ \frac{|\theta_{i,j}|}{\nu}, 1\right\} \leq \sum_j \min \left\{ \frac{|\theta_{i,j}|}{\nu}, 1\right\} \leq s_1
  \end{equation*}
  Combining all of this, we get that with high probability,
  \begin{equation*}
    \lVert \widehat{\theta}_{i, \cdot} - \theta_{i, \cdot} \rVert_1 \leq 2 (\nu + 2\lambda \lVert \Gamma_0^{-1} \rVert_1) s_1
  \end{equation*}
  Judging by the last Equation, it makes sense to choose~$\nu = 2 \lambda \lVert \Gamma_0^{-1} \rVert_1$.
  Furthermore, our sparsity hypothesis on~$\theta$ implies that for all but~$s$ of the coefficients of any row~$i$,~$\min\{|\theta_{i,j}|, \nu\} = |\theta_{i,j}| = 0$.
  We deduce that for every~$i$,
  \begin{equation*}
    \sum_j \min\left\{|\theta_{i,j}|, \nu \right\} \leq s \max_j \min\left\{|\theta_{i,j}|, \nu \right\} \leq \nu s
  \end{equation*}
  which directly implies
  \begin{equation*}
    \nu s_1 = \max_i \sum_j \min\left\{|\theta_{i,j}|, \nu \right\} \leq \nu s
  \end{equation*}
  We finally find that with high probability,
  \begin{equation*}
    \lVert \widehat{\theta}_{i, \cdot} - \theta_{i, \cdot} \rVert_1 \leq 4 \nu s_1 \leq 4 \nu s = 8 \lambda \lVert \Gamma_0^{-1} \rVert_1 s
  \end{equation*}
  With the help of a union bound, again with high probability,
  \begin{equation*}
    \lVert \widehat{\theta} - \theta \rVert_{\infty} = \max_i \lVert \widehat{\theta}_{i, \cdot} - \theta_{i, \cdot} \rVert_1 \leq 8 \lambda \lVert \Gamma_0^{-1} \rVert_1 s
  \end{equation*}
  We substitute the value of~$\lambda$ and obtain
  \begin{equation*}
    \lVert \widehat{\theta} - \theta \rVert_{\infty}
    \leq 8 (\lVert \theta \rVert_{\infty} + 1) \err(\delta) \lVert \Gamma_0^{-1} \rVert_1 s
  \end{equation*}
  Once we plug in the value of~$\err(\delta)$, the resulting high-probability error bound reads
  \begin{equation*}
    \lVert \widehat{\theta} - \theta \rVert_{\infty} \leq c
    \frac{\lVert \theta \rVert_{\infty} + 1}{\lVert \Gamma_0^{-1} \rVert_1^{-1}}
    \frac{\sigma_{\max}^2 + \omega^2}{(1-\thetamax)^2}
    \frac{s \sqrt{\log(D/\delta)}}{\sqrt{T p q_u}}
  \end{equation*}
  Since~$\thetamax$ only acted as an upper bound on~$\lVert \theta \rVert_2$ in this proof, we can define
  \begin{equation*}
    \gamma_u(\theta) = \frac{\lVert \theta \rVert_{\infty} + 1}{(1-\lVert \theta \rVert_2)^2} \frac{\sigma_{\max}^2 + \omega^2}{\lVert \Gamma_0^{-1} \rVert_1^{-1}}
  \end{equation*}
  to obtain the compressed expression
  \begin{equation*}
    \lVert \widehat{\theta} - \theta \rVert_{\infty} \leq c \gamma_u(\theta) \frac{s \sqrt{\log(D/\delta)}}{\sqrt{T p q_u}}.
  \end{equation*}
\end{proof}

\section{Proof of the Minimax Lower Bound} \label{sec:lower_bound_proof}

We now present the detailed proof of Theorem~\ref{thm:lower_bound_sparse}.

\subsection{Overview}

Our argument is based on Fano's method, which we sum up in Lemma~\ref{lem:fano_method}. For a detailed presentation, we refer the reader to \citet[Chapter 2]{tsybakovIntroductionNonparametricEstimation2008}.
Note that \citet[Chapter 15]{wainwrightHighDimensionalStatisticsNonAsymptotic2019} and \citet[Chapter 7]{duchiInformationTheoryStatistics2019} also offer good treatments of the subject.

Fano's method relies on choosing a set of parameters~$\theta_0, \theta_1, ..., \theta_M$ satisfying two seemingly contradictory conditions: their induced distributions must be hard to distinguish, yet they must lie as fart apart from one another as possible.
In particular, the crucial requirement of Fano's method is a tight upper bound on the KL divergence between two distributions generated by different parameters~$\theta_i$ and~$\theta_0$.
Taking the latter to be~$0$, we actually want to bound
\begin{equation*}
  \frac{1}{M+1} \sum_{i=1}^{M} \KL{\bbP_{\theta_i}(\Pi, Y)}{\bbP_{0}(\Pi, Y)}  \leq \max_i \KL{\bbP_{\theta_i}(\Pi, Y)}{\bbP_{0}(\Pi, Y)}
\end{equation*}
By Lemma~\ref{lem:kl_chain},
\begin{equation*}
  \KL{\bbP_{\theta_i}(\Pi, Y)}{\bbP_{0}(\Pi, Y)} = \KL{\bbP_{\theta_i}(\Pi)}{\bbP_{0}(\Pi)} + \bbE_\Pi \left[\KL{\bbP_{\theta_i}(Y|\Pi)}{\bbP_{0}(Y|\Pi)}\right]
\end{equation*}
Since~$\theta_i$ does not affect the distribution of the sampling process~$\Pi$, the first term of the right-hand side is zero, and we will concentrate on the second term.
First, we will upper-bound the random variable inside the expectation for a fixed realization of~$\Pi$, and then we will average said bound over all possible projections.

\medskip

We now give the structure of the argument in a coherent order, along with the most important intermediate results:
\begin{enumerate}
  \item Compute the conditional covariance~$\Cov_\theta [Y | \Pi]$ and decompose it into a constant term~$Q_\Pi$ (corresponding to the independent case~$\theta=0$) plus a residual~$R_\Pi(\theta)$ (Lemma~\ref{lem:y_covariance_decomposition}).
  \item Upper-bound the conditional KL divergence~$\KL{\bbP_{\theta}(Y|\Pi)}{\bbP_{0}(Y|\Pi)}$ using the \enquote{deviations from the identity}~$\Delta_\Pi(\theta) = Q_\Pi^{-1/2} R_\Pi(\theta) Q_\Pi^{-1/2}$ (Lemma~\ref{lem:kl_close_gaussians_applied}).
  \item Control~$\Delta_\Pi(\theta)$ using features of~$R(\theta)$ scaled by sampling-related factors (Lemmas~\ref{lem:delta_to_R},~\ref{lem:submatrix_averaging} and~\ref{lem:full_R_control}).
  \item Deduce an upper bound on the KL divergence~$\bbE_\Pi[\KL{\bbP_{\theta}(Y|\Pi)}{\bbP_{0}(Y|\Pi)}]$ (Lemma~\ref{lem:kl_bound_avg}).
  \item Apply Fano's method to a set of parameters~$\theta_i$ constructed from a pruned binary hypercube of well-chosen radius.

\end{enumerate}

\subsection{Change of Notations}

For this part, we slightly modify the previous conventions: we now assume that all the rows of~$\Pi_t$ that contain only zeros are removed.
In other words,~$\Pi_t$ is no longer the diagonal matrix~$\diag(\pi_t)$ but instead becomes a wide rectangular matrix with exactly one~$1$ per row and at most one~$1$ per column.
We thus have~$\Pi \Pi' = I$ unless all of the~$\pi_{t, d}$ are zero, in which case the matrix~$\Pi$ is empty, and so are the observations~$Y$.
Let us denote this very unlikely event by~$E$, and its complement by~$E^c$.
If~$\Pi$ is such that~$E$ happens, we obviously have~$\KL{\bbP_{\theta_i}(Y|\Pi)}{\bbP_{0}(Y|\Pi)} = 0$, which means that
\begin{equation} \label{eq:effect_E}
  \bbE_\Pi \left[\KL{\bbP_{\theta_i}(Y|\Pi)}{\bbP_{0}(Y|\Pi)}\right] = \bbE_\Pi \left[·\one_{E^c} \KL{\bbP_{\theta_i}(Y|\Pi)}{\bbP_{0}(Y|\Pi)}\right]
\end{equation}
For the beginning of the proof, we consider a fixed, non-empty realization of~$\Pi$.

\subsection{Covariance Decomposition}

As we announced in the proof sketch, our reference parameter will be~$\theta_0 = 0$, which is why it makes sense to express the conditional covariance of~$Y$ as a deviation from the case without interactions.
This is the aim of the following result.

\begin{lemma}[Conditional covariance decomposition] \label{lem:y_covariance_decomposition}
  The covariance matrix of~$Y$ given~$\Pi$ decomposes as
  \begin{equation*}
    \Cov_\theta [Y | \Pi] = Q_\Pi + R_\Pi(\theta),
  \end{equation*}
  where~$Q_\Pi$ is a constant term and~$R_\Pi(\theta)$ is a residual which vanishes as~$\theta \to 0$.
  They are defined as follows: the constant term is
  \begin{equation*}
    Q_\Pi = \Pi (\bdiag_T \Sigma) \Pi' + \omega^2 I
  \end{equation*}
  whereas the residual equals
  \begin{equation*}
    R_\Pi(\theta) = \Pi R(\theta) \Pi' \qquad \text{with} \qquad  R(\theta) = \begin{bmatrix}
      \theta \Gamma_0(\theta) \theta' & \Gamma_0(\theta) \theta'^1      & \Gamma_0(\theta) \theta'^2      & \cdots \\
      \theta^1 \Gamma_0(\theta)       & \theta \Gamma_0(\theta) \theta' & \Gamma_0(\theta) \theta'^1      &        \\
      \theta^2 \Gamma_0(\theta)       & \theta^1 \Gamma_0(\theta)       & \theta \Gamma_0(\theta) \theta' &        \\
      \vdots                          &                                 &                                 & \ddots
    \end{bmatrix}.
  \end{equation*}
\end{lemma}

\begin{proof}
  We use Equation~\eqref{eq:y_model} to see that the conditional distribution~$\bbP_{\theta} (Y | \Pi)$ is a centered multivariate Gaussian with covariance
  \begin{equation*}
    \Cov_\theta [Y | \Pi] = \omega^2 I + \Pi \Cov_\theta[X] \Pi'.
  \end{equation*}
  We then use Lemma~\ref{lem:x_covariance} to get an expression of~$\Cov_{\theta}[X]$ and deduce that its constant term (w.r.t to~$\theta$) is a block-diagonal matrix filled with copies of~$\Sigma$:
  \begin{equation*}
    \Cov_\theta [Y | \Pi] = \omega^2 I + \Pi \bdiag_T(\Sigma) \Pi' + \Pi \left(\Cov_\theta[X] - \bdiag_T(\Sigma)\right) \Pi'.
  \end{equation*}
  Finally, we define~$Q_\Pi = \omega^2 I + \Pi \bdiag_T(\Sigma) \Pi'$,~$R(\theta) = \Cov_\theta[X] - \bdiag_T (\Sigma)$ and~$R_\Pi(\theta) = \Pi R(\theta)\Pi'$ to obtain the decomposition we announced.
  The diagonal blocks of~$R(\theta)$ are easily computed by noticing that~$\Gamma_0(\theta) - \Sigma = \theta \Gamma_0(\theta) \theta'$.
\end{proof}

\subsection{From the KL Divergence to \texorpdfstring{$\Delta_\Pi(\theta)$}{Delta}}

Judging by Lemma~\ref{lem:y_covariance_decomposition}, choosing a parameter~$\theta$ close to~$0$ yields a conditional distribution for~$Y$ whose covariance is close to~$Q_\Pi$.
In the next result, we translate this into a bound on the KL divergence between~$\bbP_{\theta}(Y | \Pi)$ and~$\bbP_{0}(Y | \Pi)$.

\begin{lemma} \label{lem:kl_close_gaussians_applied}
  Let us define the deviation from the identity:
  \begin{equation*}
    \Delta_\Pi(\theta) = Q_\Pi^{-1/2} R_\Pi(\theta) Q_\Pi^{-1/2}.
  \end{equation*}
  Then the conditional KL divergence is upper-bounded by:
  \begin{equation*}
    \KL{\bbP_{\theta}(Y | \Pi)}{\bbP_{0}(Y | \Pi)} \leq \frac{\lVert \Delta_\Pi(\theta) \rVert_F^2}{2(1 + \lambda_{\min}(\Delta_\Pi(\theta)))}.
  \end{equation*}
\end{lemma}

\begin{proof}
  The conditional KL divergence~$\KL{\bbP_{\theta} (Y | \Pi)}{\bbP_0(Y | \Pi)}$ can be bounded using Lemma~\ref{lem:kl_around_id}.
  Indeed, both conditional distributions are Gaussian and have the same expectation, and covariance matrices that are \enquote{close} in the following sense: by Lemma~\ref{lem:y_covariance_decomposition},
  \begin{align*}
    \Cov_{0}(Y | \Pi)
     & = Q_\Pi     = Q_\Pi^{1/2} (Q_\Pi^{1/2})'                                                                                                         \\
    \Cov_{\theta}(Y | \Pi)
     & = Q_\Pi + R_\Pi(\theta) = Q_\Pi^{1/2} \Big(I + \underbrace{Q_\Pi^{-1/2} R_\Pi(\theta) Q_\Pi^{-1/2}}_{\Delta_{\Pi}(\theta)} \Big) (Q_\Pi^{1/2})'.
  \end{align*}
  By Lemma~\ref{lem:ostrowski}, there exists a real number~$r_{\min} \geq s_{\min}\left(Q_\Pi^{1/2}\right)^2 = s_{\min}(Q_\Pi)$ such that
  \begin{equation*}
    \lambda_{\min}(\Cov_\theta(Y |\Pi)) = r_{\min} \lambda_{\min}(I + \Delta_\Pi(\theta)).
  \end{equation*}
  Since~$Q_\Pi \succeq \omega^2 I \succ 0$, its minimum singular value satisfies~$s_{\min}(Q_{\Pi}) > 0$, so that~$r_{\min} > 0$.
  In addition,~$\Cov_\theta(Y |\Pi) \succeq \omega^2 I \succ 0$, so that~$\lambda_{\min}(\Cov_\theta(Y |\Pi)) > 0$.
  Therefore,
  \begin{equation*}
    \lambda_{\min}(I + \Delta_\Pi(\theta)) = \frac{\lambda_{\min}(\Cov_\theta(Y |\Pi))}{r_{\min}} > 0 \quad \text{and} \quad \lambda_{\min}(\Delta_\Pi(\theta)) > -1,
  \end{equation*}
  which means we can apply Lemma~\ref{lem:kl_around_id} with~$\bbP_1 = \bbP_{\theta} (Y | \Pi)$ and~$\bbP_0 = \bbP_0(Y | \Pi)$.
\end{proof}

\subsection{From \texorpdfstring{$\Delta_\Pi(\theta)$}{Delta} to \texorpdfstring{$R_\Pi(\theta)$}{the Projection of R(theta)}}

Lemma~\ref{lem:kl_close_gaussians_applied} strongly suggests studying a certain fraction involving~$\Delta_\Pi(\theta)$.
In the following result, we boil it down to a function of the residual term~$R_\Pi (\theta)$.

\begin{lemma} \label{lem:delta_to_R}
  Assume~$\lVert R(\theta) \rVert_2 \leq (\sigma_{\min}^2 + \omega^2) / 2$. We have the following upper bound:
  \begin{equation*}
    \frac{\lVert \Delta_\Pi(\theta) \rVert_F^2}{2(1 + \lambda_{\min}(\Delta_\Pi(\theta)))}
    \leq \frac{\lVert R_\Pi(\theta) \rVert_F^2}{\left(\sigma_{\min}^2 + \omega^2\right)^2}.
  \end{equation*}
\end{lemma}

\begin{proof}
  Since the quantity~$\lambda_{\min}(\Delta_\Pi(\theta))$ in the denominator is hard to control, we will work with the spectral norm instead.
  Indeed, whenever~$\lVert \Delta_\Pi(\theta) \rVert_2 < 1$, we have the crude bound
  \begin{equation*}
    \frac{1}{1-\lambda_{\min}(\Delta_\Pi(\theta))} \leq \frac{1}{1-\lVert \Delta_\Pi(\theta) \rVert_2}.
  \end{equation*}
  Let us start by noticing that, thanks to Lemma~\ref{lem:frobenius_spectral_product},
  \begin{align*}
    \lVert \Delta_\Pi(\theta) \rVert_F^2
     & = \lVert Q_\Pi^{-1/2} R_\Pi(\theta) Q_\Pi^{-1/2} \rVert_F^2 \leq \lVert Q_\Pi^{-1/2} \rVert_2^4 \lVert R_\Pi(\theta) \rVert_F^2 = \lVert Q_\Pi^{-1} \rVert_2^2 \lVert R_\Pi(\theta) \rVert_F^2 \\
    \lVert \Delta_\Pi(\theta) \rVert_2
     & = \lVert Q_\Pi^{-1/2} \Pi R(\theta) \Pi' Q_\Pi^{-1/2} \rVert_2 \leq \lVert Q_\Pi^{-1/2} \Pi \rVert_2^2 \lVert R(\theta) \rVert_2.
  \end{align*}
  We will later see how the spectral and Frobenius norms of the residual~$R(\theta)$ can be controlled as a function of~$\theta$.
  For now, we must work to upper bound~$\lVert Q_\Pi^{-1} \rVert_2$ and~$\lVert Q_\Pi^{-1} \Pi \rVert_2^2$.

  To simplify the following proof, we write~$\Sigmad = \bdiag_T \Sigma$.
  Since~$\Sigmad$ is block-diagonal, its spectrum is the same as the spectrum of~$\Sigma$ repeated~$T$ times, hence~$\lambda_{\min}(\Sigmad) = \sigma_{\min}^2$.
  And since we assumed~$E^c$ happens (non empty projection), we have~$\Pi \Pi' = I$ and~$\Pi' \Pi = \diag(\pi)$, which has at least one entry equal to~$1$.

  \medskip

  We start with~$\lVert Q_\Pi^{-1} \rVert_2$.
  Since~$Q_\Pi \succeq \omega^2 I \succ 0$ is non-singular and symmetric,
  \begin{equation*}
    \lVert Q_\Pi^{-1} \rVert_2 = \lambda_{\max}(Q_{\Pi}^{-1}) = \frac{1}{\lambda_{\min}(Q_\Pi)} = \frac{1}{\lambda_{\min}(\Pi \Sigmad \Pi' + \omega^2 I)}.
  \end{equation*}
  Remembering that~$\Sigmad \succeq \sigma_{\min}^2 I$, we get
  \begin{equation*}
    \Pi \Sigmad \Pi' + \omega^2 I \succeq \sigma_{\min}^2 \Pi \Pi' + \omega^2 I = (\sigma_{\min}^2 + \omega^2) I
  \end{equation*}
  and thus
  \begin{equation*}
    \lVert Q_\Pi^{-1} \rVert_2
    \leq \frac{1}{\sigma_{\min}^2 + \omega^2}.
  \end{equation*}
  We now continue with~$\lVert Q_\Pi^{-1/2} \Pi \rVert_2^2$.
  By definition of the spectral norm,
  \begin{equation*}
    \lVert Q_\Pi^{-1/2} \Pi \rVert_2^2
    = \lambda_{\max} \left( \Pi' Q_\Pi^{-1} \Pi \right) = \lambda_{\max} \left( \Pi' (\Pi \Sigmad \Pi' + \omega^2 I)^{-1} \Pi \right).
  \end{equation*}
  Because matrix inversion is decreasing w.r.t. the Loewner order on positive semi-definite matrices,
  \begin{align*}
    (\Pi \Sigmad \Pi' + \omega^2 I)^{-1}          & \preceq (\sigma_{\min}^2 + \omega^2)^{-1} I^{-1}       \\
    \Pi' (\Pi \Sigmad \Pi' + \omega^2 I)^{-1} \Pi & \preceq \frac{1}{\sigma_{\min}^2 + \omega^2} \Pi' \Pi.
  \end{align*}
  It follows that
  \begin{equation*}
    \lVert Q_\Pi^{-1/2} \Pi \rVert_2^2 \leq \frac{1}{\sigma_{\min}^2 + \omega^2} \lambda_{\max} (\Pi' \Pi) =  \frac{1}{\sigma_{\min}^2 + \omega^2}.
  \end{equation*}
  The conclusion is within reach:
  \begin{align*}
    \frac{\lVert \Delta_\Pi(\theta) \rVert_F^2}{1 + \lambda_{\min}(\Delta_\Pi(\theta))}
     & \leq \frac{\lVert \Delta_\Pi(\theta) \rVert_F^2}{1 - \lVert \Delta_\Pi(\theta) \rVert_2} \leq \frac{\lVert Q_\Pi^{-1} \rVert_2^2 \lVert R_\Pi(\theta) \rVert_F^2}{1 - \lVert Q_\Pi^{-1/2} \Pi \rVert_2^2 \lVert R(\theta) \rVert_2}               \\
     & \leq \frac{\left(\frac{1}{\sigma_{\min}^2 + \omega^2}\right)^2 \lVert R_\Pi(\theta) \rVert_F^2}{1 - \frac{1}{\sigma_{\min}^2 + \omega^2} \lVert R(\theta) \rVert_2} \leq \frac{2 \lVert R_\Pi(\theta) \rVert_F^2}{(\sigma_{\min}^2 + \omega^2)^2}
  \end{align*}
  The last inequality is justified by our assumption~$\lVert R(\theta) \rVert_2 \leq (\sigma_{\min}^2 + \omega^2) / 2$.
  Another consequence of this assumption is that
  \begin{equation*}
    \lVert \Delta_\Pi(\theta) \rVert_2
    \leq \lVert Q_\Pi^{-1/2} \Pi \rVert_2^2 \lVert R(\theta) \rVert_2
    \leq \frac{1}{\sigma_{\min}^2 + \omega^2} \frac{\sigma_{\min}^2 + \omega^2}{2} = \frac{1}{2} < 1
  \end{equation*}
  which is sufficient for the first inequality to hold.
\end{proof}

\subsection{From \texorpdfstring{$R_\Pi(\theta)$}{the Projection of R(theta)} to \texorpdfstring{$R(\theta)$}{R(theta)}}

As the previous Lemma underlines, the last step we need to get rid of the dependency in~$\Pi$ is to study the average norm of~$R_\Pi(\theta)$.

\begin{lemma} \label{lem:submatrix_averaging}
  Let~$q_\ell = \max\{1-b, 2p-(1-b)\}$. Then
  \begin{equation*}
    \bbE_\Pi \left[ \one_{E^c} \lVert R_\Pi(\theta) \rVert_F^2 \right] \leq p \Tr[R(\theta) \odot R(\theta)] + pq_\ell \lVert R(\theta) \rVert_F^2.
  \end{equation*}
\end{lemma}

\begin{proof}
  We first notice that for any matrix~$A$,
  \begin{align*}
    \bbE_\Pi \big[\one_{E^c} \lVert \Pi A \Pi' \rVert_F^2 \big]
     & = \bbE_{\Pi} \left[ \one_{E^c} \Tr \big[\Pi A \Pi' \Pi A' \Pi' \big] \right] \\
     & = \bbE_\Pi \left[ \Tr \big[ \diag(\pi) A \diag(\pi) A' \big] \right]         \\
     & = \sum_{i,j} \bbE_\Pi [\pi_i \pi_j] A_{i,j}^2.
  \end{align*}
  We can apply this to~$R_\Pi(\theta) = \Pi R(\theta) \Pi'$:
  \begin{equation*}
    \bbE_\Pi \big[ \one_{E^c} \lVert R_\Pi(\theta) \rVert_F^2 \big] = \sum_{i,j} \bbE_\Pi [\pi_i \pi_j] R(\theta)_{i,j}^2.
  \end{equation*}
  The rest of the proof consists in plugging in the moments~$\bbE_\Pi [\pi_i \pi_j]$ from Lemma~\ref{lem:proj_moments}:
  \begin{align*}
    \bbE_\Pi \big[ \lVert R_\Pi(\theta) \rVert_F^2 \big]
     & = \sum_{\substack{t_1,t_2,d_1,d_2            \\(t_1,d_1)=(t_2,d_2)}} p R(\theta)_{(t_1,d_1),(t_2,d_2)}^2 \\
     & \phantomeq + \sum_{\substack{t_1,t_2,d_1,d_2 \\d_1\neq d_2}} p^2 R(\theta)_{(t_1,d_1),(t_2,d_2)}^2                               \\
     & \phantomeq + \sum_{\substack{t_1,t_2,d_1,d_2 \\d_1=d_2,t_1 \neq t_2}} (p^2 + p(1-p)(1-a-b)^{|t_1 - t_2|}) R(\theta)_{(t_1,d_1),(t_2,d_2)}^2.
  \end{align*}
  The sum in the last term can be crudely controlled as follows:
  \begin{align*}
    \sum_{\substack{t_1,t_2,d                                                    \\t_1 \neq t_2}} (1-a-b)^{|t_1 - t_2|} R(\theta)_{(t_1,d),(t_2,d)}^2
     & \leq |1-a-b| \sum_{\substack{t_1,t_2                                      \\t_1 \neq t_2}} \sum_d (R(\theta)_{[t_1, t_2]})_{d,d}^2     \\
     & \leq |1-a-b| \sum_{t_1 \neq t_2} \lVert R(\theta)_{[t_1, t_2]} \rVert_F^2 \\
     & \leq |1-a-b| \cdot \lVert R(\theta) \rVert_F^2
  \end{align*}
  This yields a short, but probably suboptimal bound:
  \begin{equation*}
    \bbE_\Pi \big[ \one_{E^c} \lVert R_\Pi(\theta) \rVert_F^2 \big] \leq p \Tr[R(\theta) \odot R(\theta)] + (p^2 + p(1-p) |1-a-b|) \lVert R(\theta) \rVert_F^2.
  \end{equation*}
  In the previous part, we already saw that
  \begin{equation*}
    p + (1-p)(1 - a -  b) = 1-b.
  \end{equation*}
  Similarly, we obtain
  \begin{align*}
    p + (1-p)(a+b-1) & = \frac{a}{a+b} + \frac{b}{a+b}(a+b-1)                        \\
                     & = \frac{a + ba + b^2 - b}{a+b}  = \frac{a(1+b) - b(1-b)}{a+b} \\
                     & = p(1+b) - (1-p)(1-b)  = 2p - (1-b).
  \end{align*}
  As a consequence,
  \begin{equation*}
    p + (1-p)|1-a-b| = \max\{1-b, 2p-(1-b)\} = q_\ell,
  \end{equation*}
  which yields the expected result.
\end{proof}

\subsection{Bounding \texorpdfstring{$R(\theta)$}{R(theta)}}

Lemma~\ref{lem:submatrix_averaging} relates the bounds involving~$R_\Pi(\theta)$ to features of the full residual~$R(\theta)$, which we now study.

\begin{lemma} \label{lem:full_R_control}
  The residual~$R(\theta)$ satisfies the following inequalities:
  \begin{align*}
    \lVert R(\theta) \rVert_2       & \leq \frac{2 \sigma_{\max}^2}{(1 - \thetamax)^2} \lVert \theta \rVert_2                             \\
    \lVert R(\theta) \rVert_F^2     & \leq \frac{2T \sigma_{\max}^4}{(1 - \thetamax)^3} \lVert \theta \rVert_F^2                          \\
    \Tr [R(\theta) \odot R(\theta)] & \leq \frac{T \sigma_{\max}^4}{(1 - \thetamax)^2} \lVert \theta \rVert_2^2 \lVert \theta \rVert_F^2.
  \end{align*}
\end{lemma}

\begin{proof}
  We start by giving a formula for the blocks of~$R(\theta)$: by Lemma~\ref{lem:y_covariance_decomposition},
  \begin{equation*}
    R(\theta)_{[t, s]} = \begin{cases}
      \theta^{t-s} \Gamma_0(\theta)   & \text{if~$s \in [1, t-1]$}  \\
      \theta \Gamma_0(\theta) \theta' & \text{if~$s = t$}           \\
      \Gamma_0(\theta) \theta'^{t-s}  & \text{if~$s \in [t+1, T]$}.
    \end{cases}
  \end{equation*}
  These individual blocks can be bounded using Lemmas~\ref{lem:frobenius_spectral_product} and~\ref{lem:norm_gamma}: if~$r \geq 1$, then
  \begin{align*}
    \lVert \theta^r \Gamma_0(\theta) \rVert_F^2
     & \leq \lVert \Gamma_0(\theta) \rVert_2^2 \lVert \theta^r \rVert_F^2 \leq \lVert \Gamma_0(\theta) \rVert_2^2 \lVert \theta \rVert_F^2 \lVert \theta^{r-1} \rVert_2^2 \leq \frac{\sigma_{\max}^4}{(1 - \thetamax)^2} \lVert \theta \rVert_F^2 \lVert \theta \rVert_2^{2(r-1)} \\
    \lVert \Gamma_0(\theta) \theta'^r  \rVert_F^2
     & \leq \lVert \Gamma_0(\theta) \rVert_2^2 \lVert \theta^r \rVert_F^2 \leq \lVert \Gamma_0(\theta) \rVert_2^2 \lVert \theta \rVert_F^2 \lVert \theta^{r-1} \rVert_2^2 \leq \frac{\sigma_{\max}^4}{(1 - \thetamax)^2} \lVert \theta \rVert_F^2 \lVert \theta \rVert_2^{2(r-1)} \\
    \lVert \theta \Gamma_0(\theta) \theta' \rVert_F^2
     & \leq \lVert \theta \rVert_2^2 \lVert \Gamma_0(\theta) \rVert_2^2 \lVert \theta \rVert_F^2 \leq \frac{\sigma_{\max}^4}{(1 - \thetamax)^2} \lVert \theta \rVert_F^2 \lVert \theta \rVert_2^2.
  \end{align*}
  Since we control the norm of each block of~$R(\theta)$, we control the norm of the whole matrix:
  \begin{align*}
    \lVert R(\theta) \rVert_F^2
     & = \sum_{t=1}^{T} \left( \sum_{s = 1}^{t-1}{\lVert \theta^{t-s} \Gamma_0(\theta) \rVert_F^2} + \lVert \theta \Gamma_0(\theta) \theta \rVert_F^2 + \sum_{s = t+1}^{T}{\lVert  \Gamma_0(\theta)\theta^{s-t} \rVert_F^2} \right)                                 \\
     & \leq \frac{\sigma_{\max}^4 \lVert \theta \rVert_F^2}{(1-\thetamax^2)^2} \sum_{t=1}^{T} \left( \sum_{s = 1}^{t-1}{\lVert \theta \rVert_2^{2(t-s-1)}} + \lVert \theta \rVert_2^{2} + \sum_{s = t+1}^{T}{\lVert \theta \rVert_2^{2(s-t-1)}} \right)             \\
     & \leq \frac{\sigma_{\max}^4 \lVert \theta \rVert_F^2}{(1-\thetamax^2)^2} \sum_{t=1}^{T} \left( \sum_{s = -\infty}^{t-1}{\lVert \theta \rVert_2^{2(t-1-s)}} + \lVert \theta \rVert_2^{2} + \sum_{s = t+1}^{+\infty}{\lVert \theta \rVert_2^{2(s-1-t)}} \right) \\
     & = \frac{\sigma_{\max}^4 \lVert \theta \rVert_F^2}{(1-\thetamax^2)^2} T \left( \frac{1}{1 - \lVert \theta \rVert_2^2} + \lVert \theta \rVert_2^2 + \frac{1}{1 - \lVert \theta \rVert_2^2} \right)
  \end{align*}
  We now remember our hypothesis~$\lVert \theta \rVert_2 \leq \thetamax < 1$:
  \begin{align*}
    \lVert R(\theta) \rVert_F^2
     & \leq \frac{\sigma_{\max}^4 \lVert \theta \rVert_F^2}{(1-\thetamax^2)^2} T \left( \frac{1}{1 - \thetamax^2} + \thetamax^2 + \frac{1}{1 - \thetamax^2} \right)                                                                                      \\
     & = \frac{\sigma_{\max}^4 \lVert \theta \rVert_F^2}{(1-\thetamax^2)^2} T \left( \frac{2 + \thetamax^2(1 - \thetamax^2)}{1 - \thetamax^2}  \right)                                                                                                   \\
     & \leq \frac{\sigma_{\max}^4 \lVert \theta \rVert_F^2}{(1- \thetamax^2)^2} T \left( \frac{2 + 2\thetamax}{1 - \thetamax^2}  \right) = \frac{\sigma_{\max}^4 \lVert \theta \rVert_F^2}{(1- \thetamax^2)^2} T \left( \frac{2}{1 - \thetamax}  \right) \\
     & = 2T \frac{\sigma_{\max}^4 \lVert \theta \rVert_F^2}{(1 - \thetamax)^3}.
  \end{align*}
  Now that we have a handle on the Frobenius norm of~$R(\theta)$, we move on to its spectral norm.
  Notice that~$R(\theta)$ can be written as a sum of Kronecker products with the subdiagonal matrices~$\Jmat_t$:
  \begin{equation*}
    R(\theta) = I \otimes \theta \Gamma_0(\theta) \theta' + \sum_{t=1}^{T-1} \left[ \Jmat_t \otimes \theta^t \Gamma_0(\theta) + \Jmat_t' \otimes \Gamma_0(\theta) \theta'^t \right].
  \end{equation*}
  We can use Lemma~\ref{lem:kron_singular_values} and write:
  \begin{align*}
    \lVert R(\theta) \rVert_2
     & \leq  \lVert I \rVert_2 \times \lVert \theta \Gamma_0(\theta) \theta' \rVert_2 + \sum_{t=1}^{T-1} \left[ \lVert \Jmat_t \rVert_2 \times \lVert \theta^t \Gamma_0(\theta) \rVert_2 + \lVert \Jmat_t' \rVert_2 \times \lVert \Gamma_0(\theta) \theta'^t \rVert_2 \right]    \\
     & \leq \lVert \Gamma_0(\theta) \rVert_2 \left( \lVert \theta \rVert_2^2 + 2\sum_{t=1}^{T-1} \lVert \theta \rVert_2^t \right) \leq \frac{\sigma_{\max}^2}{1-\thetamax^2} \left(\lVert \theta \rVert_2^2 + 2\frac{\lVert \theta \rVert_2}{1 - \lVert \theta \rVert_2} \right) \\
     & \leq \frac{ \sigma_{\max}^2 \lVert \theta \rVert_2}{1 - \thetamax^2} \left(\thetamax + \frac{2}{1-\thetamax} \right) \leq \frac{ \sigma_{\max}^2 \lVert \theta \rVert_2}{1 - \thetamax} \left(\frac{2 + 2\thetamax}{1-\thetamax^2} \right)                                \\
     & = 2 \frac{\sigma_{\max}^2 \lVert \theta \rVert_2}{(1 - \thetamax)^2}.
  \end{align*}
  We finish with the trace of the Hadamard product~$R(\theta) \odot R(\theta)$.
  \begin{align*}
    \Tr[R(\theta) \odot R(\theta)]
     & = T \Tr[(\theta \Gamma_0(\theta) \theta') \odot (\theta \Gamma_0(\theta) \theta')]                                                                           \\
     & \leq T \lVert \theta \Gamma_0(\theta) \theta' \rVert_F^2 \leq T \sigma_{\max}^4 \frac{\lVert \theta \rVert_2^2 \lVert \theta \rVert_F^2}{(1 - \thetamax)^2}.
  \end{align*}
\end{proof}

\subsection{Upper Bound on the KL Divergence}

We now have all the tools in hand to extract a KL divergence bound.

\begin{lemma}[Final KL bound]\label{lem:kl_bound_avg}
  Assume~$\theta \in \Theta_s$ satisfies
  \begin{equation*}
    \lVert \theta \rVert_2 \leq \frac{(1 - \thetamax)^2 (\sigma_{\min}^2 + \omega^2)}{4 \sigma_{\max}^2}
  \end{equation*}
  then the expected conditional KL divergence is upper-bounded as follows:
  \begin{equation*}
    \bbE_\Pi \left[ \KL{\bbP_{\theta}(Y | \Pi)}{\bbP_{0}(Y | \Pi)} \right] \leq \KLmax(\lVert \theta \rVert_2, \lVert \theta \rVert_F)
  \end{equation*}
  where we defined
  \begin{equation*}
    \gamma_\ell = (1-\thetamax)^{3/2} \frac{\sigma_{\min}^2 + \omega^2}{\sigma_{\max}^2} \qquad \text{and} \qquad \KLmax(\lVert \theta \rVert_2, \lVert \theta \rVert_F) = \frac{2Tp(\lVert \theta \rVert_2^2 + q_\ell)\lVert \theta \rVert_F^2}{\gamma_\ell}.
  \end{equation*}
\end{lemma}

\begin{proof}
  Let us start with Lemma~\ref{lem:kl_close_gaussians_applied} on the conditional KL divergence between~$\bbP_\theta(Y | \Pi)$ and~$\bbP_0 (Y | \Pi)$: for any non-empty~$\Pi$,
  \begin{equation*}
    \KL{\bbP_{\theta}(Y | \Pi)}{\bbP_{0}(Y | \Pi)}
    \leq \frac{\lVert \Delta_\Pi(\theta) \rVert_F^2}{2(1 + \lambda_{\min}(\Delta_\Pi(\theta)))}
  \end{equation*}
  We continue with Lemma~\ref{lem:delta_to_R} linking~$\Delta_\Pi(\theta)$ to~$R_\Pi(\theta)$.
  As long as~$\lVert R(\theta) \rVert_2 \leq (\sigma_{\min}^2 + \omega^2) / 2$ (we will see to that at the end), we have
  \begin{equation*}
    \KL{\bbP_{\theta}(Y | \Pi)}{\bbP_{0}(Y | \Pi)}
    \leq \frac{\lVert R_\Pi(\theta) \rVert_F^2}{(\sigma_{\min}^2 + \omega^2)^2}.
  \end{equation*}
  Taking the expectation on the event~$E^c$ yields:
  \begin{equation*}
    \bbE_\Pi \left[\one_{E^c} \KL{\bbP_{\theta}(Y | \Pi)}{\bbP_{0}(Y | \Pi)}\right]
    \leq \frac{\bbE_\Pi \left[ \one_{E^c} \lVert R_\Pi(\theta) \rVert_F^2 \right]}{(\sigma_{\min}^2 + \omega^2)^2}
  \end{equation*}
  We can now apply Lemma~\ref{lem:submatrix_averaging}:
  \begin{equation*}
    \bbE_\Pi \left[\one_{E^c} \KL{\bbP_{\theta}(Y | \Pi)}{\bbP_{0}(Y | \Pi)}\right]
    \leq \frac{p \Tr[R(\theta) \odot R(\theta)] + pq_\ell \lVert R(\theta) \rVert_F^2}{(\sigma_{\min}^2 + \omega^2)^2}
  \end{equation*}
  We substitute the residual bounds from Lemma~\ref{lem:full_R_control}:
  \begin{align*}
    \bbE_\Pi \left[ \one_{E^c} \KL{\bbP_{\theta}(Y | \Pi)}{\bbP_{0}(Y | \Pi)} \right]
     & \leq \frac{p \times \frac{T \sigma_{\max}^4}{(1 - \thetamax)^2} \lVert \theta \rVert_2^2 \lVert \theta \rVert_F^2 + pq_\ell \times \frac{2T \sigma_{\max}^4}{(1 - \thetamax)^3} \lVert \theta \rVert_F^2}{(\sigma_{\min}^2 + \omega^2)^2} \\
     & \leq \left(\frac{\sigma_{\max}^2}{\sigma_{\min}^2 + \omega^2}\right)^2 \frac{2Tp(\lVert \theta \rVert_2^2 + q_\ell)\lVert \theta \rVert_F^2}{(1-\thetamax)^3}                                                                             \\
     & = \frac{2Tp(\lVert \theta \rVert_2^2 + q_\ell)\lVert \theta \rVert_F^2}{\gamma_\ell}.
  \end{align*}
  By Equation~\eqref{eq:effect_E}, this is equivalent to bounding the expected KL divergence regardless of the event~$E$, hence the result.
  Note that our assumption on~$\theta$, combined with Lemma~\ref{lem:full_R_control}, implies
  \begin{equation*}
    \lVert R(\theta) \rVert_2 \leq \frac{2 \sigma_{\max}^2}{(1 - \thetamax)^2} \lVert \theta \rVert_2 \leq \frac{2 \sigma_{\max}^2}{(1 - \thetamax)^2} \frac{(1 - \thetamax)^2 (\sigma_{\min}^2 + \omega^2)}{4 \sigma_{\max}^2} \leq \frac{\sigma_{\min}^2 + \omega^2}{2}
  \end{equation*}
\end{proof}

\subsection{Application of Fano's Method} \label{sec:fano_application}

Given the KL bound we just obtained, we are finally able to prove Theorem~\ref{thm:lower_bound_sparse}.

\begin{proof}
  Fano's method requires finding~$M+1$ parameters~$\theta_i$ such that~$\theta_0 = 0$ and~$\lVert \theta_i - \theta_j \rVert_F \geq 2\tau$ for~$i \neq j$ (with~$\tau$ to be specified), while keeping control upon the average KL divergence between the probability distributions~$\bbP_{\theta_i}$ and~$\bbP_0$.
  Judging by Lemma~\ref{lem:kl_bound_avg}, one way to achieve this control on the KL divergence is to bound the~$\lVert \theta_i \rVert_F$ uniformly in~$i$ (in other words, to choose them all inside a ball of fixed radius).
  We will then have to see how many~$2\tau$-separated matrices we can fit in such a ball.

  Let us consider the set~$\mathcal{H}(r)$ of all block-diagonal~$D \times D$ matrices with coefficients in~$\{0, r\}$ such that each block has size~$s \times s$ (we assume~$s$ divides~$D$).
  In particular, these matrices are all row- and column-sparse, with no more than~$s$ non-zero coefficients per row or column.
  In terms of dimensionality, we are dealing with the (scaled) matrix equivalent of a~$Ds$-dimensional hypercube, hence the notation~$\mathcal{H}(r)$.
  It has cardinality~$2^{Ds}$ and for every~$\theta \in \mathcal{H}$, we have the following norm bounds:
  \begin{equation*}
    \lVert \theta \rVert_2 \leq rs \qquad \text{and} \qquad \lVert \theta \rVert_F \leq r \sqrt{Ds}.
  \end{equation*}
  The spectral norm bound on~$\theta$ is obtained as the maximum spectral norm of each block, which we in turn control using the Frobenius norm of each block.

  \medskip

  Unfortunately, in this hypercube, not all pairs of vertices are well-separated.
  That is why we need the Gilbert-Varshamov bound of Lemma~\ref{lem:gilbert}: according to this result, there exists a pruned subset~$\mathcal{K}(r) \subset \mathcal{H}(r)$ containing~$0$ and such that
  \begin{equation*}
    |\mathcal{K}(r)| \geq |\mathcal{H}(r)|^{1/8} = 2^{Ds/8} \qquad \text{and} \qquad \lVert \vecm(\theta_i) - \vecm(\theta_j) \lVert_1 \geq \frac{r Ds}{8}
  \end{equation*}
  for all pairs of distinct vertices~$\theta_i$ and~$\theta_j$ in~$\mathcal{K}(r)$.
  We choose our set of parameters~$\theta_0, \theta_1, ..., \theta_{M}$ to be exactly this pruned subset~$\mathcal{K}(r)$, in particular~$M+1 = |\mathcal{K}(r)|$.

  \medskip

  The missing ingredient is an upper bound on the maximum average KL divergence between~$\bbP_{\theta_i}$ and~$\bbP_{0}$: we can obtain it using Lemma~\ref{lem:kl_bound_avg}.
  We only need to assume
  \begin{equation*}
    \lVert \theta_i \rVert_F \leq r \sqrt{Ds} \leq \min \left\{ \thetamax,  \frac{(1 - \thetamax)^2 (\sigma_{\min}^2 + \omega^2)}{4 \sigma_{\max}^2} \right\}
  \end{equation*}
  to get the upper bound
  \begin{align*}
    \max_i \bbE_\Pi \left[\KL{\bbP_{\theta_i}(Y|\Pi)}{\bbP_{\theta_0}(Y|\Pi)}\right]
     & \leq \max_i \KLmax(\lVert \theta_i \rVert_2, \lVert \theta_i \rVert_F) \\
     & \leq \KLmax(rs, r\sqrt{Ds}).
  \end{align*}
  Since we must satisfy the constraint from Equation~\eqref{eq:kl_constraint_fano} in Fano's method, we will choose~$r$ so that:
  \begin{equation*}
    \KLmax(rs, r\sqrt{Ds}) \leq \alpha \log(M) = \alpha \log\left(2^{Ds/8}-1\right)
  \end{equation*}
  with~$\alpha = \frac{\log 3 - \log 2}{2 \log 2}$.
  We want to solve the previous inequality for~$r$, and for that we start by replacing~$\KLmax(rs, r\sqrt{Ds})$ with its value from Lemma~\ref{lem:kl_bound_avg}, replacing~$\gamma_\ell$ with~$\gamma_\ell$ to lighten notations:
  \begin{align*}
    \KLmax(rs, r\sqrt{Ds}) \leq \alpha \log\left(2^{Ds/8}-1\right)
     & \iff \frac{2}{\gamma_\ell} Tp\left((rs)^2 + q_\ell\right) (r\sqrt{Ds})^2 \leq c Ds \\
     & \iff D s^3 r^4 + q_\ell Ds r^2 - c\frac{\gamma_\ell^2 Ds}{Tq_\ell} \leq 0.
  \end{align*}
  If we consider this as a degree two polynomial in the variable~$r^2$, its determinant is
  \begin{equation*}
    \Delta = q_\ell^2 D^2 s^2 + 4 D s^3 c\frac{\gamma_\ell^2 Ds}{Tp}.
  \end{equation*}
  For~$\beta$ to be small enough,~$r^2$ must remain below the only positive root of the polynomial, namely
  \begin{equation*}
    r^2 \leq \frac{- q_\ell Ds + \sqrt{q_\ell^2 D^2 s^2 + c\frac{\gamma_\ell^2 D^2 s^4}{Tp}}}{2D s^3} = \frac{q_\ell}{2s^2} \left(\sqrt{1 + c \frac{\gamma_\ell^2 s^2}{Tpq_\ell^2}} - 1 \right).
  \end{equation*}
  If we assume the quantity~$c \frac{\gamma_\ell^2 s^2}{Tpq_\ell^2}$ inside the square root is smaller than~$1$, i.e.
  \begin{equation} \label{eq:Tlarge_minimax_1}
    \frac{\gamma_\ell s}{\sqrt{p}q_\ell\sqrt{T}} \leq c,
  \end{equation}
  then we can lower-bound~$\sqrt{1+x}$ by its chord~$(\sqrt{2}-1)x$.
  In other words, a sufficient condition for~$r^2$ to remain small enough is given by
  \begin{equation*}
    r^2 \leq \frac{q_\ell}{2s^2} \times (\sqrt{2}-1) c \frac{\gamma_\ell^2 s^2}{Tpq_\ell^2} = c \frac{\gamma_\ell^2}{Tpq_\ell}.
  \end{equation*}
  To sum up, we have three constraints on~$r$:
  \begin{align*}
    r s \leq \thetamax                                                                                                           &  &
    r s \leq \frac{(1 - \thetamax)^2 (\sigma_{\min}^2 + \omega^2)}{4 \sigma_{\max}^2} = \frac{\sqrt{1-\thetamax}}{4} \gamma_\ell &  &
    r   \leq \sqrt{c \frac{\gamma_\ell^2}{Tpq_\ell}}.
  \end{align*}
  We can therefore choose~$r$ as the largest value satisfying all three of them:
  \begin{equation} \label{eq:r_val}
    r = \frac{1}{s} \min \left\{ \thetamax, \frac{\gamma_\ell \sqrt{1-\thetamax}}{4}, c \frac{\gamma_\ell s}{\sqrt{Tpq_\ell}} \right\}
  \end{equation}
  To reach our conclusion, we simply need to remark that the vectorized~$\ell_1$ distance between any two matrices in~$\mathcal{K}(r)$ gives us a lower bound on the operator~$\ell_{\infty}$ distance that separates them:
  \begin{align*}
    \lVert \theta_i - \theta_j \rVert_{\infty}
     & = \max_{k \in [D]} \sum_{l \in [D]} |(\theta_i - \theta_j)|_{k,l} \geq \frac{1}{D} \sum_{1 \leq k,l \leq D} |(\theta_i - \theta_j)|_{k,l} \\
     & = \frac{1}{D} \lVert \vecm(\theta_i) - \vecm(\theta_j) \rVert_1 \geq \frac{rDs}{8D} = \frac{rs}{8}
  \end{align*}
  Subsequently, our parameters~$\theta_i$ are~$2\tau$-separated (in~$\ell_{\infty}$ operator distance) with~$\tau = rs / 8$.
  As soon as the minimum in Equation~\eqref{eq:r_val} is reached by the third value, i.e. whenever
  \begin{equation} \label{eq:Tlarge_minimax_2}
    \frac{ \gamma_\ell s}{\sqrt{T p q_\ell}} \leq c \min\{\thetamax, \gamma_\ell \sqrt{1-\thetamax}\}
  \end{equation}
  we can simplify the expression of~$\tau$:
  \begin{equation*}
    \tau = c \frac{\gamma_\ell s}{\sqrt{T p q_\ell}}.
  \end{equation*}
  In this case, by Lemma~\ref{lem:fano_method}, we can conclude:
  \begin{equation*}
    \inf_{\widehat{\theta}} \sup_{\theta \in \Theta_s} \bbP_{\theta} \left[ \lVert \widehat{\theta} - \theta \rVert_{\infty} \geq c \frac{\gamma_\ell s}{\sqrt{T p q_\ell}} \right] \geq \frac{\log(M+1) - \log 2}{\log M} - \alpha \geq \frac{1}{2}.
  \end{equation*}

\end{proof}

\section{Useful Lemmas} \label{sec:lemmas}

\subsection{Linear Algebra}

The following set of results will sometimes be used in matrix calculations without explicit justifications.

\begin{lemma}[Weyl's inequality] \label{lem:weyl}
  Let~$A$ and~$B$ be two~$n \times n$ symmetric matrices. Then for all~$i$ we have:
  \begin{equation*}
    \lambda_i(A) + \lambda_n(B) \leq \lambda_i(A+B) \leq \lambda_i(A) + \lambda_1(B).
  \end{equation*}
  In particular,
  \begin{equation*}
    \lambda_{\min}(A) + \lambda_{\min}(B) \leq \lambda_{\min}(A+B).
  \end{equation*}
\end{lemma}

\begin{proof}
  See \citet[Theorem 4.3.1]{hornMatrixAnalysis2012}.
\end{proof}

\begin{lemma}[Ostrowski] \label{lem:ostrowski}
  Let~$S$ and~$A$ be two~$n \times n$ matrices with~$S$ symmetric. For all~$i$, there is a real number~$r_i \in [s_{\min}(A)^2, s_{\max}(A)^2]$ such that~$\lambda_i(ASA') = r_i \lambda_i(S)$, where~$s_{\min}$ (resp.~$s_{\max}$) denotes the minimum (resp. maximum) singular value.
\end{lemma}

\begin{proof}
  See \citet[Theorem 4.5.9 and Corollary 4.5.11]{hornMatrixAnalysis2012}
\end{proof}

\begin{lemma}[Singular values of the Kronecker product] \label{lem:kron_singular_values}
  Let~$A$ and~$B$ be two matrices. Then
  \begin{equation*}
    \lVert A \otimes B \rVert_2 \leq \lVert A \rVert_2 \lVert B \rVert_2.
  \end{equation*}
\end{lemma}

\begin{proof}
  See \citet[Theorem 4.2.15]{hornTopicsMatrixAnalysis1994}.
\end{proof}

\begin{lemma} \label{lem:frobenius_spectral_product}
  For any two matrices~$A$ and~$B$, we have:
  \begin{equation*}
    \lVert A B \rVert_F \leq \min \left\{ \lVert A \rVert_2 \lVert B \rVert_F,  \lVert A \rVert_F \lVert B \rVert_2 \right\}
  \end{equation*}
\end{lemma}

\begin{proof}
  The Loewner order on symmetric matrices satisfies the following properties:
  \begin{align*}
    \forall (P, Q) \in \mathcal{S}_n(\bbR), \forall R, \quad & P \preceq Q \implies R'PR \preceq R'QR   \\
    \forall (P, Q) \in \mathcal{S}_n(\bbR), \quad            & P \preceq Q \implies \Tr(P) \leq \Tr(Q).
  \end{align*}
  The first inequality is true because if~$x$ is a vector,~$x' R' (Q - P) R x = (R x)' (Q - P) (R x) \geq 0$ due to the Loewner positivity of~$Q - P$.
  The second inequality can be directly deduced from the relation between the spectra of~$P$ and~$Q$.
  Therefore, since~$A'A$ is symmetric,
  \begin{equation*}
    B'A'AB \leq \lambda_{\max}(A'A) B'B
  \end{equation*}
  which implies
  \begin{equation*}
    \lVert AB \rVert_F^2 = \Tr(B' A' A B) \leq \lambda_{\max}(A A') \Tr(B' B)  = \lVert A \rVert_2^2 \rVert B \rVert_F^2.
  \end{equation*}
  The proof for the other inequality is identical.
\end{proof}

\begin{lemma} \label{lem:maxnorm_l1linf}
  Let~$A$ and~$B$ be two matrices with compatible sizes: then
  \begin{equation*}
    \lVert A B \rVert_{\max} \leq \min\{\lVert A \rVert_{\infty} \lVert B \rVert_{\max}, \lVert A \rVert_{\max} \lVert B \rVert_1\}.
  \end{equation*}
\end{lemma}

\begin{proof}
  \begin{equation*}
    \lVert AB \rVert_{\max} = \max_{i,j} |(AB)_{i,j}| = \max_{i,j} \left| \sum_k A_{i,k} B_{k,j} \right|
  \end{equation*}
  We easily deduce:
  \begin{align*}
    \lVert AB \rVert_{\max} & \leq \max_{i} \left| \sum_k A_{i,k} \right| \times \lVert B \rVert_{\max} = \lVert A \rVert_{\infty} \lVert B \rVert_{\max} \\
    \lVert AB \rVert_{\max} & \leq \lVert A \rVert_{\max} \times  \max_{j} \left| \sum_k B_{k,j} \right| = \lVert A \rVert_{\max} \lVert B \rVert_1
  \end{align*}
\end{proof}

\subsection{Probability}

\begin{lemma}[Fano's method] \label{lem:fano_method}
  Let~$\theta_0, ..., \theta_M$ be~$M+1$ parameters that are~$2\tau$-separated w.r.t. a distance~$d$
  \begin{equation*}
    \forall i \neq j, \quad d(\theta_i, \theta_j) \geq 2 \tau
  \end{equation*}
  and such that the average KL divergence between~$\bbP_{\theta_i}$ and~$\bbP_{\theta_0}$ is small enough
  \begin{equation}
    \frac{1}{M+1} \sum_{i=1}^{M} \KL{\bbP_{\theta_i}}{\bbP_{\theta_0}} \leq \alpha \log M \quad \text{with} \quad 0 < \alpha < 1 \label{eq:kl_constraint_fano}
  \end{equation}
  Then the minimax probability of an error at threshold~$\tau$ satisfies:
  \begin{equation*}
    \inf_{\widehat{\theta}} \sup_{\theta \in \Theta_s} \bbP_{\theta} \left[ d\left(\widehat{\theta}, \theta\right) \geq \tau \right] \geq \frac{\log(M+1) - \log 2}{\log M} - \alpha.
  \end{equation*}
\end{lemma}

\begin{proof}
  See \citet[Section 2.2 + Corollary 2.6]{tsybakovIntroductionNonparametricEstimation2008}.
  In particular, since~$M \mapsto \frac{\log(M+1) - \log 2}{\log M}$ is increasing, setting~$\alpha = \frac{\log(3) - \log(2)}{2\log(2)} \geq 1/2$ is enough to obtain a minimax risk greater than~$\alpha$, as soon as~$M \geq 3$.
\end{proof}

\begin{lemma}[Chain rule for KL divergence] \label{lem:kl_chain}
  If~$\bbP_0$ and~$\bbP_1$ are probability densities on a product space~$\mathcal{X} \times \mathcal{Y}$ with~$\mathcal{X}$ discrete, then:
  \begin{align*}
    \KL{\bbP_0[X, Y]}{\bbP_1[X, Y]}
     & = \KL{\bbP_0[X]}{\bbP_1[X]}  + \bbE_X \left[ \KL{\bbP_0[Y|X]}{\bbP_1[Y|X]} \right].
  \end{align*}
\end{lemma}

\begin{proof}
  See \citet[Theorem 2.5.3]{coverElementsInformationTheory2012}.
\end{proof}

\begin{lemma}[KL divergence between Gaussians] \label{lem:kl_gaussian}
  The KL divergence between two multivariate Gaussian distributions~$\bbP_0 = \mathcal{N}(\mu_0, \Sigma_0)$ and~$\bbP_1 = \mathcal{N}(\mu_1, \Sigma_1)$ of dimension~$n$ is
  \begin{equation*}
    \KL{\bbP_0}{\bbP_1} = \frac{1}{2} \left( \Tr(\Sigma_0 \Sigma_1^{-1}) + (\mu_1 - \mu_0)' \Sigma_1^{-1} (\mu_1 - \mu_0) - n + \logdet (\Sigma_1 \Sigma_0^{-1}) \right).
  \end{equation*}
\end{lemma}

\begin{proof}
  See \citet[page 13]{duchiDerivationsLinearAlgebra2007}.
\end{proof}

\begin{lemma}[KL divergence between close Gaussians] \label{lem:kl_around_id}
  Let~$\Delta$ be a symmetric matrix of size~$n$ such that~$\lambda_{\min}(\Delta) > -1$, and let~$M$ be a rectangular matrix such that~$MM' \succ 0$.
  Then the KL divergence between
  \begin{equation*}
    \bbP_1 = \mathcal{N}(\mu, M (I + \Delta) M') \quad \text{and} \quad \bbP_0 = \mathcal{N}(\mu, MM')
  \end{equation*}
  satisfies
  \begin{equation*}
    \KL{\bbP_1}{\bbP_0} \leq \frac{\lVert \Delta \rVert_F^2}{2(1 + \lambda_{\min}(\Delta))}.
  \end{equation*}
\end{lemma}

\begin{proof} From Lemma~\ref{lem:kl_gaussian} (beware of the switch between~$\bbP_0$ and~$\bbP_1$) we get:
  \begin{align*}
    \KL{\bbP_1}{\bbP_0}
     & = \frac{1}{2} \left( \Tr(\Sigma_1 \Sigma_0^{-1}) + (\mu_0 - \mu_1)' \Sigma_0^{-1} (\mu_0 - \mu_1) - n + \logdet (\Sigma_0 \Sigma_1^{-1}) \right) \\
     & = \frac{1}{2} \left( \Tr(M(I+\Delta)M^{-1}) - n - \logdet (M (I+\Delta) M^{-1}) \right)                                                          \\
     & = \frac{1}{2} \left( \Tr(\Delta) - \logdet (I+\Delta) \right).
  \end{align*}
  As it happens, for small deviations from the identity, the log-determinant is almost equal to the trace.
  Indeed, since
  \begin{equation*}
    \forall x > -1, \quad \log(1+x) \geq \frac{x}{1+x},
  \end{equation*}
  we have
  \begin{align*}
    \Tr(\Delta) - \logdet(I+\Delta)
     & = \sum_{k=1}^{n} \lambda_k(\Delta) - \sum_{k=1}^{n} \log(1 + \lambda_k(\Delta))                                                                     \\
     & \leq \sum_{k=1}^{n} \lambda_k(\Delta) - \sum_{k=1}^{n} \frac{\lambda_k(\Delta)}{1 + \lambda_k(\Delta)}                                              \\
     & = \sum_{k=1}^{n} \frac{\lambda_k(\Delta)^2}{1 + \lambda_k(\Delta)} \leq \frac{1}{\min_k (1 + \lambda_k(\Delta))} \sum_{k=1}^{n} \lambda_k(\Delta)^2 \\
     & = \frac{\lVert \Delta \rVert_F^2}{1 + \lambda_{\min}(\Delta)}.
  \end{align*}
\end{proof}

\begin{lemma}[Chernoff inequality for Bernoulli variables] \label{lem:chernoff_bernoulli}
  Let~$(X_t)$ be sequence of independent~$\mathcal{B}(p)$ variables. Their average satisfies
  \begin{align*}
    \forall u \in [0, 1], \quad \bbP\left( \left\lvert\frac{1}{T} \sum_{t=1}^{T} X_t - p \right\rvert \geq u p \right) \leq c_1 \exp\left(-c_2 u^2 T p\right).
  \end{align*}
\end{lemma}

\begin{proof}
  See \citet[Theorem 1.1]{dubhashiConcentrationMeasureAnalysis2009}.
\end{proof}

\begin{lemma}[Doeblin condition and mixing time] \label{lem:doeblin}
  Let~$(X_t)$ be an irreducible aperiodic Markov chain with state space~$\mathcal{X}$, transition matrix~$P$ and stationary distribution~$\mu$.
  Suppose that~$(X_t)$ satisfies the Doeblin condition:
  \begin{equation*}
    \exists r \in \bbN,  \exists \delta > 0,  \forall (x, y) \in \mathcal{X}^2, \quad P^r(x, y) \geq \delta \mu(y).
  \end{equation*}
  Then the mixing time of~$X_t$, defined as
  \begin{equation*}
    \tmix(\epsilon) = \min \left\{t \in \bbN:  \max_{x \in \mathcal{X}} \left\lVert P^t(x, \cdot) - \mu \right\rVert_{\TV} \leq \epsilon \right\},
  \end{equation*}
  satisfies:
  \begin{equation*}
    \tmix(\epsilon) \geq r \left(1 +  \frac{\log\frac{1}{\epsilon}}{\log\frac{1}{1-\delta}}\right).
  \end{equation*}
\end{lemma}

\begin{proof}
  The proof of \citet[Theorem 5.4]{levinMarkovChainsMixing2017} shows that with our assumptions,
  \begin{equation*}
    \forall x \in \mathcal{X}, \quad \left\lVert P^t(x, \cdot) - \mu \right\rVert_{\TV} \leq (1-\delta)^{\lfloor t/r \rfloor}.
  \end{equation*}
  From which we can deduce a sufficient condition for~$\epsilon$-mixing:
  \begin{equation*}
    (1-\delta)^{\lfloor t/r \rfloor} \leq \epsilon
    \quad \iff \quad  \left\lfloor \frac{t}{r} \right\rfloor \geq \frac{\log(\epsilon)}{\log(1-\delta)}
    \quad \impliedby \quad \frac{t}{r} - 1 \geq \frac{\log\frac{1}{\epsilon}}{\log\frac{1}{1-\delta}}.
  \end{equation*}
  The result follows easily.
\end{proof}

\begin{lemma}[Chernoff inequality for Markov chains] \label{lem:chernoff_markov}
  Let~$(X_t)$ be an ergodic stationary Markov chain with finite state space~$\mathcal{X}$.
  We consider a function~$f: \mathcal{X} \to \bbR$ such that~$\bbE[f(X_t)] = \mu$.
  Then
  \begin{equation*}
    \forall u \in [0, 1], \quad \bbP\left( \left\lvert \frac{1}{T} \sum_{t=1}^{T} X_t - \mu \right\rvert \geq u \mu \right) \leq c_1 \exp\left(-c_2\frac{u^2 T \mu}{\tmix(1/8)}\right)
  \end{equation*}
\end{lemma}

\begin{proof}
  See \citet[Theorem 3]{chungChernoffHoeffdingBoundsMarkov2012}
\end{proof}

\begin{lemma}[Chernoff inequality for Markov chains under Doeblin condition] \label{lem:chernoff_doeblin}
  Under the hypotheses of the previous two Lemmas (\ref{lem:doeblin} and~\ref{lem:chernoff_markov}), if the parameters~$r$ and~$\delta$ in the Doeblin condition are constants, then we have:
  \begin{equation*}
    \forall u \in [0, 1], \quad \bbP\left( \left\lvert \frac{1}{T} \sum_{t=1}^{T} X_t - \mu \right\rvert \geq u \mu \right) \leq c_1 \exp\left(-c_2 u^2 T \mu\right)
  \end{equation*}
\end{lemma}

\begin{proof}
  By Lemma~\ref{lem:doeblin}, since~$r$ and~$\delta$ are constants, the~$\frac{1}{8}$-mixing time of~$(X_t)$ can be bounded by a constant
  \begin{equation*}
    \tmix(1/8) \leq r \left(1 +  \frac{\log(8)}{\log\frac{1}{1-\delta}}\right) \leq c_3,
  \end{equation*}
  which we merge with the~$c_2$ inside the exponential of Lemma~\ref{lem:chernoff_markov}.
\end{proof}

\begin{lemma}[Gilbert-Varshamov] \label{lem:gilbert}
  Let~$\mathcal{H} = \{0, 1\}^d$ be the~$d$-dimensional binary hypercube.
  If~$d \geq 8$, there exists a pruned subset~$\mathcal{K} \subset \mathcal{H}$ such that
  \begin{equation*}
    \forall (x, y) \in \mathcal{K}, ~  \lVert x - y \rVert_1 \geq \frac{d}{8} \quad \text{and} \quad |\mathcal{K} | \geq 2^{d/8}.
  \end{equation*}
\end{lemma}

\begin{proof}
  See \citet[Lemma 2.9]{tsybakovIntroductionNonparametricEstimation2008}
\end{proof}

\begin{lemma}[Hanson-Wright inequality: Gaussian case] \label{lem:hanson_wright}
  Let~$A$ be a square matrix. If~$X$ and~$Y$ are two independent standard Gaussian vectors, we have:
  \begin{align*}
    \bbP \left(|X'AX - \bbE[X'AX]| \geq u \right) & \leq 2 \exp \left( -c \min \left\{\frac{u^2}{\lVert A \rVert_F^2}, \frac{u}{\lVert A \rVert_2} \right\} \right)  \\
    \bbP \left(|X'AY - \bbE[X'AY]| \geq u\right)  & \leq 2 \exp \left( -c \min \left\{\frac{u^2}{\lVert A \rVert_F^2}, \frac{u}{\lVert A \rVert_2} \right\} \right).
  \end{align*}
\end{lemma}

\begin{proof}
  See \citet[Theorem 6.2.1]{vershyninHighDimensionalProbabilityIntroduction2018} for the first inequality. We will see that it implies the second one. Let us define
  \begin{equation*}
    \widetilde{A} = \begin{bmatrix}
      0 & A \\
      0 & 0
    \end{bmatrix} \quad \text{and} \quad \widetilde{X} = \begin{bmatrix}
      X \\ Y
    \end{bmatrix}.
  \end{equation*}
  We note that~$\lVert \widetilde{A} \lVert_F = \lVert A \lVert_F$ and~$\lVert \widetilde{A} \lVert_2 = \lVert A \lVert_2$.
  Applying the first inequality to~$\widetilde{X}' \widetilde{A} \widetilde{X} = X' A Y$ yields the expected result.
\end{proof}

\begin{lemma}[Conditional Hanson-Wright inequality] \label{lem:conditional_hanson_wright}
  Let~$A$ be a random square matrix such that with probability~$1 - \delta$,
  \begin{equation*}
    \lVert A \rVert_2 \leq M_2 \qquad \text{and} \qquad \lVert A \rVert_F^2 \leq M_F^2.
  \end{equation*}
  If~$X$ and~$Y$ are two independent standard Gaussian vectors independent of~$A$, we have:
  \begin{align*}
    \bbP \left(|X'AX - \bbE[X'AX]| \geq u \right) & \leq \delta + 2 \exp \left( -c \min \left\{\frac{u^2}{M_F^2}, \frac{u}{M_2} \right\} \right) + \bbP\left(|\Tr(A - \bbE[A])| \geq u/2\right) \\
    \bbP \left(|X'AY - \bbE[X'AY]| \geq u \right) & \leq \delta + 2 \exp \left( -c \min \left\{\frac{u^2}{M_F^2}, \frac{u}{M_2} \right\} \right).
  \end{align*}
\end{lemma}

\begin{proof}
  We start with the first case. Since~$A$ is a discrete random matrix with a finite set~$\mathcal{A}$ of possible values,
  \begin{align*}
    \bbP(|X'AX - \bbE[X'AX]| \geq u)
     & = \sum_{a \in \mathcal{A}} \bbP(|X'AX - \bbE[X'AX]| \geq u \cap A=a)  \\
     & = \sum_{a \in \mathcal{A}} \bbP(|X'aX - \bbE[X'AX]| \geq u \cap A=a).
  \end{align*}
  Using independence between~$X$ and~$A$ gives us
  \begin{equation*}
    \bbP(|X'AX - \bbE[X'AX]| \geq u)
    = \sum_{a \in \mathcal{A}} \bbP(|X'aX - \bbE[X'AX]| \geq u) \bbP(A=a).
  \end{equation*}
  We now split the set of feasible values~$\mathcal{A}$ into
  \begin{equation*}
    \mathcal{A}_{\leq} = \{a \in \mathcal{A}: \lVert a \rVert_F^2 \leq M_F^2\} \quad \text{and} \quad \mathcal{A}_{>} = \{a \in \mathcal{A}: \lVert a \rVert_F^2 > M_F^2\}.
  \end{equation*}
  Since we assumed~$\bbP(A \in \mathcal{A}_{>}) = \sum_{a \in A_>} \bbP(A=a) \leq \delta$, we get:
  \begin{equation*}
    \bbP(|X'AX - \bbE[X'AX]| \geq u)
    \leq \delta + \sum_{a \in \mathcal{A}_\leq} \bbP(|X'aX - \bbE[X'AX]| \geq u) \bbP(A=a).
  \end{equation*}
  Unfortunately, Lemma~\ref{lem:hanson_wright} only lets us bound
  \begin{equation*}
    \bbP(|X'aX - \bbE[X'aX]| \geq u)\quad \text{and not} \quad \bbP(|X'aX - \bbE[X'AX]| \geq u)
  \end{equation*}
  (notice the change inside the expectation), which means we need an additional step. For a fixed~$a \in \mathcal{A}_{\leq}$, we use independence and normality to obtain
  \begin{align*}
    \bbE[X'aX] - \bbE[X'AX] & = \bbE[\Tr(X'(a-A)X)] = \Tr(\bbE[XX'(a-A)])  \\
                            & = \Tr(\bbE[XX'] \bbE[a-A]) = \Tr(a-\bbE[A]).
  \end{align*}
  We are now ready to decompose, with the help of the union bound:
  \begin{align*}
    \bbP(|X'aX - \bbE[X'AX]| \geq u)
     & = \bbP\left(|X'aX - \bbE[X'aX] + \bbE[X'aX] - \bbE[X'AX]| \geq u\right)                                                                                         \\
     & \leq \bbP\left(|X'aX - \bbE[X'aX]| \geq u/2\right) + \bbP\left(|\bbE[X'aX] - \bbE[X'AX]| \geq u/2\right)                                                        \\
     & \leq 2 \exp \left(-c \min \left\{\frac{u^2}{\lVert a \rVert_F^2}, \frac{u}{\lVert a \rVert_2} \right\} \right) + \one \left\{|\Tr(a-\bbE[A])| \geq u/2\right\}.
  \end{align*}
  This implies:
  \begin{align*}
    \bbP(|X'AX - \bbE[X'AX]| \geq u)
     & \leq \delta + \sum_{a \in \mathcal{A}_\leq} \bbP(A=a) \bbP(|X'aX - \bbE[X'AX]| \geq u)                                                                                  \\
     & \leq \delta + \sum_{a \in \mathcal{A}_\leq} \bbP(A=a) \times 2 \exp \left[ -c \min \left\{\frac{u^2}{\lVert a \rVert_F^2}, \frac{u}{\lVert a \rVert_2} \right\} \right] \\
     & \phantomleq + \sum_{a \in \mathcal{A}_\leq} \bbP(A=a) \times \one \left\{|\Tr(a-\bbE[A])| \geq u/2\right\}.
  \end{align*}
  By definition of~$\mathcal{A}_\leq$,
  \begin{align*}
    \bbP(|X'AX - \bbE[X'AX]| \geq u)
     & \leq \delta + \sum_{a \in \mathcal{A}_\leq} \bbP(A=a) \times 2 \exp \left( -c \min \left\{\frac{u^2}{M_F^2}, \frac{u}{M_2} \right\} \right)  \\
     & \phantomleq + \bbP\left(|\Tr(A - \bbE[A])| \geq u/2\right)                                                                                   \\
     & \leq \delta + 2 \exp \left( -c \min \left\{\frac{u^2}{M_F^2}, \frac{u}{M_2} \right\} \right) + \bbP\left(|\Tr(A - \bbE[A])| \geq u/2\right).
  \end{align*}
  The proof for~$X'AY$ follows the same lines, except that we replace~$\bbE[XX'] = I$ by~$\bbE[XY'] = 0$, which removes the trace term in the final expression.
\end{proof}

\begin{lemma}[Heuristic optimality of the signal-to-noise ratio] \label{lem:signal_to_noise}
  In the one-dimensional setting with full observations, the dependency of the error in~$1 + \frac{\sigma^2}{\omega^2}$ is \enquote{coherent} with the asymptotic behavior of the MLE.
\end{lemma}

\begin{proof}
  Let us consider the case where~$D = 1$ and~$p = 1$, since we are mainly interested in the role of the parameters~$\sigma^2$ and~$\omega^2$.
  In this case, Theorem~\ref{thm:lower_bound_sparse} argues that the error of any estimator should grow at least like~$\gamma_\ell = 1 + \frac{\omega^2}{\sigma^2}$.
  We also note that in this simple scenario, Theorem~\ref{thm:convergence_rate_theta} states that~$\gamma_u \propto \gamma_\ell$.

  We will compare this to the asymptotic error of the Maximum Likelihood Estimator (MLE)~$\widehat{\theta}$, which (for well-behaved models) is given by the inverse of the Fisher information matrix.
  To make this statement more precise, we will invoke \citet[Proposition 2.14]{doucNonlinearTimeSeries2014}. Let us verify the conditions:
  \begin{itemize}
    \item The process is stable, i.e.~$\rho(\theta) < 1$.
          We made sure of that by assuming~$\lVert \theta \rVert_2 \leq \thetamax < 1$.
    \item The sampling matrix~$\Pi_t$ is constant across time.
          Although this assumption is not essential, it is true here since~$p = 1$ and~$D = 1$ hence~$\Pi_t = I_1$.
    \item The model has the smallest possible dimension.
    \item The true parameter~$\theta$ is identifiable and does not lie on the boundary of~$\Theta_s$.
          Identifiability is easily deduced from Lemma~\ref{lem:x_covariance} by observing that~$\theta = \Gamma_1(\theta) \Gamma_0(\theta)^{-1}$ can be entirely deduced from distribution moments.
  \end{itemize}
  Since all of these prerequisites hold here, \citet[Proposition 2.14]{doucNonlinearTimeSeries2014} gives us a Central Limit Theorem for the MLE of linear Gaussian models:
  \begin{align*}
    \sqrt{T}(\widehat{\theta} - \theta) \xrightarrow[T \to \infty]{\mathcal{L}} \mathcal{N}(0, \mathcal{I}_{\infty}(\theta)^{-1}) \quad \text{where} \quad \mathcal{I}_{\infty}(\theta) = \lim_{T \to \infty} \frac{\mathcal{I}_T(\theta)}{T}.
  \end{align*}
  We only have to compute the Fisher information matrix~$\mathcal{I}_T(\theta)$.
  The covariance matrix of~$Y$ is given by Lemma~\ref{lem:y_covariance_decomposition}, but in our case the sampling matrix is constant, and we obtain the simpler (unconditional) result
  \begin{equation*}
    \Cov_\theta[Y] = (\sigma^2 + \omega^2) I_T + R(\theta),
  \end{equation*}
  where the residual~$R(\theta)$ is of order 1 in~$\theta$.
  Indeed, our simplifying assumptions imply~$\Gamma_0(\theta) = \frac{\sigma^2}{1 - \theta^2}$ and therefore
  \begin{equation*}
    R(\theta) = \frac{\sigma^2}{1 - \theta^2} \begin{pmatrix}
      \theta^2 & \theta^1 & \theta^2 & \cdots \\
      \theta^1 & \theta^2 & \theta^1 &        \\
      \theta^2 & \theta^1 & \theta^2 &        \\
      \vdots   &          &          & \ddots
    \end{pmatrix}
    \qquad
    \partial_\theta R(\theta) = \sigma^2 \begin{pmatrix}
      0      & 1 & 0 & \cdots \\
      1      & 0 & 1 &        \\
      0      & 1 & 0 &        \\
      \vdots &   &   & \ddots
    \end{pmatrix} + \mathcal{O}(\theta).
  \end{equation*}
  The Fisher information of~$Y$ with respect to~$\theta$ has an explicit formula \citep[Section 3.5]{malagoInformationGeometryGaussian2015}:
  \begin{align*}
    \mathcal{I}_T(\theta)
     & = \frac{1}{2} \Tr \left[ \Cov_\theta[Y]^{-1}   \partial_\theta \Cov_\theta[Y]   \Cov_\theta[Y]^{-1}   \partial_\theta \Cov_\theta[Y] \right]                                                                                                                              \\
     & = \frac{1}{2} \Tr \left[ \left(I + \frac{R(\theta)}{\sigma^2 + \omega^2} \right)^{-1} \frac{\partial_\theta R(\theta)}{\sigma^2 + \omega^2}  \left(I + \frac{R(\theta)}{\sigma^2 + \omega^2} \right)^{-1} \frac{\partial_\theta R(\theta)}{\sigma^2 + \omega^2}  \right].
  \end{align*}
  If assume~$\theta$ is small and perform a Taylor expansion, we get:
  \begin{equation*}
    \mathcal{I}_T(\theta) \approx \frac{1}{2(\sigma^2 + \omega^2)^2} \Tr \left[(\partial_\theta R(\theta))^2\right].
  \end{equation*}
  Incidentally, we also note that at the lowest order in~$\theta$,
  \begin{equation*}
    \Tr[(\partial_\theta R(\theta))^2] = \lVert \partial_\theta R(\theta) \rVert_F^2 \approx 2 \sigma^4 (T-1).
  \end{equation*}
  Which gives us an approximate information matrix for~$T$ steps:
  \begin{equation*}
    \mathcal{I}_T(\theta) \approx \frac{\Tr[(\partial_\theta R(\theta))^2]}{2 (\sigma^2 + \omega^2)^2} \approx \frac{T}{2}\left( \frac{\sigma^2}{\sigma^2 + \omega^2} \right)^2.
  \end{equation*}
  Taking the temporal limit yields:
  \begin{equation*}
    \mathcal{I}_{\infty}(\theta) = \lim_{T \to \infty} \frac{\mathcal{I}_T(\theta)}{T} \approx \frac{1}{2}\left( \frac{\sigma^2}{\sigma^2 + \omega^2} \right)^2.
  \end{equation*}
  In conclusion, this informal analysis reveals an asymptotic error equivalent to
  \begin{equation*}
    \frac{1}{\sqrt{T}} \sqrt{\mathcal{I}_{\infty}(\theta)^{-1}} \approx \frac{\sqrt{2}}{\sqrt{T}} \left(1+\frac{\omega^2}{\sigma^2} \right),
  \end{equation*}
  which is coherent with the dependency we identified in Theorem~\ref{thm:lower_bound_sparse}.

\end{proof}

\section{Glossary} \label{sec:glossary}

\subsection{Notations}

For any integer~$n$, let~$[n] = \{1, ..., n\}$.
The symbol~$\one_{\{...\}}$ stands for an indicator function.
When dealing with random variables, we write~$\bbP(X=x)$ for a probability density,~$\bbE[X]$ for an expectation,~$\Var[X]$ for a variance (scalar of vector) and~$\Cov[X, Y]$ for a covariance (scalar or matrix).
The symbols~$\mathcal{B}(p)$ and~$\mathcal{N}(\mu, \Sigma)$ denote a Bernoulli distribution and a (possibly multivariate) Gaussian distribution.
When we write~$\log(x)$, we mean the natural (base-$e$) logarithm.

Given a real number~$a$, we denote by~$|a|$ its absolute value.
Given a vector~$x$, we denote by~$\lVert x \rVert_2$ (resp.
$\lVert x \rVert_1$,~$\lVert x \rVert_{\infty}$,~$\lVert x \rVert_0$) its Euclidean norm (resp.~$\ell_1$ norm,~$\ell_\infty$ norm, number of nonzero entries).
The notation~$\basis_i$ stands for a vector with a single non-zero coordinate at position~$i$.

A matrix can be defined by its coefficients~$M = (M_{i,j})_{i,j}$ or by its blocks~$M = (M_{[b_1, b_2]})_{b_1,b_2}$.
We write~$I$ for the identity matrix, and~$\Jmat_r$ for the square matrix entirely filled with zeros, except for the subdiagonal of rank~$r$ which is filled with ones.
The notation~$\diag(\lambda)$ stands for the diagonal matrix with coefficients~$\lambda_1, ..., \lambda_n$, while~$\bdiag_T(M)$ stands for a block-diagonal matrix with~$T$ copies of~$M$ on the diagonal and zeros elsewhere.
We write~$\vecm(M)$ for the column-wise flattening of matrix~$M$ into a vector.
When we want to apply a function elementwise, we often use notation that is standard for real numbers but not for matrices: for instance,~$\sqrt{M} = (\sqrt{M_{i, j}})_{i,j}$ and~$1/M = (1/M_{i,j})_{i,j}$.
Given a real matrix~$M$, we denote by
\begin{itemize}
  \item~$M'$ its transposition,~$M^\dagger$ its Moore-Penrose pseudo-inverse and~$M^{-1}$ its inverse;
  \item~$\Tr(M)$ its trace and~$\det(M)$ its determinant;
  \item~$\lambda_{\max}(M)$ (resp.~$\lambda_{\min}(M)$,~$\lambda_i(M)$) its maximum (resp. minimum,~$i$-th largest) eigenvalue, so that
        \begin{equation*}
          \lambda_{\max}(M) = \lambda_1(M) \geq \lambda_2(M) \cdots \geq \lambda_n(M) = \lambda_{\min}(M)
        \end{equation*}
  \item~$s_{\max}(M)$ (resp.~$s_{\min}(M)$,~$s_i(M)$) its maximum (resp. minimum,~$i$-th largest) singular value;
  \item~$\lVert M \rVert_1 = \sup \frac{\lVert M x \rVert_1}{\lVert x \rVert_1} = \max_j \sum_i |M_{i,j}|$ its operator~$\ell_1$ norm, which is the maximum~$\ell_1$ norm of a column of~$M$;
  \item~$\lVert M \rVert_2 = \sup \frac{\lVert M x \rVert_2}{\lVert x \rVert_2} = |s_{\max}(M)| = \sqrt{\lambda_{\max} (M'M)}$ its operator~$\ell_2$ norm, also known as the spectral norm;
  \item~$\lVert M \rVert_\infty = \sup \frac{\lVert M x \rVert_\infty}{\lVert x \rVert_\infty} = \max_i \sum_j |M_{i,j}|$ its operator~$\ell_\infty$ norm, which is the maximum~$\ell_1$ norm of a row of~$M$;
  \item~$\lVert M \rVert_F = \lVert \vecm(M) \rVert_2 = \Tr(M'M)$ its Frobenius norm;
  \item~$\lVert M \rVert_{\max} = \lVert \vecm(M) \rVert_{\infty} = \max_{i,j} |M_{i,j}|$ the maximum absolute value of its entries;
  \item~$\rho(M)$ its spectral radius.
\end{itemize}
See \citet{petersenMatrixCookbook2012} for a collection of inequalities relating all of these quantities.
Given two real matrices~$A$ and~$B$, we denote by
\begin{itemize}
  \item~$A \otimes B$ their Kronecker product;
  \item~$A \odot B$ Hadamard (elementwise) product;
  \item~$A \succeq B$ or~$A \preceq B$ the (partial) Loewner order on symmetric matrices.
\end{itemize}

\subsection{Frequent symbols}

Here is a list of the most frequent symbols and their meaning.

\medskip

Dimensions:
\begin{itemize}
  \item~$t \in [T]$: time step
  \item~$d \in [D]$: dimension
\end{itemize}

State process:
\begin{itemize}
  \item~$X_t$: state process
  \item~$\theta$: transition matrix
  \item~$\innov_t$: innovations
  \item~$\Sigma$: covariance matrix of~$\innov_t$
  \item~$\sigma_{\min}^2, \sigma_{\max}^2$: extremal eigenvalues of~$\Sigma$
  \item~$s$: sparsity level of~$\theta$ (number of non-zero coefficients in each row)
  \item~$\thetamax$: maximum~$\ell_2$ norm for~$\theta$
  \item~$\Theta_s$: set of feasible values for~$\theta$
  \item~$\Gamma_h(\theta)$: covariance between~$X_{t+h}$ and~$X_t$
\end{itemize}

Observations:
\begin{itemize}
  \item~$\pi_t$: random sampling vector
  \item~$\Pi_t$: diagonal random sampling matrix
  \item~$p$: fraction of state components activated by observations
  \item~$\mathcal{T}$: transition matrix for Markov sampling
  \item~$a, b$: transition probabilities for Markov sampling
  \item~$\chi$: minimum distance between~$a$ or~$b$ and~$\{0, 1\}$ (considered constant)
  \item~$Y_t$: observations
  \item~$\noise_t$: noise
  \item~$\omega^2$: variance of~$\noise_t$
\end{itemize}

Estimation:
\begin{itemize}
  \item~$h$: covariance time lag
  \item~$h_0$: minimum covariance time lag for transition estimation
  \item~$S(h)$: scaling matrix for covariance estimation
  \item~$p q_u$: smallest coefficient of the scaling matrix
\end{itemize}

Other:
\begin{itemize}
  \item $g$: standard Gaussian vector
  \item~$\Psi_{\innov}$ (resp.~$\Psi_{\noise}$): link between~$X$ (resp.~$\noise$) and a standard Gaussian vector
  \item~$L$: random bilinear form
  \item~$u$: threshold in concentration inequalities
  \item~$\delta$: small probability
  \item~$Q_{\Pi}$: constant term in the conditional variance of~$Y$
  \item~$R_{\Pi}(\theta)$: varying term in the conditional variance of~$Y$
  \item~$\Delta_\Pi(\theta)$: deviation from the identity
  \item~$\gamma_\ell$ (resp.~$\gamma_u(\theta)$): signal-to-noise ratio in the lower bound (resp. the upper bound)
\end{itemize}

\end{document}